\documentclass{siamonline1116}
\pdfoutput=1
\usepackage[T1]{fontenc}
\usepackage{graphicx}
\usepackage{wrapfig}
\usepackage{amsmath,amssymb,amsfonts}
\usepackage{amsopn}
\numberwithin{theorem}{section}
\usepackage{multirow,multirow}
\usepackage{makecell}
\usepackage[small, labelfont={bf, sf, color=header1}, textfont=it, tableposition=top]{caption}
\usepackage{xcolor}
\usepackage{url}
\usepackage{afterpage}
\usepackage{textcomp}
\usepackage{rotating}
\usepackage{siunitx}
\usepackage{paralist}
\usepackage{booktabs}
\usepackage{lscape}
\usepackage{bbold}
\usepackage{bm}
\usepackage{algorithm}
\usepackage{algpseudocode}
\floatname{algorithm}{Procedure}

\usepackage{subcaption}
\usepackage{longtable}
\usepackage[version=4]{mhchem}
\usepackage[toc,page]{appendix}
\newcounter{runidnum}

\usepackage{tikz}
\usetikzlibrary{shapes,arrows,trees,snakes}
\usepackage{pgfplots}
\usepackage{pgfplotstable}
\usepackage{tkz-euclide}
\usetikzlibrary{calc}
\usetikzlibrary{decorations.pathreplacing,arrows}
\usetikzlibrary{shapes,arrows,trees,snakes,positioning}
\usetikzlibrary{decorations.markings}

\newcolumntype{R}{>{\columncolor{gray!20}}r}
\newcolumntype{L}{>{\columncolor{gray!20}}l}
\newcolumntype{C}{>{\columncolor{gray!20}}c}
\DeclareGraphicsExtensions{.pdf,.png}
\usepackage{lineno}
\modulolinenumbers[5]
\allowdisplaybreaks

\newsiamthm{claim}{Claim}

\newcommand{\vect}[1]{\boldsymbol{#1}} 

\DeclareMathOperator*{\argmin}{arg\,min}

\newcommand{\alert}[1]{\textcolor{red}{#1}} \renewcommand{\alert}[1]{}

\definecolor{light-gray}{gray}{0.80}

\newcommand{\rbr}[1]{\left(#1\right)}

\newcommand{\opA}{\ensuremath{f}}
\newcommand{\uq}{\ensuremath{\vect{x}}}
\newcommand{\uqmap}{\ensuremath{\vect{x}_{\mathrm{MAP}}}}
\newcommand{\md}{\ensuremath{y}}
\newcommand{\DT}{\ensuremath{\mathbb{Y}}}
\newcommand{\ind}{\ensuremath{\boldsymbol{\mathrm{1}}}}

\newcommand{\var}{\ensuremath{\mathbb{V}}}

\newcommand{\mrm}[1]{\ensuremath{\mathrm{#1}}}

\graphicspath{figs}

\newcommand{\TheTitle}{BIMC: The Bayesian Inverse Monte Carlo 
    method for goal-oriented uncertainty quantification. Part I.}
\newcommand{\TheAuthors}{Siddhant Wahal and George Biros}

\headers{\TheTitle}{\TheAuthors}

\title{{\TheTitle}}

\author{
  Siddhant Wahal\thanks{
    Oden Institute for Computational Engineering and Sciences, 
    The University of Texas at Austin, Austin, TX, 78712,
(\email{siddhant@oden.utexas.edu}, \email{biros@oden.utexas.edu}). }
  \and
  George Biros\footnotemark[1]}

\ifpdf
\hypersetup{
  pdftitle={\TheTitle},
  pdfauthor={\TheAuthors}
}
\fi

\begin{document}
\maketitle

\begin{abstract}
We consider the problem of estimating rare event probabilities, 
focusing on systems whose evolution is governed by differential 
equations with uncertain input parameters. If the system dynamics is expensive 
to compute, standard sampling algorithms such as the Monte Carlo method may 
require infeasible running times to accurately evaluate these probabilities. 
We propose an importance sampling scheme (which we call \emph{``BIMC''}) that 
relies on solving an auxiliary, \emph{``fictitious''} Bayesian inverse problem. The solution of the inverse problem yields a posterior PDF, a local Gaussian approximation to which serves as the importance sampling density. We apply BIMC to several problems and  demonstrate that it can lead to computational savings of several orders of  magnitude over the Monte Carlo method. We delineate conditions under which BIMC is optimal, as well as conditions when it can fail to yield an effective IS density.
\end{abstract}



\section{Introduction}
\label{section:intro}
We consider the following goal-oriented uncertainty quantification
(UQ) problem. Let $\opA(\uq):\mathbb{R}^m\rightarrow\mathbb{R}$ be a
\emph{smooth nonlinear operator}, and $p(\uq)$ a given probability
density function (PDF) for $\uq$. Given a \emph{target interval} $\DT
\subset \mathbb{R}$, our goal is to compute $\mu = \mathbb{P}(\opA(\uq)
\in \DT)$. Equivalently, $\mu$ is the expectation under $p(\uq)$ of
the \emph{indicator function} $\ind_{\DT}(\opA( \uq))$. 
\footnote{The indicator
function, $\ind_{\DT}(y)$ assumes the value 1 if $y \in \DT$, and 0 otherwise.}
We focus on the case when $\mu \ll 1$, \emph{i.e.,} the event $f(\uq) \in \DT$
is rare.

In our context, $f(\uq)$ is a map from some random finite dimensional parameter
space to a quantity-of-interest (QoI). Such
parameter-to-QoI maps are often a composition of the solution of a
differential equation for a state variable, and an operator that 
extracts the QoI from the state. The parameters $\uq$ represent uncertain 
parameters in the physical model. This uncertainty can arise from a variety of
sources, such as lack of knowledge, measurement errors, or noise. 
Here, we assume that the uncertainty is described by a known PDF, $p(\uq)$.
A \emph{Monte Carlo} (MC) method can be used to compute $\mu$ by 
sampling $\uq$ from $p(\uq)$ and then checking whether $\opA(\uq) \in \DT$. 
But such an approach can be prohibitively expensive if the operator $\opA$ is
expensive to evaluate, especially when $\mu \ll 1$.

\paragraph{Summary of the methodology}
We propose a variance reduction scheme based on importance sampling (IS). 
In IS, samples are drawn from a new distribution, say $q(\uq)$, in order to
increase the occurrences of the rare event. We construct our IS density as
follows. We begin by setting up an auxiliary \emph{inverse problem}. 
First we select a $\md \in \DT$, and then we find $\uq$
such that $\opA(\uq)\approx \md$. This is an ill-posed or inverse problem since given a scalar $\md$ we want to reconstruct the vector $\uq$. A simple counting argument shows  that this is impossible unless we use some kind of  regularization. To address this ill-posedness we adopt a Bayesian perspective, that is, the solution of the inverse problem is not a specific point estimate $\uq$ but  a ``posterior distribution'', 
$p(\uq | \md)$, a PDF on the parameters $\uq$ conditioned on $\md$. We will use a Gaussian
approximation of this posterior around the \emph{Maximum A Posteriori} (MAP) point as the importance sampling distribution. The mean of the approximating Gaussian is the MAP point itself, and its covariance is the inverse of the Gauss-Newton Hessian, $\mathbf{H}_{\mathrm{GN}}^{-1}$, of $-\log(p(\uq | y))$ at the MAP point. 

\paragraph{Contributions}
In summary, our contributions are the following. 
\begin{itemize}
    \item We introduce the concept of solving \emph{inverse} problems for
    \emph{forward} uncertainty quantification.
  \item To our knowledge, this is the first algorithm that exploits derivatives
  of the forward operator $f$ to arrive at an IS density for simulating rare
  events. 
  \item We offer a thorough theoretical analysis of the affine-Gaussian inverse
      problem. This analysis establishes conditions for optimality of our
      algorithm, as well as guides the tuning of various algorithmic
      ``knobs''.
  \item We apply our methodology to several real and synthetic
    problems and demonstrate orders-of-magnitude speedup over a vanilla MC
    implementation.
\end{itemize}
\paragraph{Limitations}
\begin{itemize}
    \item The success of our algorithms depends strongly upon the quality
      (both in terms of accuracy and speed) of the inverse problem
      solution.  When the operator $\opA$ involves differential
      equations, efficiently solving the inverse problem requires  adjoint operators
      and perhaps sophisticated PDE-constrained optimization solvers and
      preconditioners.
    \item Our methodology has several failure mechanisms. These are described in
        detail in \Cref{section:experiments}. In light of these
        failure mechanisms, the question of \emph{a priori} assessing
        the applicability of BIMC to a given problem (\emph{i.e.}, a given
        combination of $f(\uq)$, $p(\uq)$, and $\DT$) has also been left
        unexplored.
\end{itemize}

\paragraph{Related work} The literature on goal-oriented techniques,
importance sampling, rare-event probability estimation, and Bayesian inference is quite
extensive. Here, we review work that is most relevant.

\subparagraph{Goal-oriented methods}
The idea of goal-oriented techniques for UQ isn't new 
(see~\cite{lieberman-willcox13, lieberman2014nonlinear, spantini2017goal}).
However, most of these works focus on dimensionality reduction, and not rare events.
Also pertinent is the measure-theoretic approach to inverse
problems~\cite{breidt-butler-estep11, butler-estep-sandelin12}.

\subparagraph{Rare-event probability estimation}
A large body of work on rare events has been
motivated by the problem of assessing the reliability of systems.
In such problems, the task is to compute the probability of failure of a
system, which occurs when $f(\uq) < 0$ (or in our framework, when $\DT =
(-\infty, 0)$). 

Analytical approaches to approximate this failure probability
include the First and Second Order Reliability Methods
(see~\cite{rackwitz2001reliability} for a review). These methods are based
on approximating $f$ with a truncated Taylor series expansion
around a ``design'' point. A drawback of these methods is that
they have no means of estimating the error in the computed failure probability.
We would like to note that the concept of
``design'' points here is similar to the MAP point in our algorithm, but they
are not exactly identical. The
design point, say $\uq^*$, is always constrained to satisfy $f(\uq^*) = 0$. That
is, it lies at the edge of the pre-image $f^{-1}(\DT)$. The MAP point, 
on the other hand, is expected to lie in the interior of the region
$f^{-1}(\DT)$. Moreover, the fact that $\uq_{\mathrm{MAP}}$ lies in the interior
of $f^{-1}(\DT)$ is accounted for, and in fact, exploited, when we choose
tunable parameters of our algorithm.

Statistical approaches to evaluate the failure probability have received
considerable attention (see~\cite{rubino2009rare} for a review). 
As opposed to analytical methods, these methods have
well-understood convergence properties, and they come with a natural error
estimate. In this context, a simple Monte Carlo method is usually inefficient,
and some form of variance reduction is usually required. Several  importance sampling methods have been  proposed to this effect.  We refer the reader to~\cite{mcbook} for a general introduction to importance
sampling. Several IS algorithms 
(\cite{schueller1987critical, bucher1988adaptive, melchers1989importance, 
au1999reliability}) reuse the concept of design points by placing normal
distributions centered there. In \cite{bucher1988adaptive,
schueller1987critical}, the covariance of the 
IS distribution is either set equal to that of
$p(\uq)$, or evaluated heuristically, for example, from samples. In
our method, approximating the posterior via a Gaussian yields a natural
covariance for the IS density. 

Within reliability analysis, another class of algorithms uses surrogate models
to reduce the computational effort required to build an IS 
density~\cite{peherstorfer2016, peherstorfer2017multifidelity, kramer2017multifidelity}. 
A different approach involves simulating a sequence of relatively higher
frequency events to arrive at the rare event probability. 
This idea is used in the Cross Entropy
algorithm~\cite{de2005tutorial} to arrive at an optimal IS distribution within a
parametric family. It has also been coupled with Markov Chain Monte Carlo
methods for high-dimensional reliability 
problems~\cite{au1999new, katafygiotis2007estimation, au2001estimation}.

A common feature of all these algorithms is that they only use pointwise
evaluations of the forward model $f$ (or its low-fidelity surrogates) to arrive
at the IS distribution. Hence, these methods are ``non-intrusive''. On the other
hand, the manner in which we construct
our IS density naturally endows it with information from derivatives of $f$.
To our knowledge, the only other algorithms that utilize derivative information
to construct IS densities are IMIS and 
LIMIS~\cite{raftery2010estimating, fasiolo2018langevin}. However, these aren't
tailored for rare-event simulation. Directly substituting the zero-variance
(zero-error) IS density (see \Cref{subsection:imp_sampling}) 
for the target distribution in these algorithms wouldn't
work, since the zero-variance density is non-differentiable, owing to the
presence of the characteristic function.

\subparagraph{PDE-constrained optimization and Bayesian inverse problems}
In BIMC, we rely on adjoints to compute gradients and Hessians of $-\log p(\uq |
y)$. We refer to~\cite{gunzburger2002perspectives} for an introduction to the 
method of adjoints. Computing the MAP point is a PDE-constrained optimization 
problem which can require sophisticated algorithms~\cite{akccelik2006parallel}.
Scalable algorithms for characterizing the Hessian of $-\log(p(\uq | y))$ are
described in~\cite{isaac2015scalable}.
Because we construct our IS density through the solution of an inverse
problem, our approach can be easily built on top of existing scalable
frameworks for solving Bayesian inverse problems, such
as~\cite{villa2018hippylib, villa2019hippylib}.

\paragraph{Outline of the paper}
The rest of this paper is organized as follows - 
\Cref{table:notation} introduces the
notation adopted in this paper. \Cref{section:background} provides 
introductions to the Monte Carlo method, importance sampling, as well as 
Bayesian inference. In \Cref{section:methodology}, we describe our algorithm. 
including analysis that governs the choice of tunable parameters
that arise in the algorithm. \Cref{section:experiments} contains numerical experiments and their
results, as well as a description of the failure mechanisms of our method. 
Finally, we summarize our conclusions in \Cref{section:conclusion}.

\label{section:notation}
\begin{table}[h]
\small
\centering
\begin{tabular}[c]{l l}
\toprule
Symbol             & Meaning\\
\midrule
$f$                & The input-output, or the forward, map\\
$\uq$              & Vector of input parameters to $f$\\
$p(\uq)$           & Input probability density for $\uq$\\
$\DT$              & Target interval for $f(\uq)$\\
$\mathbb{P}(f(\uq) \in \DT)$
                   & Probability of the event $f(\uq) \in \DT$\\
$\mu$              & $\mathbb{P}\rbr{f(\uq) \in \DT}$ \\
$\mathcal{N} \rbr{\uq_0, \bm{\Sigma}_0}$ 
                   & Normal distribution with mean $\uq_0$
                     and covariance $\bm{\Sigma}_0$\\
$N$                & Number of Monte Carlo (MC) or Importance Sampling 
                     (IS) samples\\
$\hat{\mu}^N$        & MC estimate for $\mu$ computed using $N$ samples\\
$\tilde{\mu}^N$      & IS estimate for $\mu$ computed using $N$ samples\\
$\hat{e}_{\mathrm{RMS}}$ 
                   & Root Mean Square (RMS) error in $\hat{\mu}^N$\\
$\tilde{e}_{\mathrm{RMS}}$ 
                   & RMS error in $\tilde{\mu}^N$\\
$p(y | \uq)$       & The likelihood density\\
$p(\uq | y)$       & The posterior density\\
$\uqmap$           & The Maximum \emph{A Posteriori} (MAP) point of $p(\uq|y)$\\
$\mathbf{H}_{\mathrm{GN}}$ 
                   & The Gauss-Newton Hessian of $-\log p(\uq | y)$\\
$D_{\mathrm{KL}}(p||q)$ 
                   & The Kullback-Leibler divergence between 
                     densities $p$ and $q$\\
\bottomrule
\end{tabular}
\caption{Summary of key notation used in this paper.}
\label{table:notation}
\end{table}

\section{Background}
\label{section:background}
\subsection{The Monte Carlo Method}
One way to compute the rare-event probability, $\mu$, is using the Monte Carlo 
method. The forward operator is applied on $N$ independent, identically 
distributed (i.i.d.) samples from $p(\uq)$, ${\{\uq_i\}}_{i = 1}^{N}$. Then, 
an unbiased estimate of $\mu$ is:

\begin{align}
    \hat{\mu}^N = \frac{\sum_{i = 1}^{N} \ind_{\DT}(\opA \left(\uq_i\right))}{N}.
    \label{mcProb}
\end{align}

The law of large numbers guarantees that in the limit $N \rightarrow \infty$, 
$\hat{\mu}$ converges to $\mu$~\cite{Robert2004Monte}. The relative Root Mean 
Square Error (RMSE) in $\hat{\mu}$ is: 

\begin{align*}
    \hat{e}_{\mathrm{RMS}}
            = \frac{1}{\mu}\sqrt{\mathbb{E}_p\Big((\hat{\mu}^N - \mu)^2\Big)}
            = \frac{1}{\mu}
                    \sqrt{\frac{\var_p\Big(\ind_{\DT}(\opA (\uq))\Big)}{N}}
            = \sqrt{\frac{\sigma_p^2}{\mu^2 N}}.
\end{align*}

Since $\ind_{\DT}(\opA(\uq))$ is a binary random variable, its variance is 
$\sigma_p^2 = \mu(1 - \mu)$. This implies that the relative RMSE is 
approximately $\sqrt{1/\mu N}$ when $\mu \ll 1$. In order to achieve a 
specified relative accuracy threshold, the number of samples must then scale as 
$N \sim 1/\mu$. This is problematic since it can render evaluating extremely 
rare probabilities virtually impossible if $f(\uq)$ is expensive to evaluate. 
The evaluation of rare probabilities
can be made tractable by reducing the variance of the MC estimate. In BIMC, 
we aim to achieve variance reduction through importance sampling, which is 
briefly introduced in the next section.

\subsection{Importance Sampling}
\label{subsection:imp_sampling}
Importance sampling biases samples towards regions which trigger the rare event 
(or in our context, where  $f(\uq) \in \DT$)  with the help of a new probability 
density $q$. The contribution from each sample, however, must be weighed to 
account for the fact that one is no longer sampling from the original 
distribution $p$. Thus,

\begin{align}
    \mu &= \int_{\mathbb{R}^m}\ind_{\DT}\left(\opA(\uq)\right)p(\uq)\mrm{d}\uq
         = \int_{\mathbb{R}^m} \ind_{\DT}
           \left(\opA (\uq)\right)\frac{p(\uq)}{q(\uq)}q(\uq) \mrm{d} \uq
         = \mathbb{E}_q\left(\ind_{\DT}\left(f(\uq)\right)
           \frac{p(\uq)}{q(\uq)}\right).
\end{align}

Then, $q$ is called the importance distribution and $p(\uq)/q(\uq)$ 
is the likelihood ratio. The importance sampling estimate for $\mu$ is:

\begin{align}
    \tilde{\mu}^N = \frac{1}{N} \sum_{i = 1}^{N} 
                  \frac{\ind_{\DT}\big(\opA (\uq_i)\big) 
                  p(\uq_i)}{q(\uq_i)},\quad \uq_i \sim q(\uq).
    \label{impSamplingProb}
\end{align}

The relative RMSE in estimating $\mu$ using importance sampling is:

\begin{align}
\begin{split}
    \tilde{e}_{\mathrm{RMS}} &= \frac{1}{\mu}\sqrt{\mathbb{E}_q
                          \big((\tilde{\mu}^N - \mu)^2\big)}
              = \sqrt{\frac{\sigma_q^2}{\mu^2 N}},\,\text{where,}\\
  \sigma_q^2 &= \var_q\bigg(\ind_{\DT}\big(\opA (\uq)\big)
                    \frac{p(\uq)}{q(\uq)}\bigg).
\label{isRMSE}
\end{split}
\end{align}

If $\sigma_q^2$ is smaller than $\sigma_p^2$, the importance sampling estimate
of $\mu$ is more accurate than the one obtained using simple MC. The main 
challenge in importance sampling is selecting an importance density $q$ such 
that $\sigma_q < \sigma_p$. The IS density that 
minimizes $\sigma_q$ is known to be $q^* = \ind_{\DT}(\opA (\uq)) p(\uq) / \mu$ 
(see \cite{kahn1953methods}). That is, the optimal density for importance
sampling is just $p(\uq)$ truncated over regions where $f(\uq) \in \DT$, and 
then appropriately renormalized. However, $q^*$ cannot be 
sampled from, since the renormalization constant $\mu$ is exactly the 
probability we set out to compute in the first place. Nevertheless, it defines
characteristics desirable of a good importance density - it must have most
of its mass concentrated over regions where  $f(\uq) \in \DT$ and resemble
$p(\uq)$ in those regions. 

So the first step in constructing an effective IS density is identifying regions
where $f(\uq) \in \DT$. As mentioned in \Cref{section:intro}, this is done by
solving a Bayesian inverse problem. Before describing
the BIMC methodology in detail, we first provide a brief introduction to
Bayesian inference in a generalized setting.

\subsection{Bayesian inference}

In a general setting where inference must be performed, the problem is slightly
different. Here the goal is to infer input parameters $\uq$ from a (possibly
noisy) real-world observation of the output, say $y$. In the Bayesian approach,
this problem is solved in the statistical sense. The
solution of a Bayesian inference problem is a probability density over the space
of parameters that takes into account any prior knowledge about the parameters 
as well as uncertainties in measurement and/or modeling. This probability 
density, known as the posterior, expresses how likely it is for a particular 
estimate to be the true parameter corresponding to the observation. 

In addition to the observation $y$, assume the following quantities have been 
specified - i) a suitable probability density $p(\uq)$ that 
captures prior knowledge about the parameters $\uq$, and ii)
the conditional probability density of observing the data $y$ given the
parameters $\uq$, $p(y | \uq)$. Then, from Bayes' theorem, the posterior is
given by:

\begin{align}
    p(\uq | y) \propto p(y | \uq) p(\uq).
\end{align}

The posterior can also be interpreted to be updated beliefs once the data and
errors have been assimilated. We would like to emphasize here that in an 
actual inverse  problem, the observation $y$, as well as the likelihood 
density $p(y | \uq)$ are physically meaningful. The former corresponds to 
real-world measurements of the output of the forward model. The latter 
describes a model for errors arising out due to modeling inadequacy or 
measurement. 

The posterior by itself is of little use. Often, the task is to evaluate
integrals involving the posterior. This might be the case, for example, when
trying to characterize uncertainty in the inferred parameters by 
evaluating moments (mean, covariance) of the posterior. Analytical evaluation 
of these integrals is often out of the question and a sample based estimate 
must be used. Except in certain cases, the posterior is
an arbitrary PDF in $\mathbb{R}^m$ and generating samples from it 
requires sophisticated methods such as Markov Chain Monte Carlo. For easy 
sample generation, the posterior can be locally approximated by a Gaussian
around its mode (also known as the Maximum \emph{A Posteriori} point). 
By linearizing $f$ around the MAP point, it can be
shown that the mean of the approximating Gaussian is the MAP point, and its
covariance is the inverse of the Gauss-Newton Hessian matrix of 
$-\log p(\uq | y)$ at the MAP point~\cite{isaac2015scalable}.

As a concrete example, consider the case when the likelihood density represents Gaussian 
additive error of magnitude $\sigma$, $p(y | \uq)  = \mathcal{N}(f(\uq), 
\sigma^2)$. Then, $p(\uq | y) \propto \exp \left(-\frac{\left(y -
f(\uq)\right)^2}{2\sigma^2}\right)p(\uq)$, and we have (up to an additive constant),

\begin{align}
    -\log p(\uq | y) =  
            \frac{1}{2\sigma^2}\left(y - f(\uq)\right)^2 - \log p(\uq),
\label{negLogPost}
\end{align}
 
and, $\uqmap$ can be found as:

\begin{align}
    \begin{split}
        \uqmap &= \argmin_{\uq \in \mathbb{R}^m} \frac{1}{2\sigma^2}
                   \left(y - f(\uq)\right)^2 - \log p(\uq).
        \label{mapEqn}
    \end{split}
\end{align}

Then, the Gauss-Newton Hessian matrix of $-\log(p(\uq | y))$ 
can be written as 

\begin{align}
    \begin{split}
        \mathbf{H}_{\mathrm{GN}} &= - \nabla_{\uq}^2 \log p(\uq | y)\\
                   &= \frac{1}{\sigma^2}(\nabla_{\uq}f) (\nabla_{\uq}f)^T 
                      - \nabla_{\uq}^2 \log p(\uq).
    \end{split}
\end{align}

Note that, the Gauss-Newton Hessian has the attractive 
property of being positive-definite.
These expressions show that $\uqmap$ can be interpreted as that point in
parameter space that minimizes mismatch with the observation but is also 
highly likely under the prior. %
So sampling from a Gaussian approximation of the posterior can be thought of as
drawing samples in the vicinity of a point that is consistent with the data as
well as the prior. In addition, the covariance or spread of the samples is
informed by the derivatives of the forward model.
While constructing the IS density in BIMC, this feature of the Gaussian 
approximation of the posterior in a general, real-world setting will be used in 
conjunction with the knowledge of the shape of the ideal IS density. This
completes the presentation of the necessary theoretical background and we are
ready to describe the BIMC methodology.

\section{Methodology}
\label{section:methodology}

Recall that the forward UQ problem is to compute $\mathbb{P}\left(f(\uq) \in
\DT\right)$ when $\uq \sim p(\uq)$. In BIMC, we use the ingredients of the
forward UQ problem to construct a \emph{fictitious} Bayesian inverse problem as 
follows. We
\begin{enumerate}
    \item select some $y \in \DT$ as a surrogate for real-world observation,
    \item use $p(\uq)$ as the prior, and, 
    \item concoct a likelihood density $p(\uq | y)$.
\end{enumerate}

This enables us to define a pseudo-posterior $p(\uq | y)$, and subsequently, a
Gaussian approximation to it. We call this inverse problem fictitious
because both the observation $y$ and the likelihood density $p(\uq | y)$ are arbitrarily chosen by us. Neither is $y$ a real-world measurement of a physical quantity, nor does  $p(\uq | y)$ correspond to an actual error model. From here on, we will refer to these artificial quantities as the pseudo-data and the pseudo-likelihood respectively.

We propose using the Gaussian approximation to the posterior as an IS 
density. As outlined in the previous section, in the real-world setting, the 
mean of the Gaussian approximation of the posterior (the MAP point) is that 
point in parameter space that is consistent with the data as well as the prior. 
So by solving the fictitious Bayesian inverse problem defined earlier, we expect
the mean of the IS density to be a point that is consistent with some 
$y \in \DT$ as well as the nominal PDF $p(\uq)$. This ensures the IS density is
centered around regions where $f(\uq) \in \DT$. Further, the covariance matrix 
of the Gaussian approximation, and hence the IS density,
contains first-order derivative information. This approach is illustrated in
\Cref{fig:summary}.

Since a Gaussian likelihood model has been assumed, the pseudo-posterior is
proportional to 
$\exp\left(-{\left(y - f(\uq)\right)^2}/{2\sigma^2}\right) p(\uq)$. Thus,
an alternative interpretation of the pseudo-posterior in this case is as a
``mollified'' approximation of the ideal IS density, 
$\ind_{\DT}\left(f(\uq)\right) p(\uq)$,
where the mollification has been achieved by smudging the sharply defined
characteristic function $\ind_{\DT}(f(\uq))$ into a Gaussian, 
$\exp\left(-{\left(y - f(\uq)\right)^2}/{2\sigma^2}\right)$
The advantage of doing this lies in the fact that the mollified ideal IS density has
well-defined derivatives and can be explored via derivative-aware methods,
unlike the true ideal IS density, which isn't differentiable. Algorithms like
IMIS~\cite{raftery2010estimating},  and LIMIS~\cite{fasiolo2018langevin} 
can now be employed for rare-event probability estimation by
plugging in the pseudo-posterior as the target.

\begin{figure}[htbp]
        \centering
        \begin{subfigure}[b]{0.28\textwidth}
            \includegraphics[width=\textwidth,bb=0 0 266 347]{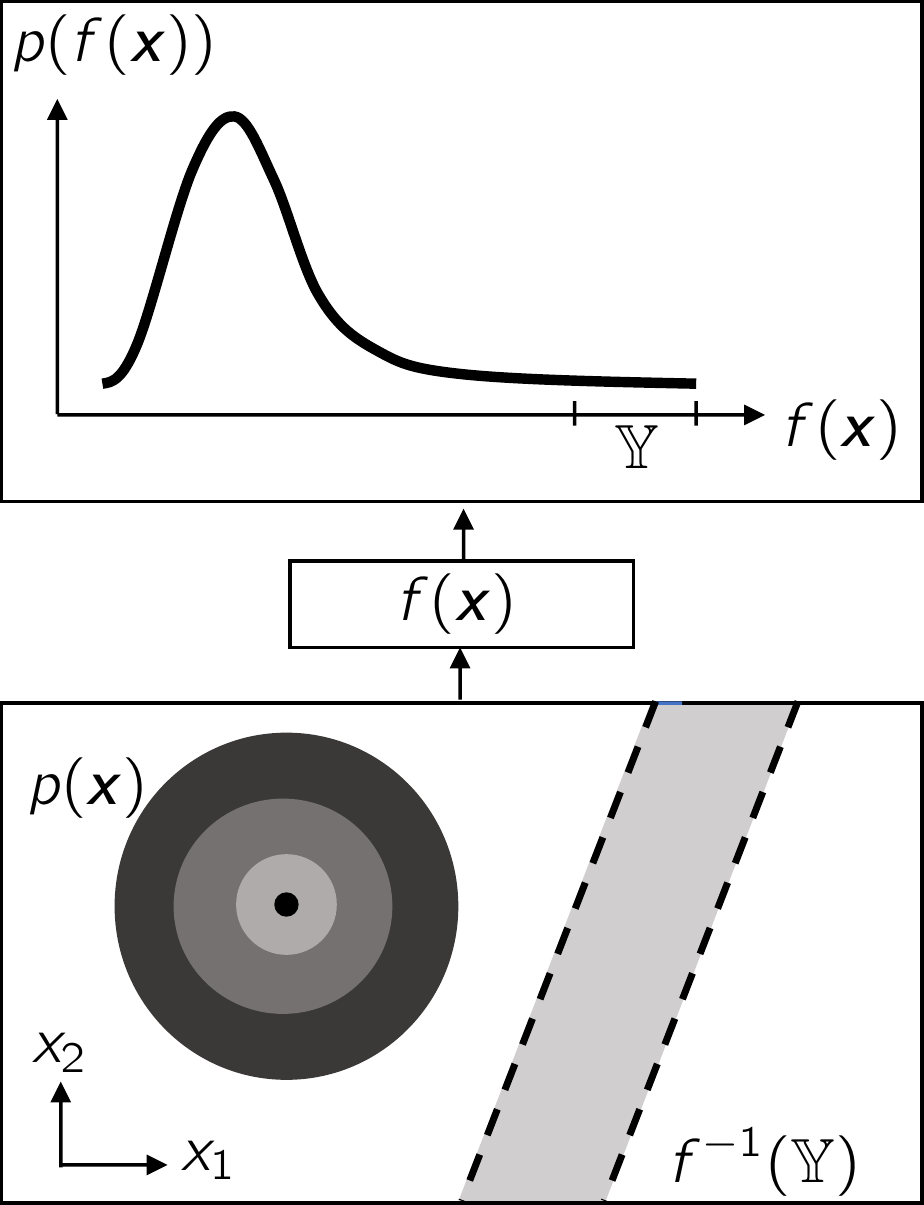}
            \caption{}
            \label{fig:invFig1}
        \end{subfigure}
        ~
        \begin{subfigure}[b]{0.28\textwidth}
            \includegraphics[width=\textwidth,bb=0 0 266 347]{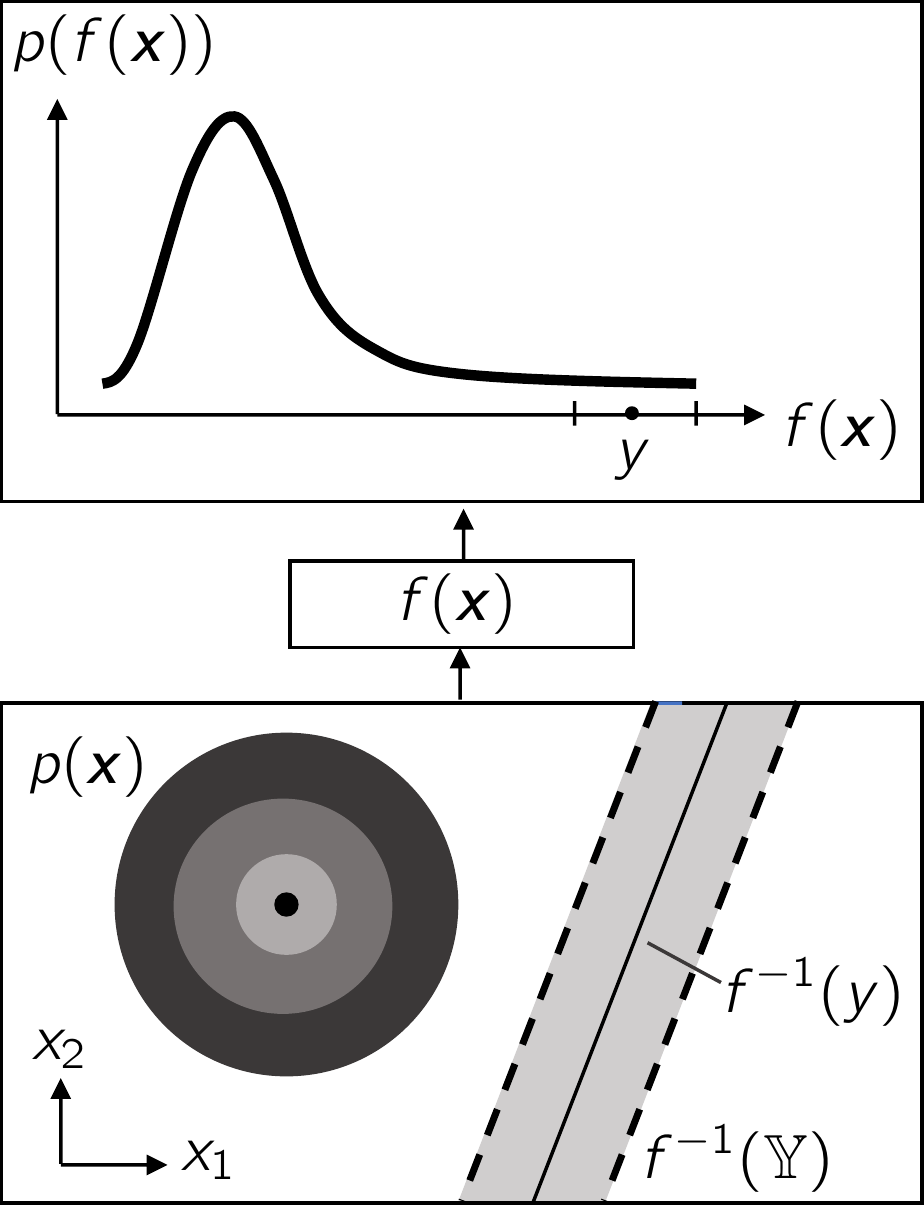}
            \caption{}
            \label{fig:invFig2}
        \end{subfigure}\\
        \begin{subfigure}[b]{0.28\textwidth}
            \includegraphics[width=\textwidth,bb=0 0 266 347]{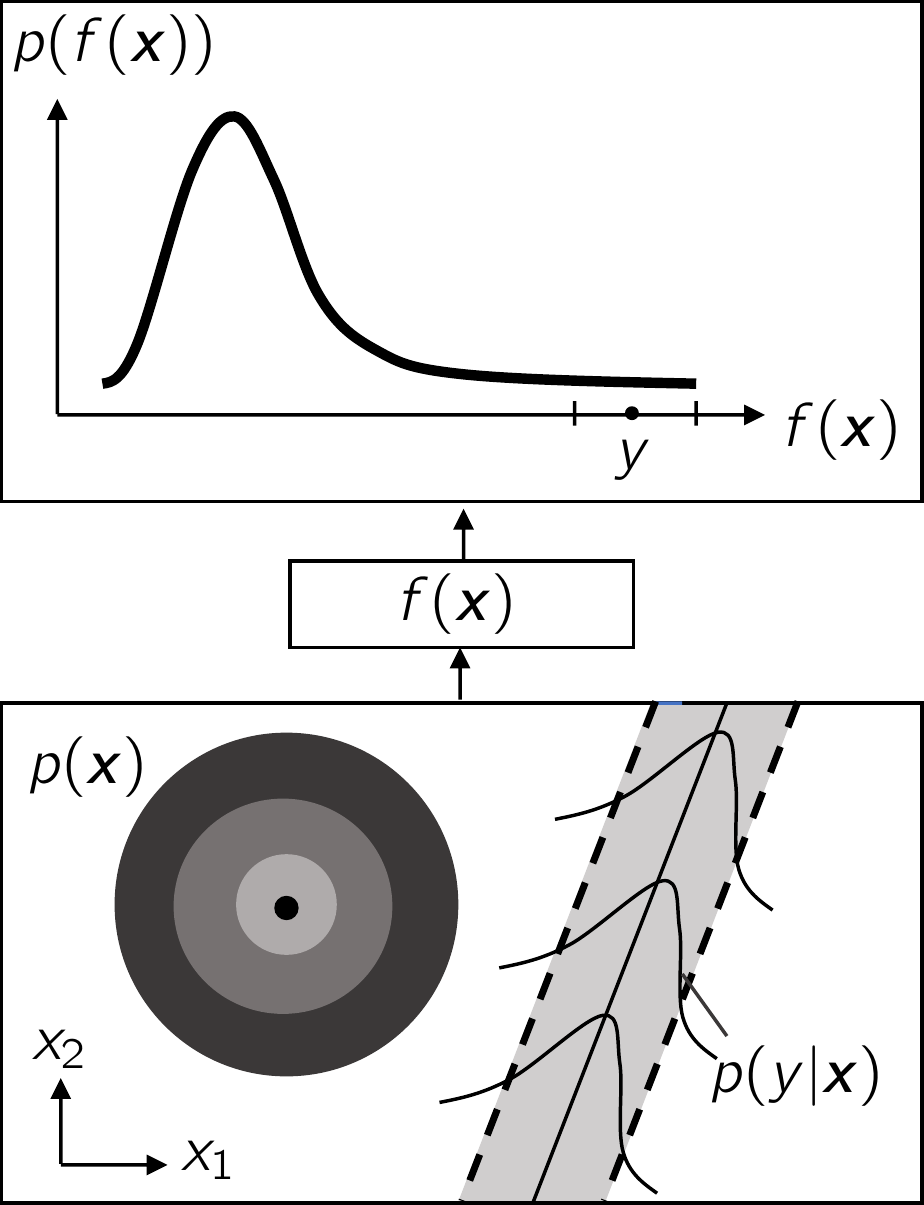}
            \caption{}
            \label{fig:invFig3}
        \end{subfigure}
        ~
        \begin{subfigure}[b]{0.28\textwidth}
            \includegraphics[width=\textwidth,bb=0 0 266 347]{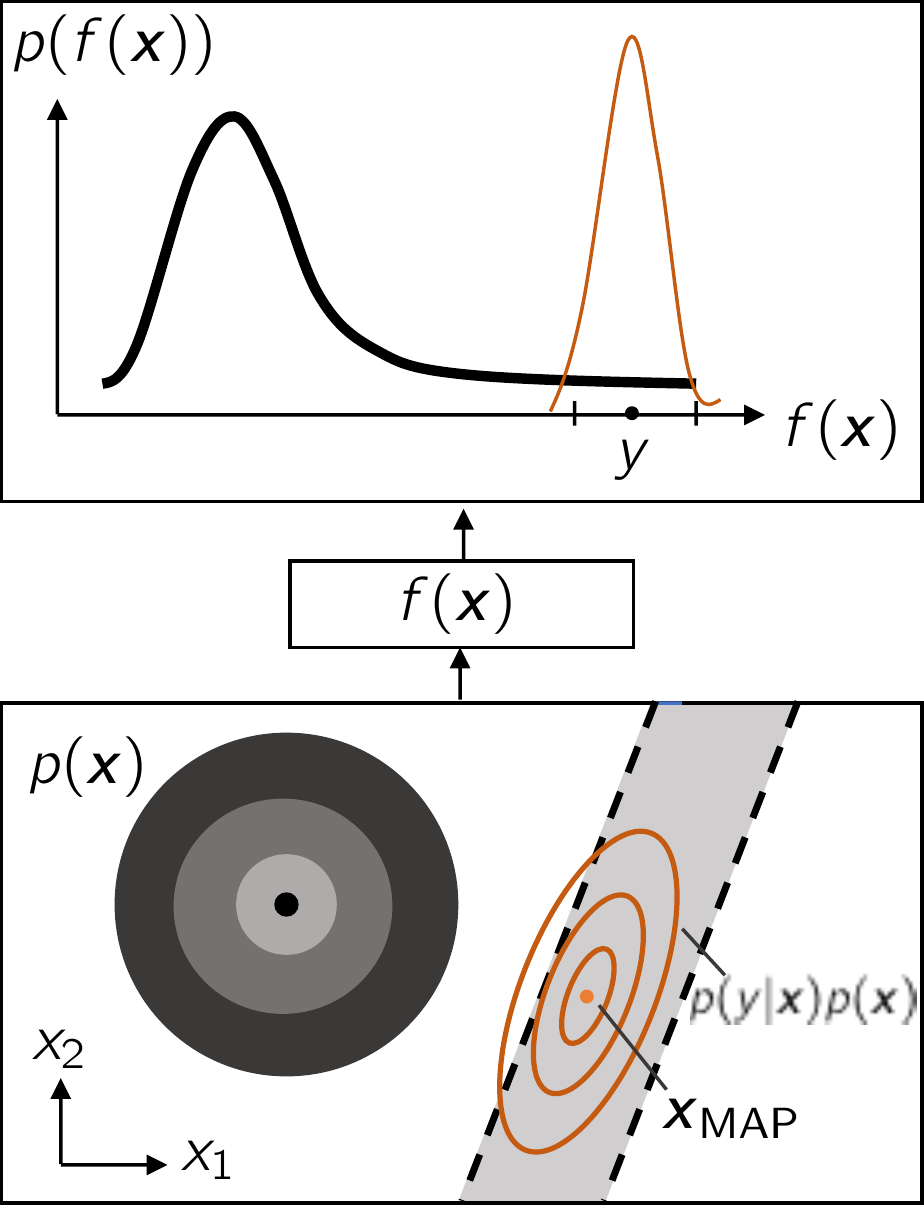}
            \caption{}
            \label{fig:invFig4}
        \end{subfigure}
        \caption{Summary of the BIMC methodology. In (\subref{fig:invFig1}), the
            problem statement is summarized - we need to compute the
            probability of the pre-image of the target interval, $f^{-1}(\DT)$.
            In (\subref{fig:invFig2}), we introduce the pseudo-data point $y \in \DT$ and
            redefine $p(\uq)$ to be the pseudo-prior in a fictitious Bayesian inverse
            problem. The true inverse of the pseudo-data point $y$ is a straight line in
            $\mathbb{R}^2$. Next, in (\subref{fig:invFig3}), we select a pseudo-likelihood 
            density, $p(y | \uq)$. The pseudo-likelihood density can also be viewed
            as a mollified approximation of the characteristic function. 
            Part (\subref{fig:invFig4}) shows the contours of the
            posterior, $p(\uq | y)$, which is proportional to $p(y | \uq)p(\uq)$. We use
            a Gaussian approximation to this posterior as an IS density.}
        \label{fig:summary}
\end{figure}

Irrespective of the interpretation, this methodology introduces 
two tunable parameters---the pseudo-data $y \in \DT$
and variance of the pseudo-likelihood density, $\sigma^2$. These parameters can have a
profound effect on the accuracy of the importance sampler and must be tuned with
care. The tuning strategy depends on the nature of $f(\uq)$ as well as $p(\uq)$.
Next, we discuss possible cases as well as the corresponding tuning strategy. 

\subsection{Affine $f$, Gaussian $p$}
Although for a affine $f(\uq)$ and Gaussian $p(\uq)$, the probability $\mu =
\mathbb{P}(f(\uq) \in \DT)$ may be analytically computed, the availability of
analytical expressions for $\uqmap$ and $\mathbf{H}_{\mathrm{GN}}$ in this case illustrates
how our importance sampler achieves variance reduction. 
Let $\opA(\uq) = \vect{v}^T\uq +\beta$, $p(\uq) = \mathcal{N}(\uq_0, 
\mathbf{\Sigma}_0)$ for some $\vect{v}$,
$\uq_0 \in \mathbb{R}^m$, $\beta \in \mathbb{R}$ and $\mathbf{\Sigma}_0 \in
\mathbb{R}^{m\times m}$. Then, $\mu$ is given analytically as

\begin{align}
    \mu = \Phi\left(\frac{y_{\max} - \nu}{\gamma}\right) - 
    \Phi\left(\frac{y_{\min} - \nu}{\gamma}\right),
    \label{linProb}
\end{align}

where $\nu = \vect{v}^T\uq_0 + \beta$, $\gamma^2 =
\vect{v}^T\mathbf{\Sigma}_0\vect{v}$, and $\Phi$ is the standard Normal CDF.

Now, suppose the pseudo-data is some $y \in \DT$
and the variance of the pseudo-likelihood is some $\sigma^2 \in \mathbb{R}$.
Then the pseudo-posterior $p(\uq | y)$ is also a
Gaussian and no approximations are necessary. Hence, the IS density is given 
by
$q(\uq) = \mathcal{N}(\uqmap, \mathbf{H}_{\mathrm{GN}}^{-1})$, where,

\begin{align}
    \uqmap = \uq_0 + \frac{y - f(\uq_0)}{\sigma^2 + \vect{v}^T
    \mathbf{\Sigma}_0\vect{v}} \mathbf{\Sigma}_0 \vect{v},\quad
    \mathbf{H}_{\mathrm{GN}}^{-1} = \mathbf{\Sigma}_0 - \frac{1}{\sigma^2 + 
    \vect{v}^T \mathbf{\Sigma}_0 \vect{v}}\big(\mathbf{\Sigma}_0\vect{v}\big)
    {\big(\mathbf{\Sigma}_0 \vect{v}\big)}^{T}.
\end{align}
 
The expressions for $\uqmap$ and $\mathbf{H}_{\mathrm{GN}}^{-1}$ expose  how our 
importance sampler achieves variance reduction. The MAP point identifies the 
region in parameter space where $f(\uq) \approx y$. The spread of the 
importance sampler, as encapsulated in its covariance, is reduced over its 
nominal value $\bm{\Sigma}_0$ in a direction informed by $\vect{v}$, the 
gradient of $f(\uq)$. Note that the reduction in variance occurs in just one
direction, $\bm{\Sigma}_0\vect{v}$; the variance of $p(\uq)$ is retained in all other directions. 
The parameter 
$\sigma^2$ controls how much $q(\uq)$ is updated over $p(\uq)$- a small value 
for $\sigma$ results in a larger shift from $\uq_0$ and a larger reduction in 
its spread. These claims become more transparent by noticing that the 
pushforward of $q(\uq)$ under $f$ is another Gaussian 
distribution in $\mathbb{R}$ (the pushforward density represents how $f(\uq)$ 
will be distributed if $\uq$ is distributed according to $q(\uq)$). 
The mean $\chi$ and variance $\xi^2$ of this pushforward density are- 

\begin{align}
    \chi = (1 - \rho^2) y + \rho^2 f(\uq_0),\,
    \xi^2 = \rho^2 \vect{v}^T\bm{\Sigma}_0\vect{v},\,\text{where,}\,
    \rho^2 = \frac{\sigma^2}{\sigma^2 + \vect{v}^T\bm{\Sigma}_0\vect{v}} < 1 .
\end{align}

A small $\sigma$ implies small $\rho$, which means $\chi$ is closer $y$ and
$\xi^2$ is small. 

Since our goal is importance sampling, we wish to select those values for the
tunable parameters that deliver just the right amount of update over $p$. We do
this by minimizing the
Kullback-Leibler (KL) divergence between $q(\uq)$ and the 
ideal IS distribution $q^*(\uq)$. Although not a true metric, the KL divergence 
between two probability densities is a measure of the distance between them.
It is defined as:

\begin{align}
    D_{\mathrm{KL}}(p || q) = \int p(\uq) \log \frac{p(\uq)}{q(\uq)}
    \mathrm{d}\uq .
\end{align}

Then, the optimal pseudo-data point and the optimal variance of the
pseudo-likelihood density can be obtained as:

\begin{align}
    \label{param_selection}
    \begin{pmatrix}
        \sigma^*\\
        y^*
    \end{pmatrix}    &= \argmin_{\sigma, y} D_{\mathrm{KL}}(q^* || q;
    \sigma, y) .
\end{align}

Analytic expressions for $D_{\mathrm{KL}}$, $y^*$, and $\sigma^*$ are derived
for the affine Gaussian case in the supplement in \Cref{section:supplementAffineGauss}. 
Selecting the tunable parameters in this way, in fact, allows us to make the
following claim regarding the resulting IS distribution:

\begin{claim}[BIMC optimality]\label{claim:bimc_optimality}
    In the affine-Gaussian case, the importance sampling density 
    that results from the BIMC procedure is
    equivalent to the Gaussian distribution closest in KL divergence to
    $q^*(\uq)$.
\end{claim}
\begin{proof}
    Proof given in the supplement in \Cref{section:supplementAffineGauss}:
\end{proof}

Hence, BIMC is implicitly searching for the best Gaussian approximation of
$q^*(\uq)$. A Gaussian distribution in $m$ dimensions has ${m(m + 1)}/{2}$
free variables, so a naive search for the best Gaussian approximation of
$q^*(\uq)$ will optimize over all $\mathcal{O}(m^2)$ free variables. However,
BIMC accomplishes this task be optimizing just 2 free variables. This can be
attributed to the similar structure of the pseudo-posterior ${p(y |
\uq)p(\uq)}/{p(y)}$, and the ideal IS density,
${\ind_{\DT}(f(\uq))p(\uq)}/{\mu}$.

\begin{figure}[htbp]
    \centering
    \includegraphics[width=0.7\textwidth]{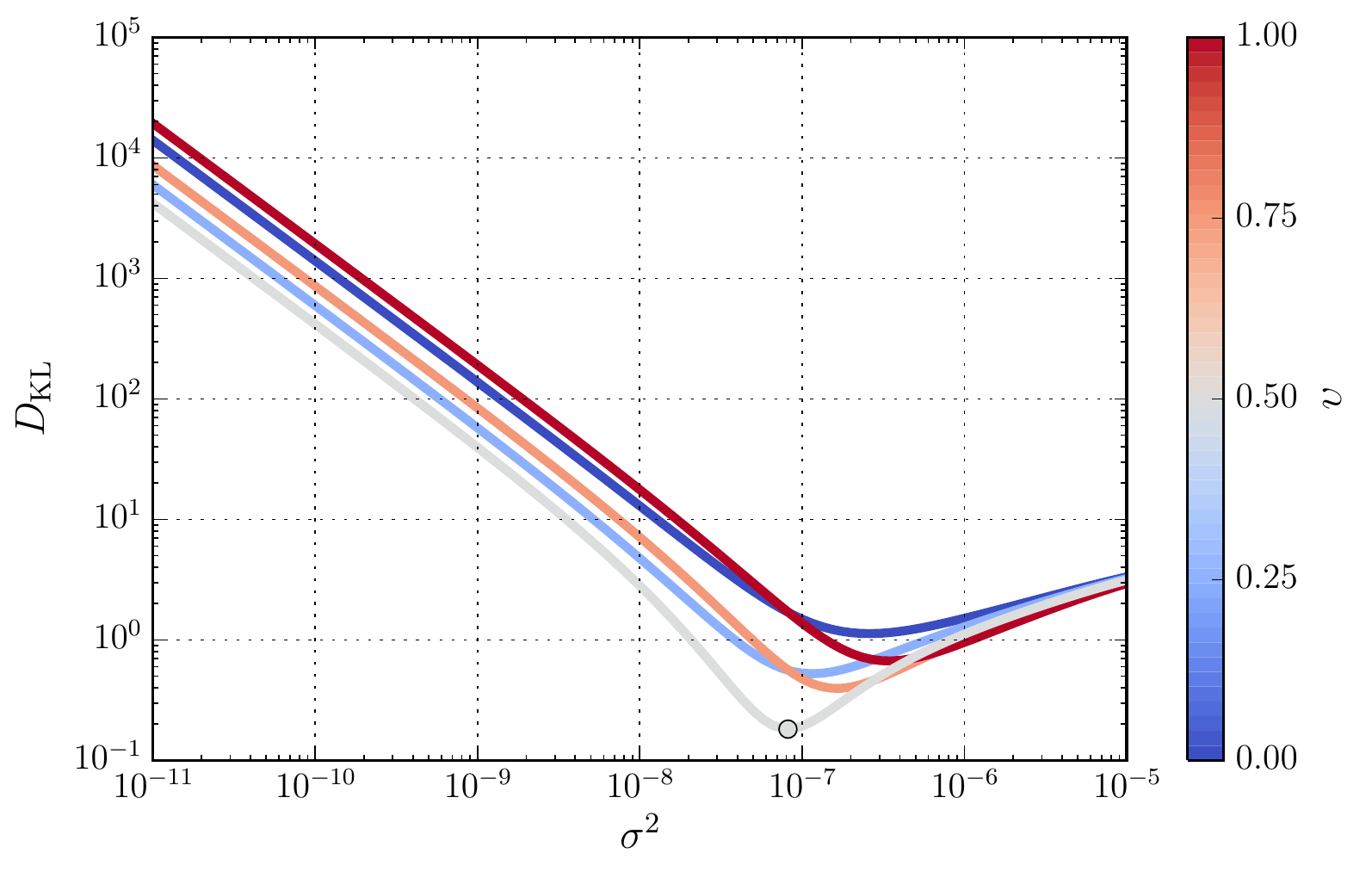} 
    \caption{This figure shows the variation of $D_{\mathrm{KL}}$ at various
    values of the data point $y$ and pseudo-likelihood variance $\sigma^2$. 
    Here, the data point $y$ has been
    normalized to $\upsilon = (y - y_{\min})/(y_{\max} - y_{\min})$ and
    $f(\uq)$ is an affine transformation from $\mathbb{R}^{100}$ to
    $\mathbb{R}$. The marker shows the $(\sigma, \upsilon)$ combination 
    that resulted from numerical minimization of $D_{\mathrm{KL}}$.}
\label{fig:dklLinear}
\end{figure}

To verify whether minimizing $D_{\mathrm{KL}}$ to obtain parameters actually
translates to improved performance of our IS density, we synthesized a affine 
map from $\mathbb{R}^{100}$ to $\mathbb{R}$ (implementation details are provided
in the supplement in \Cref{supplement:fwdModels}). We measure performance by the 
relative RMSE in the probability estimate, $\tilde{e}_{\mathrm{RMS}}$, and we 
expect $\tilde{e}_{\mathrm{RMS}}$ to be small when $y = y^*$ and 
$\sigma = \sigma^*$. In addition, in this case, $\mu$ is available to us 
analytically. This provides yet 
another indicator of performance- the absolute difference between the 
analytical value and the IS estimate must be small when the optimal parameters 
are being used. 

\Cref{fig:dklLinear} shows the variation of 
$D_{\mathrm{KL}}$ with $\sigma^2$ at various
$y$ in addition to the optimal $(\sigma, y)$ combination that results from
numerical optimization. We conclude the following from the figure:

\begin{itemize}
    \item The optimal pseudo-data point lies almost exactly at the mid-point of $\DT$. 
    \item $D_{\mathrm{KL}}$ is extremely sensitive to the spread of the
        pseudo-likelihood probability density $\sigma$, much more so than the
        pseudo-data $y$. Intuitively, a large value for $\sigma$ emphasizes the pseudo-prior
        over the data so that sampling from $q(\uq)$ is akin
        to sampling from $p(\uq)$. On the other hand,
        too small a value for $\sigma^2$ results in $q(\uq)$ not having enough
        spread to cover the region where $f(\uq) \in \DT$, which could result in
        significant bias when the number of samples is small.
\end{itemize}

\begin{figure}[htbp]
    \centering
    \begin{subfigure}[t]{0.48\textwidth}
        \includegraphics[width=\textwidth]{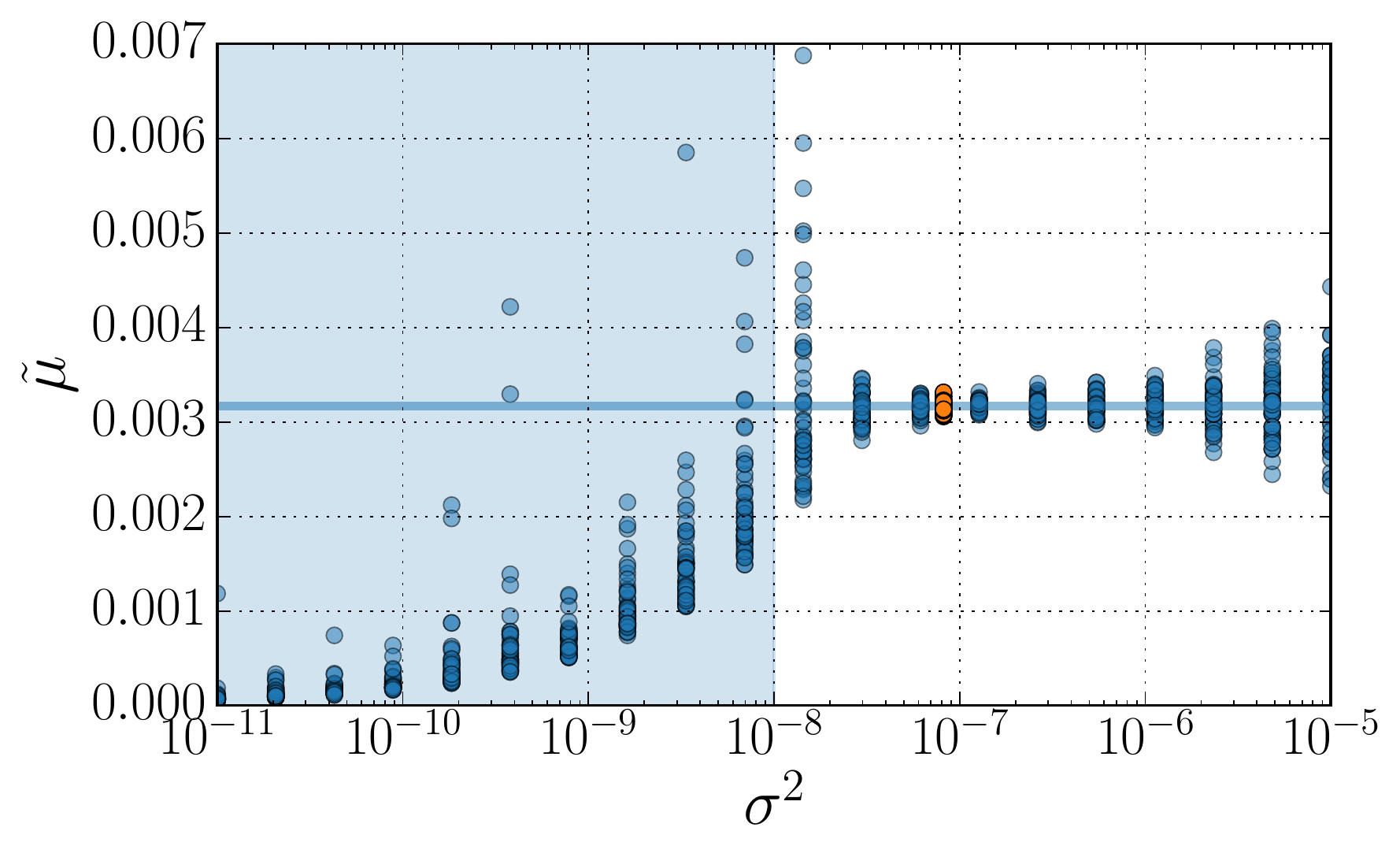} 
        \caption{The IS probability estimate $\tilde{\mu}$ against $\sigma^2$
            at $y = y^*$. Each marker indicates $\tilde{\mu}$ from an
            individual run. The solid line indicates the true value of $\mu$.
        The markers in orange denote runs at $\sigma = \sigma^*$.}
   \label{fig:varConvProbLinearNOpt1}
    \end{subfigure}
    ~
    \begin{subfigure}[t]{0.48\textwidth}
        \includegraphics[width=\textwidth]{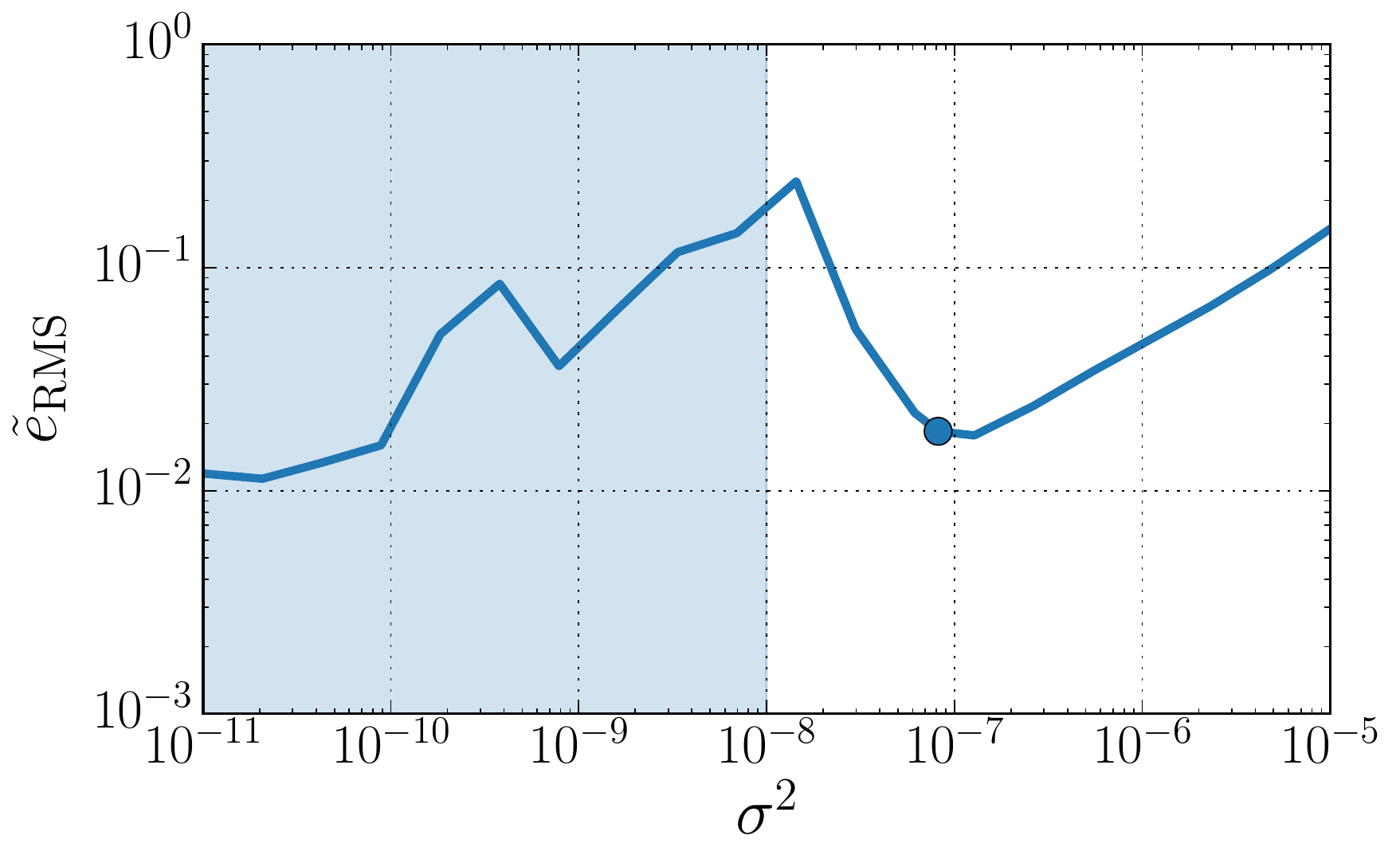} 
        \caption{The relative RMSE $\tilde{e}_{\mathrm{RMS}}$ against 
            $\sigma^2$ at $y
            = y^*$. Only the average $\tilde{e}_{\mathrm{RMS}}$ over all runs at fixed
        $\sigma^2$ is plotted. The marker is at $\sigma = \sigma^*$.} 
   \label{fig:varConvRMSELinearNOpt1}
    \end{subfigure}
    \caption{This figure shows the variation of the probability estimate
        $\tilde{\mu}$ and the relative RMSE $\tilde{e}_{\mathrm{RMS}}$ with the 
        likelihood 
        variance $\sigma^2$. At fixed $\sigma^2$, we performed 50 independent 
        runs using $N = 1000$ samples. For $\tilde{\mu}$, 
        we plot the probability estimate obtained from each run, whereas for
        $\tilde{e}_{\mathrm{RMS}}$, we plot the ensemble average at each 
        $\sigma$.} 
\label{fig:varConvLinearNOpt1}
\end{figure}

In \Cref{fig:varConvLinearNOpt1}, we fix $y = y^*$ and plot the variation of
the probability estimate, $\tilde{\mu}$, and the relative RMSE,
$\tilde{e}_{\mathrm{RMS}}$, with $\sigma^2$. For each value of $\sigma^2$, we 
performed
several independent runs. \Cref{fig:varConvProbLinearNOpt1} plots $\tilde{\mu}$
obtained from each run. In \Cref{fig:varConvRMSELinearNOpt1}, we plot the 
ensemble average of $\tilde{e}_{\mathrm{RMS}}$ over all simulations at fixed 
$\sigma^2$.
Both figures demonstrate that when $\sigma^2$ is small, both the probability 
estimate and the associated RMSE have significant bias (shaded region in the 
figures).  When $\sigma^2$ is large, error increases with $\sigma^2$ since the
emphasis on pseudo-data decreases. There lies an optimal $\sigma^2$ somewhere in 
between, and indeed, minimizing $D_{\mathrm{KL}}$ helps identify it. 
So far we've been using just one pseudo-data point $y \in \DT$. However, it is
also possible to use multiple pseudo-data points, $\{y_i\}_{i = 1}^n, y_i \in 
\DT$. In this case, using the same pseudo-likelihood density for all $y_i$, a 
posterior $p(\uq | y_i)$ and its corresponding Gaussian approximation can 
be obtained for each $y_i$. These Gaussians can then be collected into a 
mixture distribution to form the IS density. So, a possibility is to use the
following IS density:
\begin{align}
    q(\uq) = \frac{1}{n}\sum_{i = 1}^{n} \mathcal{N}\left(\uqmap^{(i)}, 
    {\left(\mathbf{H}^{(i)}_{\mathrm{GN}}\right)}^{-1}\right).
\end{align}

where $\uqmap^{(i)}$ is the MAP point corresponding to $y_i$ and
$\mathbf{H}_{\mathrm{GN}}^{(i)}$ is
the Hessian matrix of $-\log p(\uq | y_i)$ at $\uqmap^{(i)}$. Next,
we investigate whether using $n > 1$ pseudo-data points in $\DT$ leads to better performance
than using just one pseudo-data point, i.e., $n = 1$. 

\begin{figure}[htbp]
    \centering
    \begin{subfigure}[b]{0.49\textwidth}
    \includegraphics[width=\textwidth]{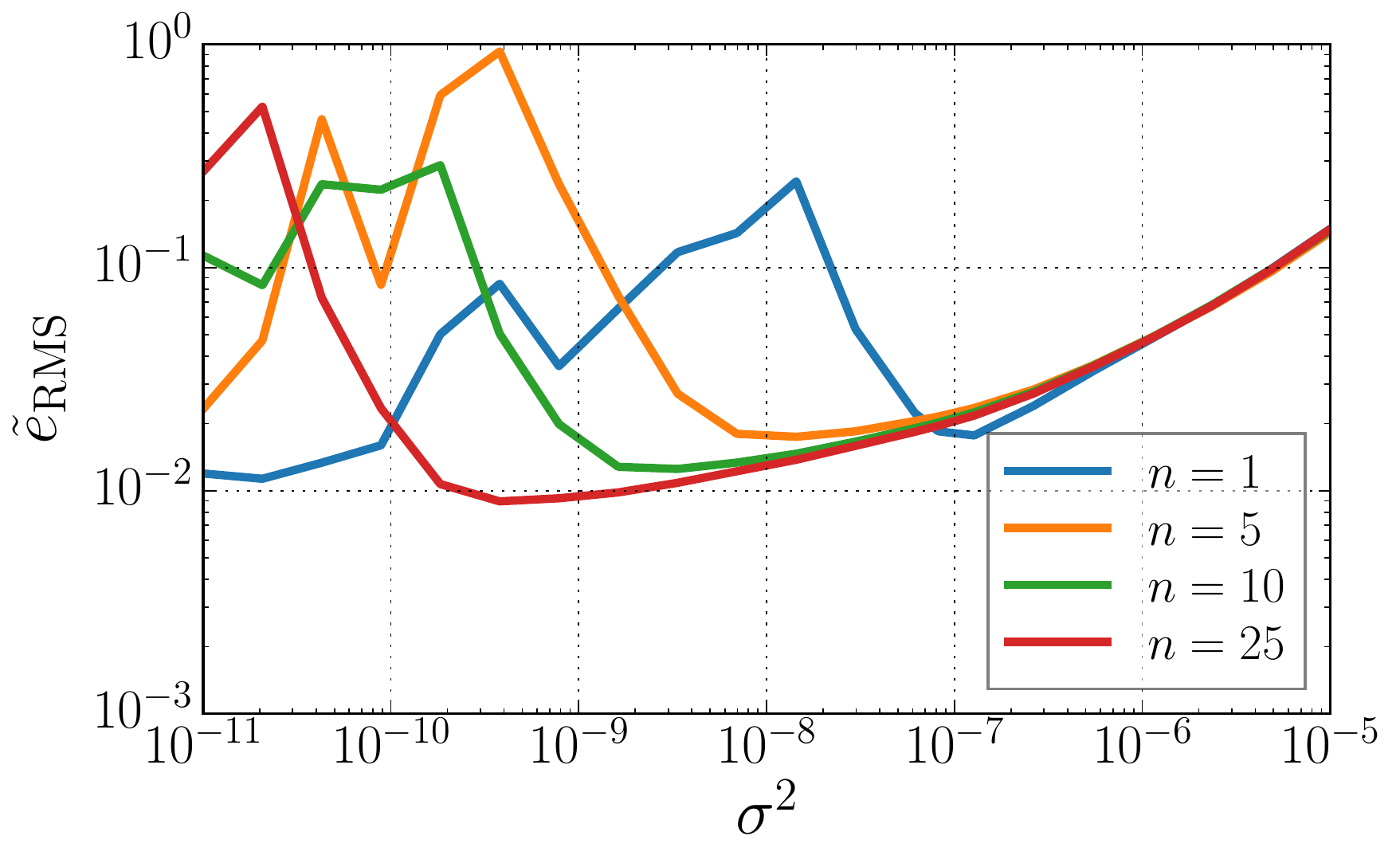} 
    \caption{Linear inverse problem.}
    \end{subfigure}
    \begin{subfigure}[b]{0.49\textwidth}
    \includegraphics[width=\textwidth]{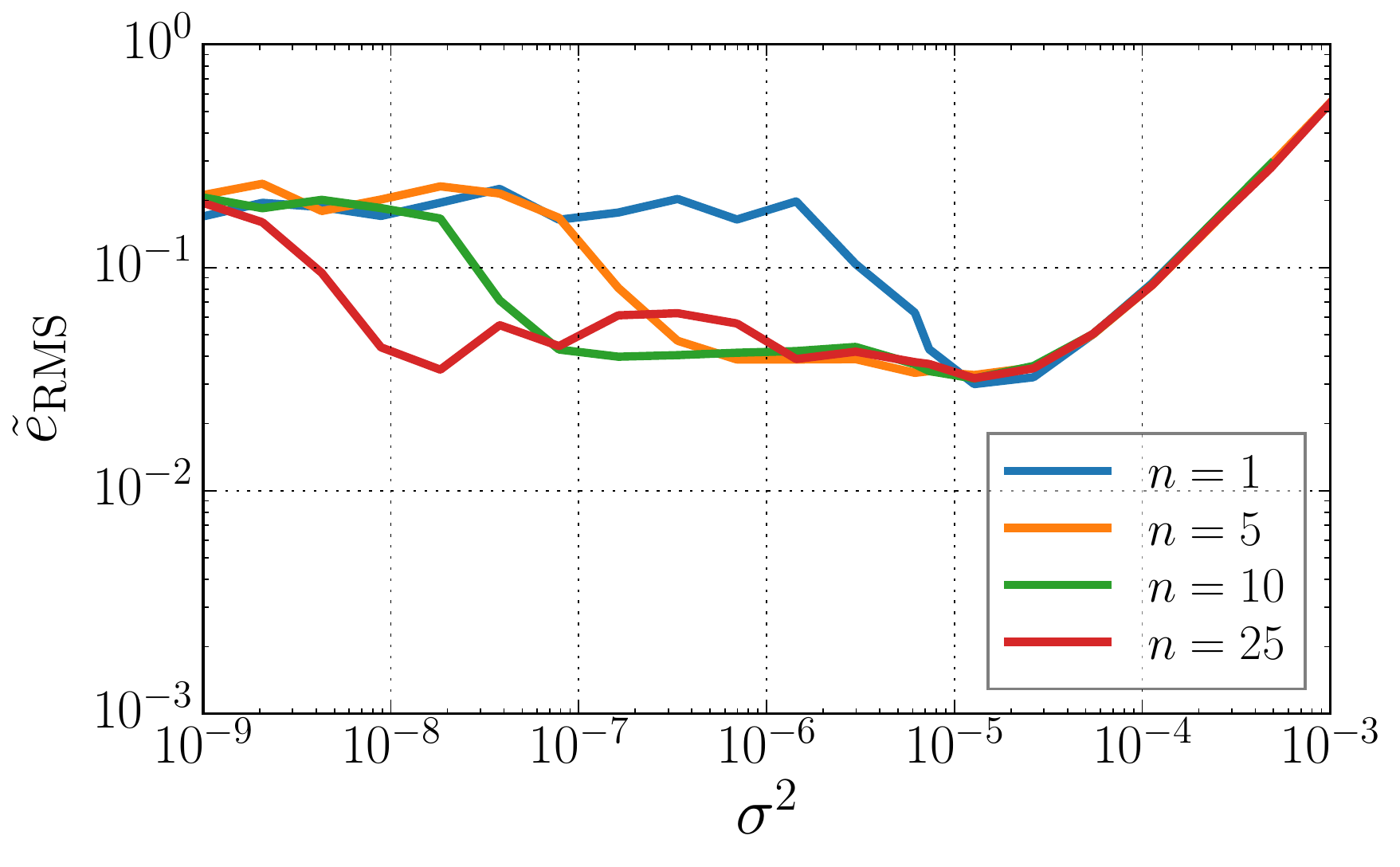} 
    \caption{Synthetic non-linear problem.}
    \end{subfigure}
    \begin{subfigure}[b]{0.49\textwidth}
    \includegraphics[width=\textwidth]{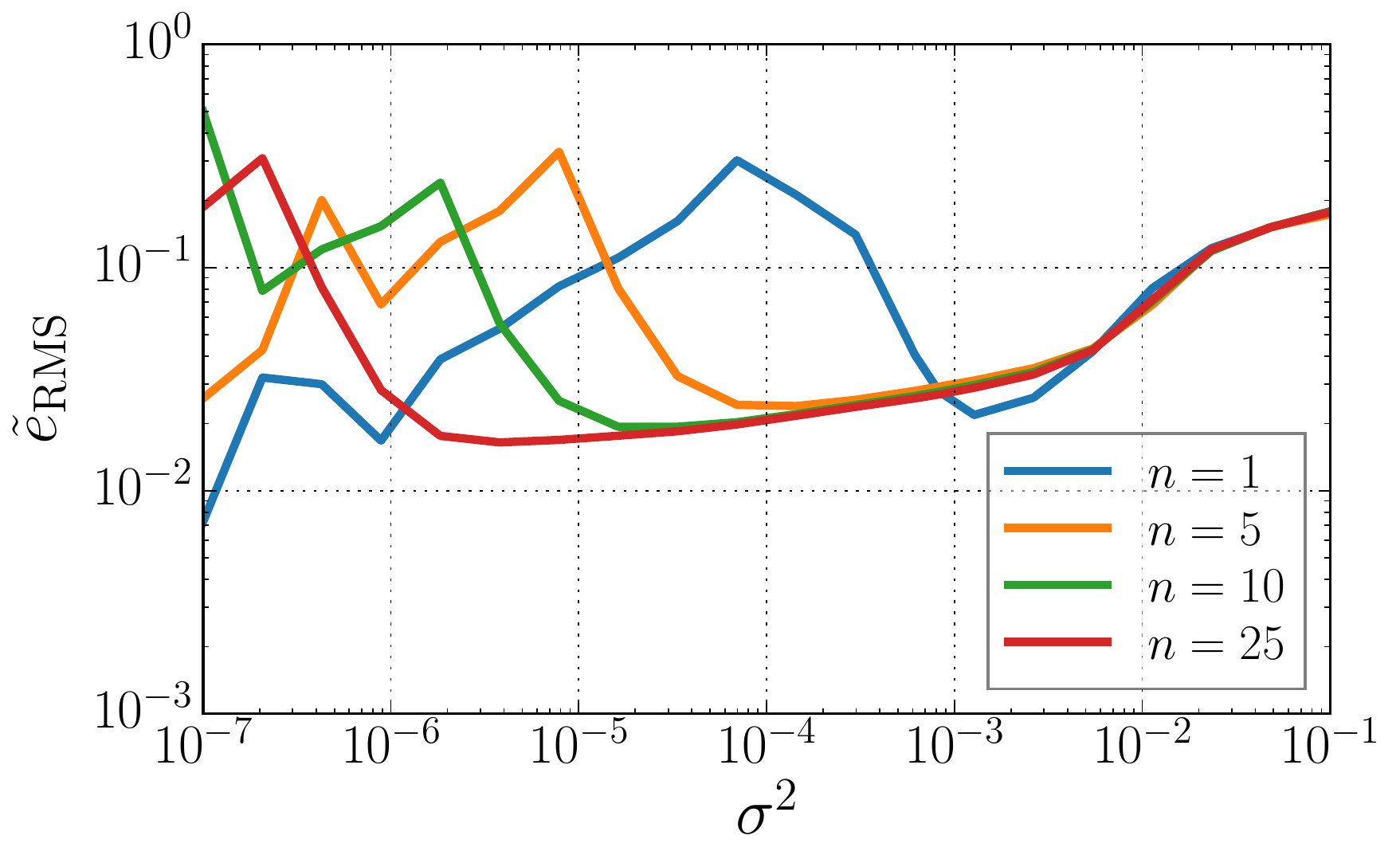} 
    \caption{Single step reaction.}
    \end{subfigure}
    \begin{subfigure}[b]{0.49\textwidth}
    \includegraphics[width=\textwidth]{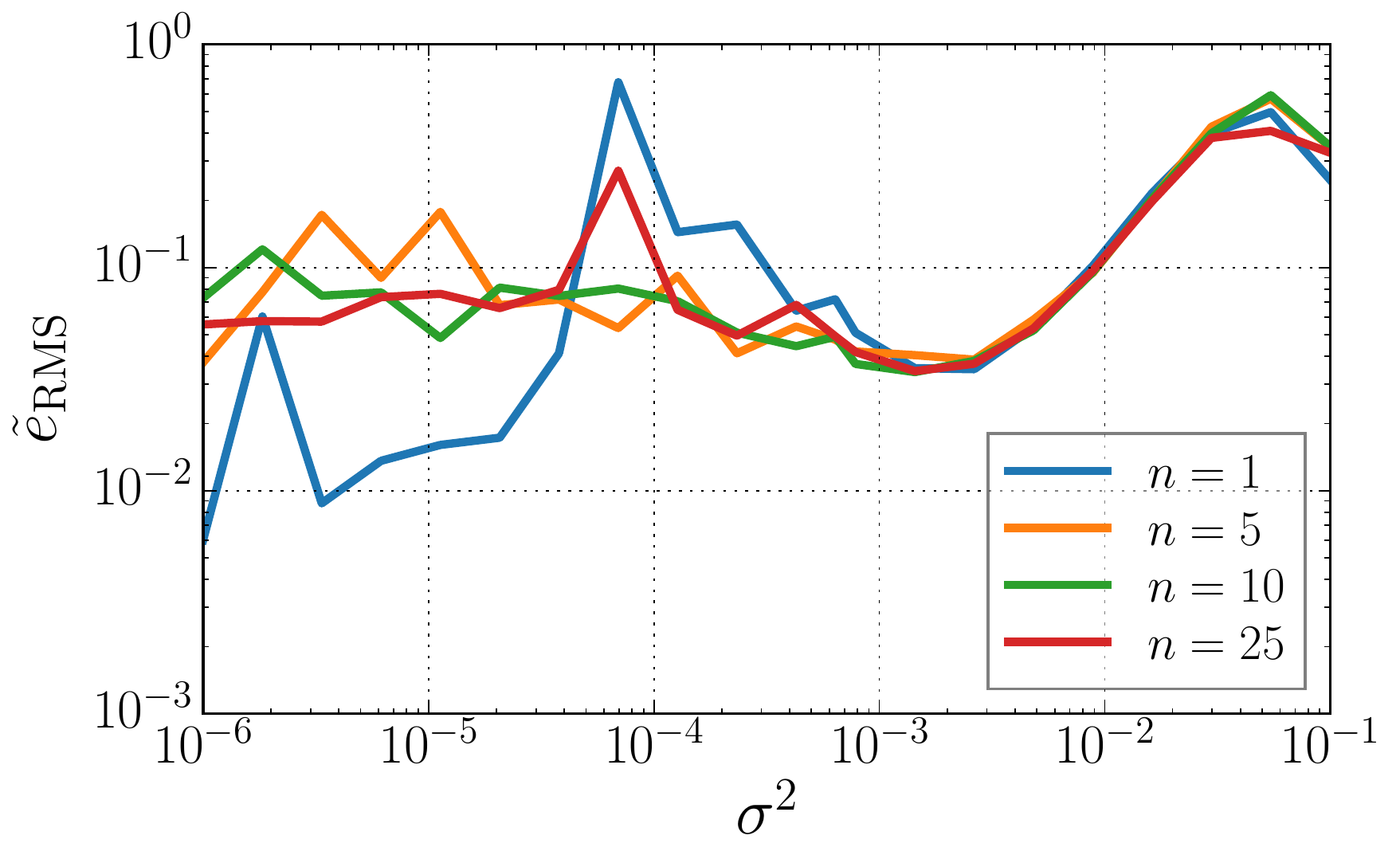} 
    \caption{Elliptic PDE.}
    \end{subfigure}
    \caption{This figure shows how $\tilde{e}_{\mathrm{RMS}}$ varies with 
        $\sigma^2$ at different values of $n$ for various forward models. At 
        fixed $\sigma^2$, we perform 50 independent simulations and report the 
        ensemble averaged $\tilde{e}_\mathrm{RMS}$. When $n = 1$, we use the 
        optimal data point obtained employing the appropriate tuning strategy
        described in the text. When $n > 1$, we select $n$ evenly spaced 
        points in $\DT$. Similar to \Cref{fig:varConvLinearNOpt1}, IS estimates are
        biased when $\sigma^2$ is small. The extent of the biased regions depends on
        $n$ and appears to decrease as $n$ increases.}
\label{fig:varConv}
\end{figure}

To ensure a fair comparison between the two cases, they must each be run using
their respective optimal parameters. When $n > 1$, the tunable parameters are 
the number of pseudo-data points, $n$, their values, $\{y_i\}_{i = 1}^{n}$, 
and the variance of the common pseudo-likelihood density, $\sigma^2$. However, 
minimizing the Kullback-Leibler distance between $q$ and $q^*$ to obtain 
parameters is no longer possible. This is because the Kullback-Leibler distance 
between two Gaussian mixtures doesn't have a closed form
expression~\cite{hershey2007approximating}. To 
proceed, given $n > 1$, we fix $\{y_i\}$ to be $n$ evenly spaced points in 
$\DT$. We then sweep over several values of $n$ and $\sigma^2$ to investigate 
whether increasing $n$ has any advantages.

Empirical evidence seems to suggest no. In \Cref{fig:varConv}, we plot the
variation of the ensemble averaged $\tilde{e}_{\mathrm{RMS}}$ with $\sigma^2$ 
at various
values of $n$ and using various forward models, both linear and non-linear
(details of the forward models are provided in the supplement in \Cref{supplement:fwdModels}).
While the  error decreases with increasing $n$ for some cases, we believe the 
decrease isn't large enough to justify the increased computational cost of 
solving additional inverse problems. 

\begin{figure}[H]
\centering
    \begin{tikzpicture}[every text node part/.style={align=center}]
    \draw[black, ultra thick] (0,1) rectangle (7,5);
    \node at (2, 3) {$\bullet$};
    \draw[dashed, ultra thick] (4, 1) -- (4.5,5)
        node[pos=0.5, pin={[pin distance=3em, pin edge=solid]
        320:{\footnotesize $f^{-1}(\mathbb{Y})$}}] {};
    \draw[dashed, ultra thick] (5, 1) -- (5.5,5);
    \path[fill=gray!50] (4, 1) to (4.5, 5) 
                               to (5.5, 5) to (5, 1)
                               to (4, 1); 
    \draw [thick] plot [smooth cycle] coordinates {
            (4.25, 1.5) (4.5, 1.75) 
            (3, 2.5) (3, 3.1) (3.5, 3.5)
            (5.0, 4.0) (4.75, 4.5)
            (1.5, 3.75) (1.5, 2.25) 
            }
            node[
            pin={[pin distance=1em, pin edge=solid] 
            below left:{\footnotesize $p(\boldsymbol{x})$}}] {};;
    \draw [thick] plot [smooth cycle] coordinates {
            (4, 4) (3.9, 4.2) 
            (1.9, 3.7)
            (1.5, 3) (1.8, 2.3) 
            (3.25, 1.8) (3.35, 2.0)
            (2.6, 2.7) (2.7, 3.2)};
     \draw [thick] plot [smooth cycle] coordinates {
            (3, 3.8) 
            (1.8, 3.3)
            (1.75, 2.6)
            (2.5, 2.2)
            (2.7, 2.2)
            (2.3, 2.8) (2.4, 3.2)
            (2.9, 3.6)};
    \coordinate (f0) at (3, 0.25) {};
\end{tikzpicture}
\caption{Arbitrary unimodal $p$. In this case, the true posterior and the ideal
    IS density are multimodal. The optimizer used for computing $\uqmap$ will
only find one of these modes. Our IS density will then only sample around that
mode, leading to incorrect estimates.}
\label{fig:tikz_banana}
\end{figure}
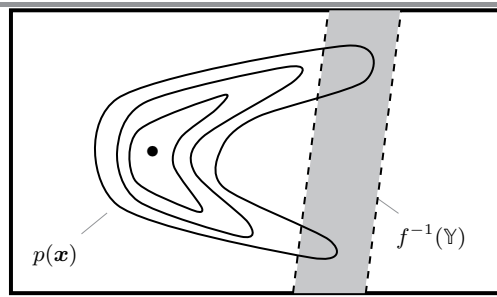

\subsection{Affine $f$, mixture-of-Gaussians $p$}
\label{subsection:linearMixGauss}

\sloppy If $p(\uq)$ is a mixture of Gaussians, $p(\uq) = \sum_{i = 1}^{k} w_i
p_i(\uq)$, then, notice that,

\begin{align}
    \mu = \int \ind_{\DT}\left(f(\uq)\right) p(\uq) \mathrm{d}\uq
        = \sum_{i = 1}^{k} w_i \int \ind_{\DT}\left(f(\uq)\right)
          p_i(\uq)\mathrm{d}\uq
        = \sum_{i = 1}^{k}w_i\mu_i
        \label{mixtureProb}
\end{align}

The contribution to $\mu$ from each component $p_i(\uq)$, $\mu_i$, can then be
calculated using BIMC as described above. If $\tilde{\mu}_i$ is the estimate
from each component, $\mu$ can be estimated as $\tilde{\mu} = \sum_{i =
1}^{k}w_i\tilde{\mu}_i$.

\subsection{Affine $f$, arbitrary unimodal $p$}

In this case, even though $p$ is unimodal, $q^*$, and the pseudo-posterior $p(\uq |
y)$, can be multi-modal (see \Cref{fig:tikz_banana}). Then, depending on the initial guess provided, 
the optimizer used for computing the MAPs may converge to only one of the 
modes. A local Gaussian characterization of the pseudo-posterior will only sample 
near this mode and all the other modes will be ignored. This will cause $\mu$ 
to be underestimated. To avoid this, we propose approximating $p$ with a 
mixture of Gaussians and then proceeding with the methodology outlined in the 
previous section. This will lead to an estimate whose accuracy is as good as 
the accuracy in approximating $p$ with a mixture of Gaussians.

\subsection{Non-linear $f$, Gaussian $p$}
\label{subsection:nonlinearGauss}

When $f(\uq)$ is non-linear, the KL divergence may not have a tractable 
closed-form expression even when only one pseudo-data point is used. Although a
sample based estimate of the KL divergence can be obtained, it would require
evaluating $f(\uq)$ for each sample, increasing the cost of constructing the IS
density. To compute the optimal parameters in this case, we instead
propose linearizing  $f(\uq)$ around the MAP point corresponding to 
an initial pseudo-data point, $y^{\mathrm{mid}} = \mrm{mid}\,\DT$, 
which we denote $\uqmap^{\mathrm{mid}}$. This necessitates solving 
another optimization problem (as in \Cref{mapEqn}), for which we 
require $\sigma^2$, a quantity we set out 
to tune in the first place. However, this $\sigma^2$ is only used to 
construct the linearization and has little bearing on subsequent sampling. 
We recommend setting $\sigma = 0.1  (y_{\max} - y_{\min})$. Once we have 
$\uqmap^{\mathrm{mid}}$, we linearize $f(\uq)$ as follows:

\begin{align}
    f(\uq) \approx f(\uqmap^{\mathrm{mid}}) + \mathbf{J}^{\mathrm{mid}} 
                    (\uq - \uqmap^{\mathrm{mid}})
\end{align}

Here, $\mathbf{J}^{\mathrm{mid}} \in \mathbb{R}^{1\times m}$ is the Jacobian 
matrix of the $f(\uq)$ evaluated at $\uqmap^{\mathrm{mid}}$. From here on, we 
can proceed to obtain the optimal parameters as in the affine case by 
identifying $\vect{v}^T \equiv \mathbf{J}^{\mathrm{mid}}$ and 
$\beta \equiv f(\uqmap^{\mathrm{mid}}) 
        - \mathbf{J}^{\mathrm{mid}}\uqmap^{\mathrm{mid}}$. Note that such a
procedure will not reveal the true optimal parameters that correspond to the
non-linear forward model. It only provides an estimate, but allows us to use
analytically derived expressions and keep computational costs low. Another
consequence of linearizing $f(\uq)$ is that it allows for the analytical
computation of the rare event probability associated with the linearized map
(\Cref{linProb}). We will refer to this estimate of $\mu$ as the linearized 
probability estimate, $\mu_{\mathrm{lin}}$.

\subsection{Non-linear $f$, mixture-of-Gaussian $p$}

This case is similar to \Cref{subsection:linearMixGauss}. Recall that $\mu$ is
just the weighted sum of probability corresponding to each component mixtures,
$\mu_i$. Each $\mu_i$ can be estimated by the method outlined above, and then
weighed and summed to obtain an estimate for $\mu$. 
 
\begin{algorithm}[H] \caption{{BIMC}}
\label{algo:bimc}
  \begin{algorithmic}[1]
      \Require $f(\uq)$, $p(\uq)$, $\DT$, $N$
      \Ensure $\tilde{\mu}$\\
      \Comment{\textcolor{gray}{\% Select optimal parameters, $y^*, \sigma^*$}}
      \State $y^{\max} \gets \max \DT$,  
      $y^{\min} \gets \min \DT$, $y^{\mathrm{mid}} \gets 0.5 (y^{\mathrm{min}} +
      y^{\mathrm{max}})$
      \State $\sigma_0 \gets 0.1(y^{\max} - y^{\min})$
      \State $\uqmap^{\mathrm{mid}} \gets \texttt{getMAP}(y_{\mathrm{mid}},
          \sigma_0)$ \Comment{{\scriptsize \textcolor{gray}{\% Minimize \Cref{negLogPost} using $y =
          y^{\mathrm{mid}}, \sigma=\sigma_0$}}}
      \State $\vect{v}^T \gets
             \left. \frac{\partial f(\uq)}{\partial \uq}\right|_{\uq = \uqmap}$
      \State $\beta \gets 
             f(\uqmap^{\mathrm{mid}}) - \vect{v}^{T}
             \uqmap^{\mathrm{mid}}$
             \State $(y^*, \sigma^*) \gets
            \texttt{minimizeKLDiv}(\vect{v}, \beta, p)$
            \Comment{{\scriptsize \textcolor{gray}{\% Minimize
            $D_{\mathrm{KL}}(q^* || q)$ as in \Cref{dkl}}}}\\

      \State
      \Comment{\textcolor{gray}{\% Build IS density using optimal parameters}}
      \State $\uqmap \gets \texttt{getMAP}(y^*, \sigma^*)$
          \Comment{{\scriptsize \textcolor{gray}{\% Minimize \Cref{negLogPost} using $y =
          y^*, \sigma=\sigma^*$}}}

      \State $\mathbf{H}_{\mathrm{GN}} \gets \texttt{getHessian}(\uqmap,
      y^*, \sigma^*)$
      \Comment{{\scriptsize \textcolor{gray}{\% Compute Hessian of
                  \Cref{negLogPost} at $\uqmap$ using $y =
          y^{*}, \sigma=\sigma^*$}}}

      \State $q(\uq) \gets \mathcal{N}(\uqmap,
      \mathbf{H}_{\mathrm{GN}}^{-1})$\\

      \State
      \Comment{\textcolor{gray}{\% Sample from $q$ to estimate $\mu$}}
      \For{$i = 1, \ldots, N$}
            \State $\uq_i \sim q(\uq)$
            \State $w_i \gets \ind_{\DT}(f(\uq_i))p(\uq_i) / q(\uq_i)$
      \EndFor

      \State $\tilde{\mu} \gets \sum_{i = 1}^{N}w_i/N$

      \State \Return $\tilde{\mu}$
  \end{algorithmic}
\end{algorithm}

\subsection{Summary}
To summarize, in this section we described how a fictitious Bayesian inverse
problem can be constructed from the components of the forward UQ problem. The
solution of this fictitious inverse problem yields a posterior whose Gaussian
approximation is our IS density. The parameters on 
which the IS density depends can be tuned by minimizing an analytical 
expression for its Kullback-Leibler divergence with respect to the ideal IS 
density. A drawback of our method is that we're restricted to nominal densities 
that are Gaussian mixtures or easily approximated by one. The overall algorithm
for arbitrary, non-linear $f$ is given in \Cref{algo:bimc}. Next, we present 
and discuss results of our numerical experiments.

\section{Experiments}

\begin{figure}[htbp]
        \centering
        \begin{subfigure}[b]{0.3\textwidth}
            \includegraphics[width=\textwidth]{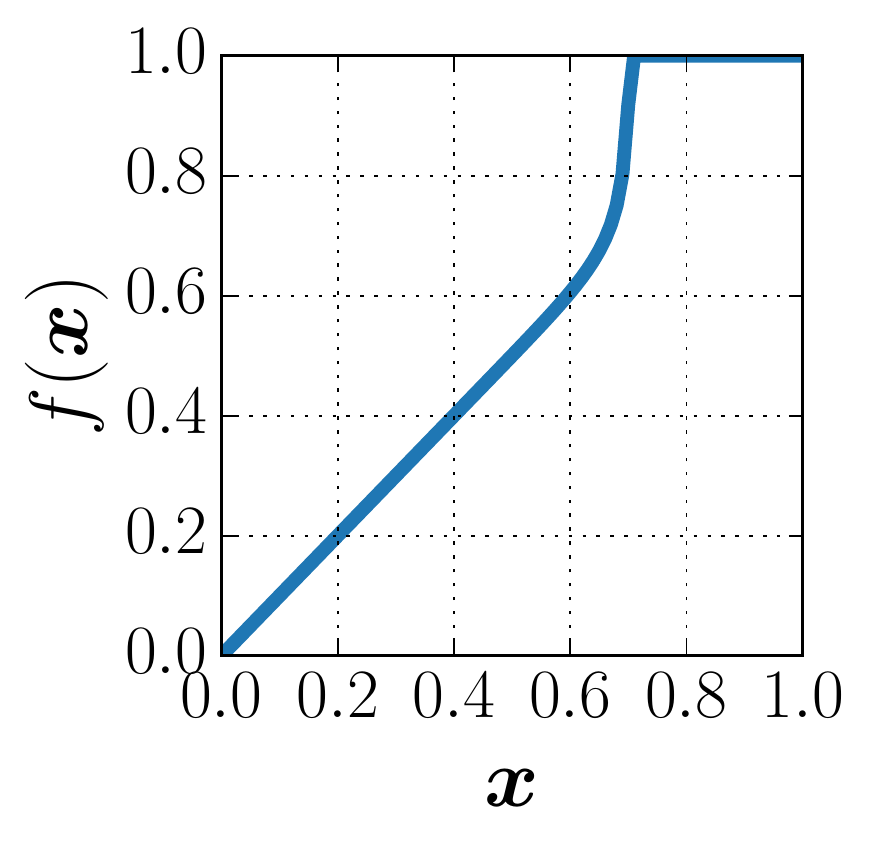}
            \caption{Single step reaction}
        \end{subfigure}
        \begin{subfigure}[b]{0.34\textwidth}
            \includegraphics[width=\textwidth]{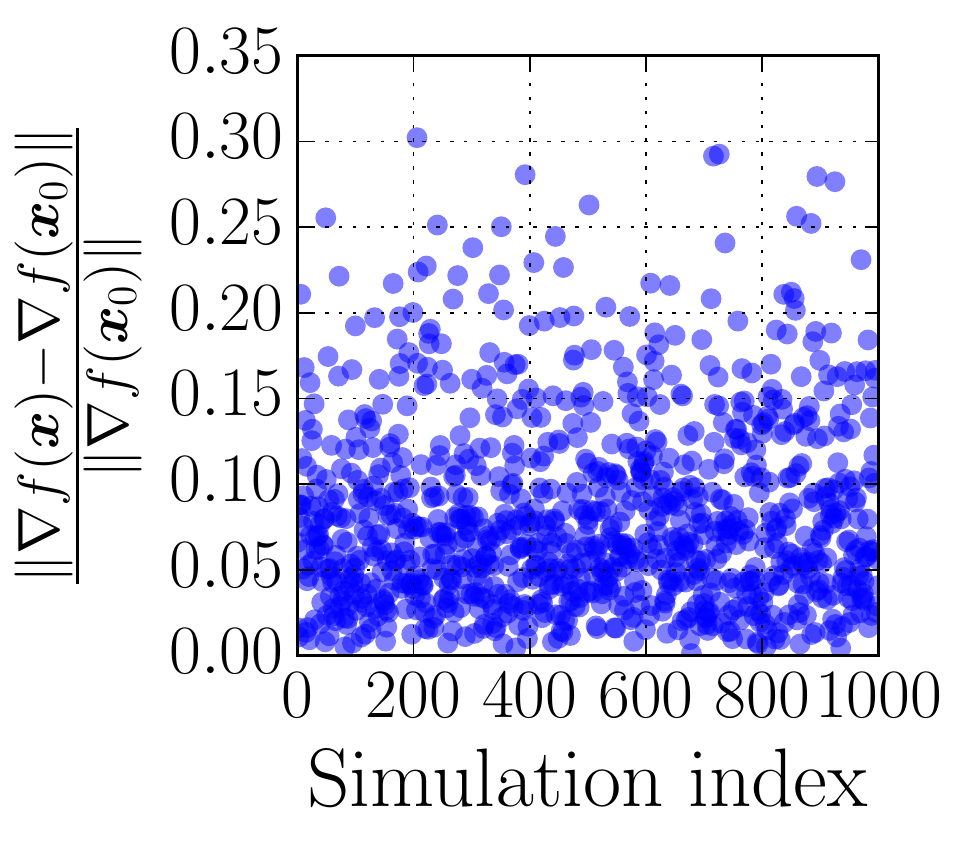}
            \caption{Autoignition}
        \end{subfigure}
        \begin{subfigure}[b]{0.34\textwidth}
            \includegraphics[width=\textwidth]{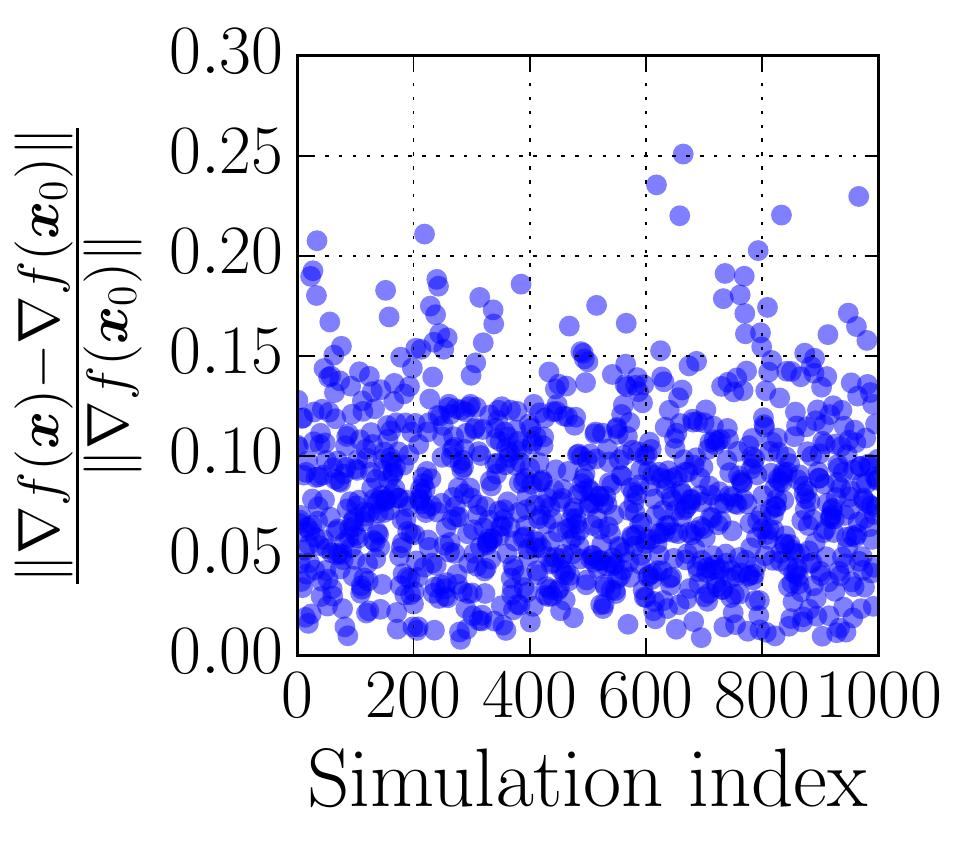}
            \caption{Synthetic non-linear}
        \end{subfigure}
        \caption{Non-linearity of $f$. For the single step reaction problem,
        we plot the full forward map $f$ for all possible values of the input
        $\uq$. For the autoignition and synthetic non-linear problems, we
        demonstrate how $\nabla f(\uq)$ varies. We draw samples, $\uq_i$,
        from $p(\uq)$, and evaluate $\|\nabla f(\uq) - \nabla
        f(\uq_0)\|/\|\nabla f(\uq_0) \|$, where $\uq_0$ is the prior mean.
        Departure from 0 of this quantity signifies the degree of non-linearity.}
        \label{fig:nonlinDemo}
\end{figure}

\label{section:experiments}
In this section, we present results that demonstrate the efficacy of our method.
We also report cases where our method fails (detailed discussion about failure
mechanisms of BIMC is postponed to the end of this section). The forward models 
we used in our experiments are briefly summarized
below. A detailed description of the models and the problem setup is given in
the supplement in 
\Cref{supplement:fwdModels}. \Cref{fig:nonlinDemo} shows the variation of
$f$ for select models and demonstrates that it is indeed non-linear.

\begin{itemize}
    \item Affine case: In this case $f(\uq)$ is a affine map  from
        $\mathbb{R}^{m}$ to $\mathbb{R}$. We choose $m = 2$ for illustration,
        and $m = 100$ for comparison with MC.
    \item Synthetic non-linear problem: In this case, $f(\uq)$ is defined to be the
        following map from $\mathbb{R}^m$ to $\mathbb{R}$. 
        \begin{align}
            f(\uq) = \vect{o}^T\vect{u},\,\text{where}
            \left(\mathbf{S} + \varepsilon \uq\uq^T\right)\vect{u} = \vect{b}.
        \end{align}
        Here, $\varepsilon \in \mathbb{R}$, $\vect{o}, \vect{u}, \vect{b} 
        \in \mathbb{R}^m$, and $\mathbf{S} \in \mathbb{R}^{m\times m}$. Again,
        $m = 2$ was chosen for illustration and $m = 10$ for comparison with MC.
    \item Single step reaction: The forward model here describes a single step
        chemical reaction using an Arrhenius type rate equation. A progress
        variable $u \in \left[0, 1\right]$ is used to describe the reaction. The
        parameter $\uq$ is the initial value of progress variable $u(0)$ and the
        observable $f(\uq)$ is the value of the progress variable at some final time 
        $t_f, u(t_f)$. Thus $f(\uq)$ is a map from $\mathbb{R}$ to $\mathbb{R}$.
    \item Autoignition: Here, we allow a mixture of hydrogen and air to undergo
        autoignition in a constant pressure reactor. A simplified mechanism with
        5 elementary involving 8 chemical species is used to describe the
        chemistry. The parameter $\uq$ is the vector of the initial equivalence
        ratio, initial temperature and the initial pressure in the reactor and
        the observable is the amount of heat released so that $f(\uq)$ is a map
        from from $\mathbb{R}^3$ to $\mathbb{R}$. 
    \item Elliptic PDE: In this system, we invert for the discretized 
        log-permeability field in some spatial domain given an observation of 
        the pressure at some point. The forward problem, that is, 
        obtaining the pressure from the log-permeability field, is 
        governed by an elliptic PDE. A finite element discretization results 
        in $f(\uq)$ being a map from $\mathbb{R}^{4225}$ to $\mathbb{R}$.
    \item The Lorenz system: Here, the forward problem is governed by the
        chaotic Lorenz equations \cite{lorenz1963deterministic}. 
        The parameter $\uq$ is the initial condition of
        the system while the observable is value of the first component of the
        state vector at some final time $t_f$. We simulate the Lorenz system
        over three time horizons, $t_f = 0.1$s, $t_f = 5$s, and $t_f = 15$s.
        BIMC fails over longer time horizons, i.e., when $t_f = 5$s and $t_f = 15$s. 
    \item Periodic case: Here, $f(\uq)$ is a periodic function in 
        $\mathbb{R}^2$, $f(\uq) = \sin(x_1) \cos(x_2)$. This is another case
        when BIMC fails. 
\end{itemize}

\paragraph{Sampling illustration}
We begin by presenting examples in low-dimensions that illustrate the quality of
samples from BIMC. In \Cref{fig:viz}, we compare samples generated using MC 
and BIMC. We also depict the ideal IS density $q^*$ in the figures, either using 
contours, or through samples. As expected, the variance of the IS density in our
method is only decreased in one data-informed direction. The extent of this
decrease depends on the variance of the pseudo-likelihood density, $p(y | \uq)$,
and a tuning algorithm based on minimizing the Kullback-Leibler distance leads
to a good fit between the spread of $q^*(\uq)$ and $q(\uq)$ in this direction.
In all other directions, the spread of $q(\uq)$ is same as that of $p(\uq)$.
This is because the pseudo-data $y$ does not inform these directions.

\begin{figure}[htbp]
    \centering
    \begin{subfigure}[b]{0.45\textwidth}
        \includegraphics[width=\textwidth]{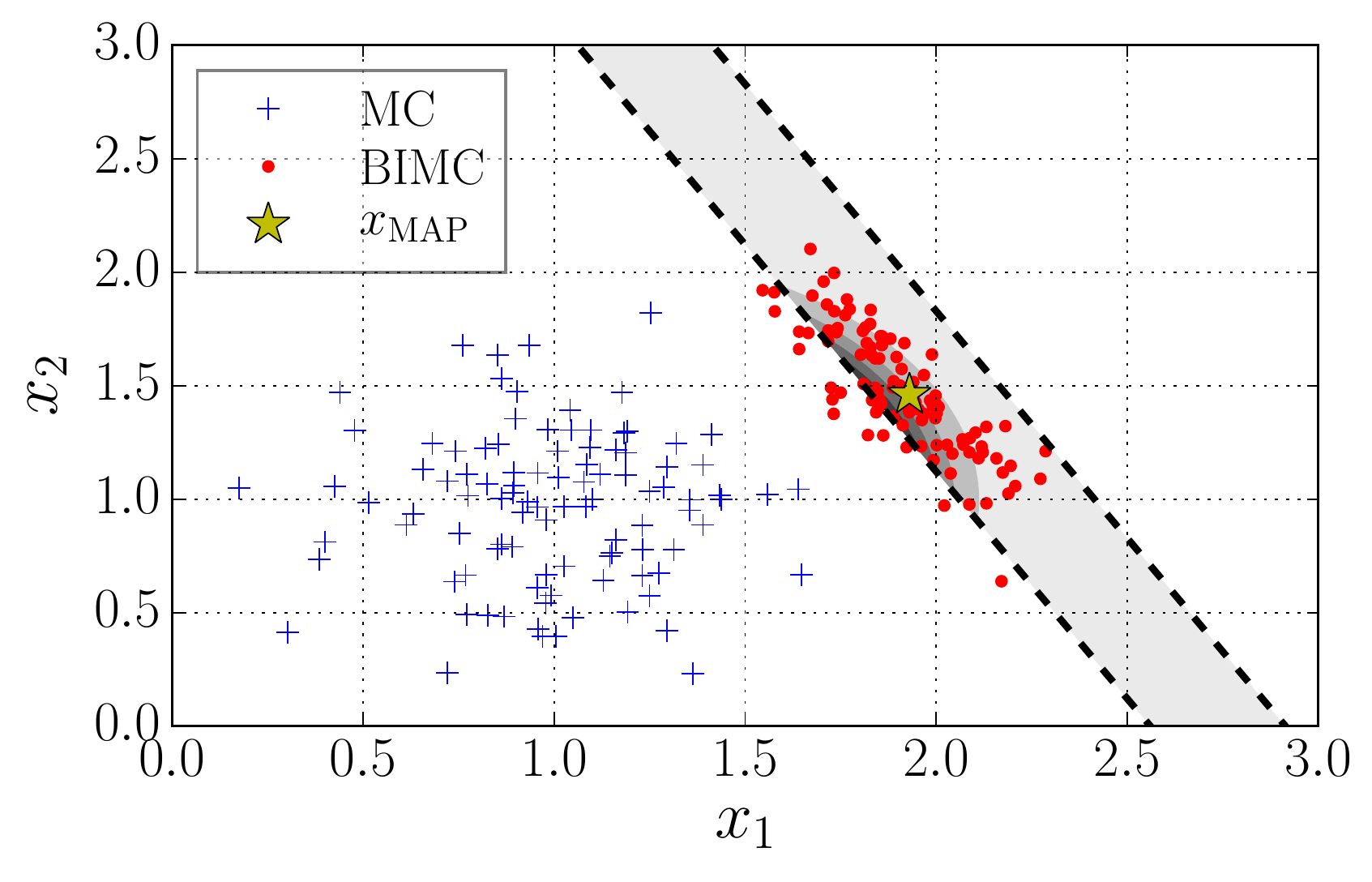} 
        \caption{The affine case.}
        \label{fig:vizLinear}
    \end{subfigure}
    ~
    \begin{subfigure}[b]{0.45\textwidth}
        \includegraphics[width=\textwidth]{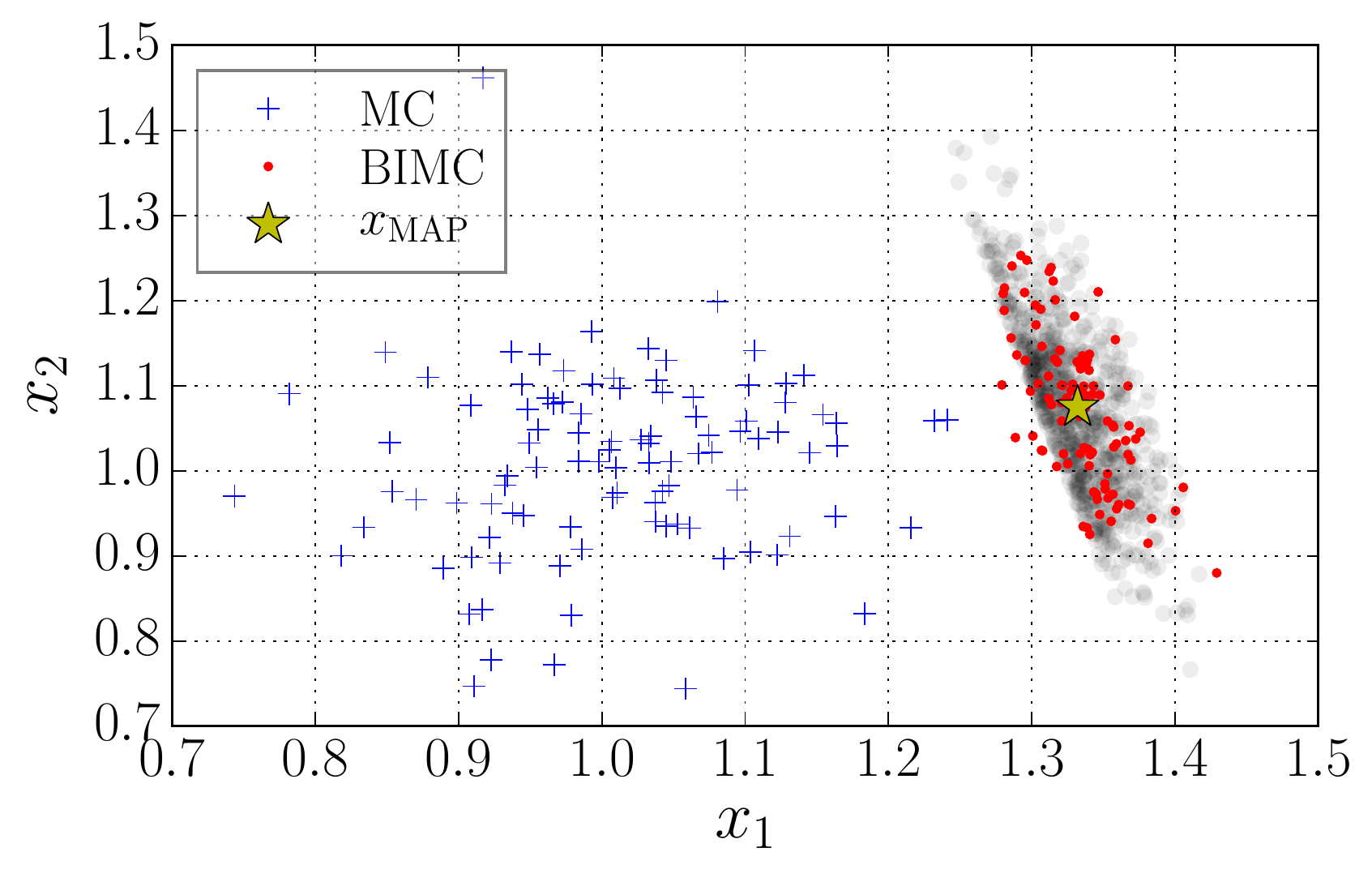} 
        \caption{The synthetic non-linear case.}
        \label{fig:vizRk1pert}
    \end{subfigure}
    ~
    \begin{subfigure}[b]{0.45\textwidth}
        \includegraphics[width=\textwidth]{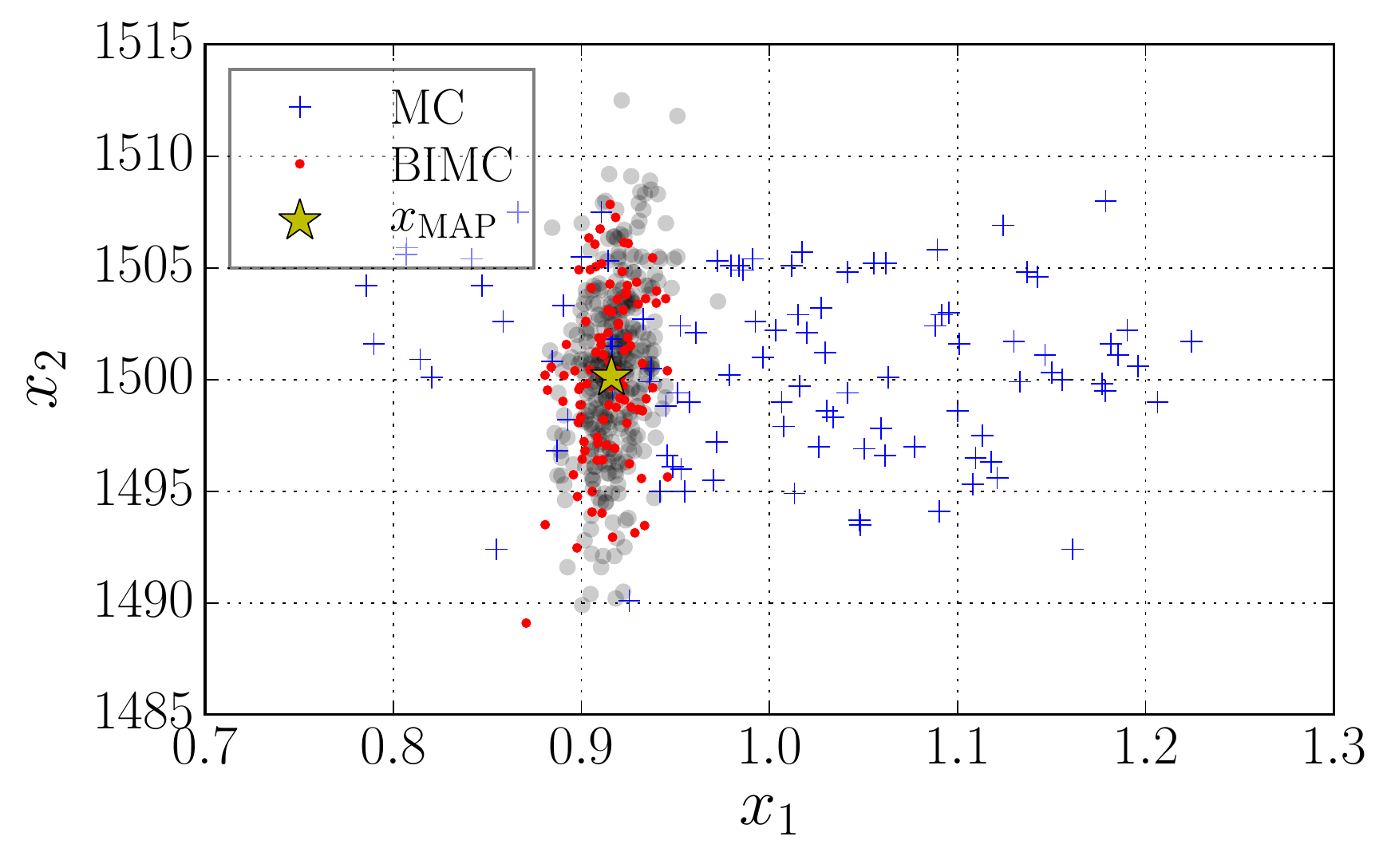} 
        \caption{Autoignition, $x_1$ - $x_2$ plane.}
        \label{fig:vizAuto12}
    \end{subfigure}
    ~
    \begin{subfigure}[b]{0.45\textwidth}
        \includegraphics[width=\textwidth]{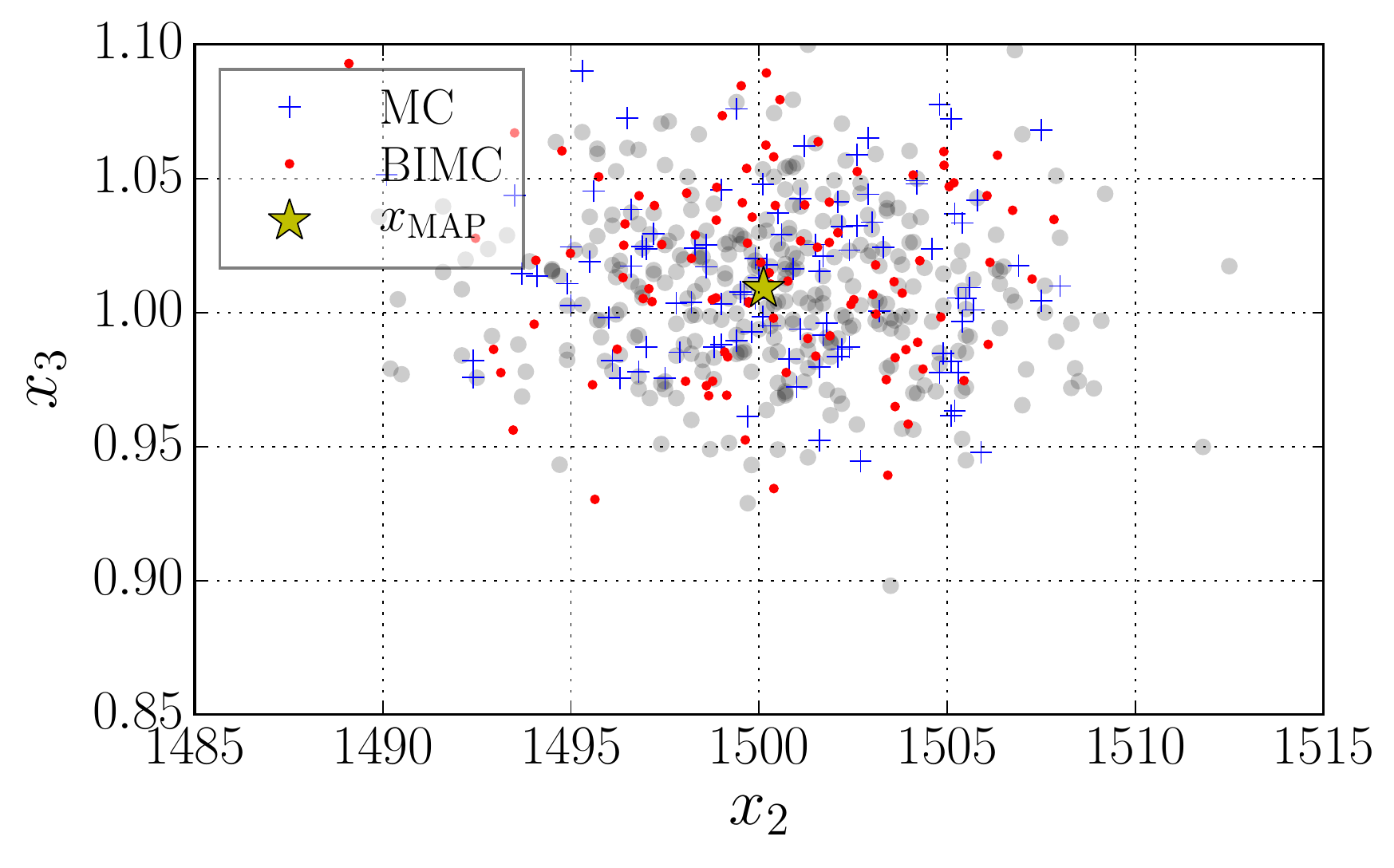} 
        \caption{Autoignition, $x_2$ - $x_3$ plane.}
   \label{fig:vizAuto23}
    \end{subfigure}
    ~
    \begin{subfigure}[b]{0.45\textwidth}
        \includegraphics[width=\textwidth]{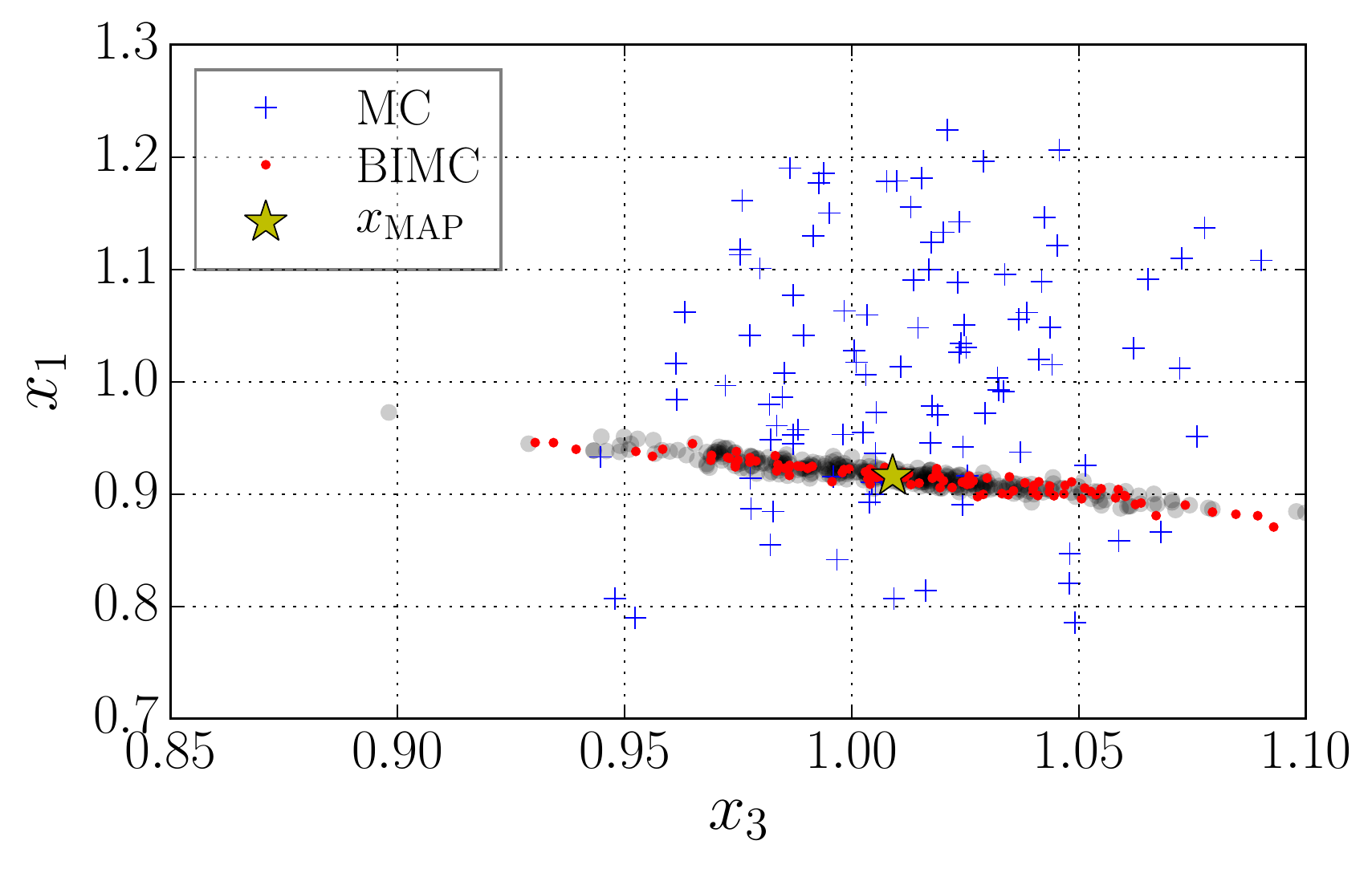} 
        \caption{Autoignition, $x_3$ - $x_1$ plane.}
   \label{fig:vizAuto31}
    \end{subfigure}
    ~
    \begin{subfigure}[b]{0.45\textwidth}
        \includegraphics[width=\textwidth]{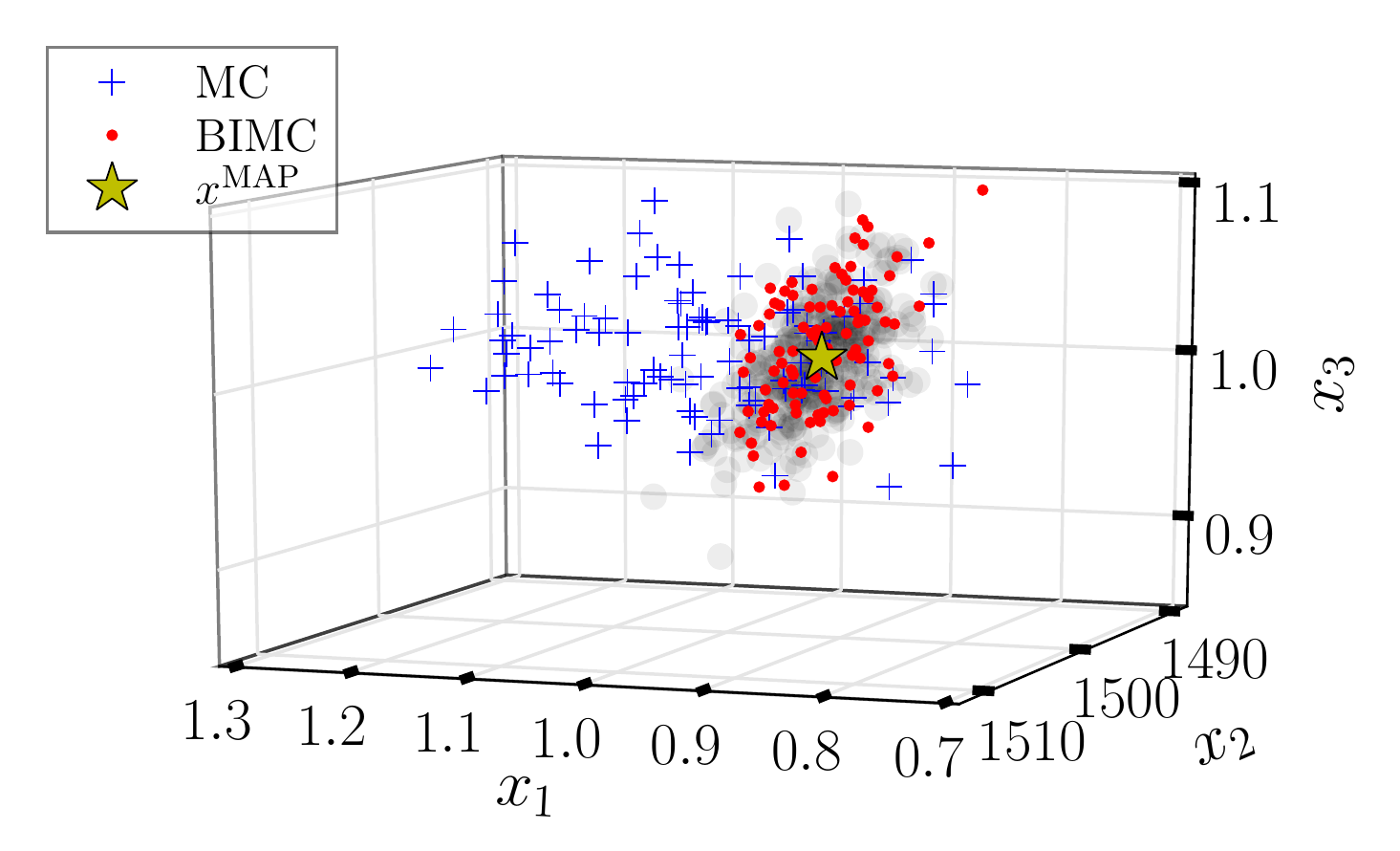} 
        \caption{Autoignition, 3D view.}
   \label{fig:vizAuto3D}
    \end{subfigure}
    \caption{Sampling illustration. In this figure, we plot 100 samples from
    $p(\uq)$ (which corresponds to vanilla MC) as well as $q(\uq)$ (which
    corresponds to BIMC) with $n = 1$ for the affine, synthetic non-linear, and
    the autoignition problems. For the affine case 
    (\subref{fig:vizLinear}), the region in $\mathbb{R}^2$ that evaluates 
    inside $\DT$ is analytically available and is plotted between the thick, 
    dashed lines. Also analytically available is the
    ideal IS density $q^*$ whose contours are plotted. For all other forward
    models, a scatter plot of samples drawn from $q^*$
    is used to represent its magnitude.} 
    \label{fig:viz}
\end{figure}

\begin{figure}[htbp]
    \centering
    \includegraphics[width=0.8\textwidth]{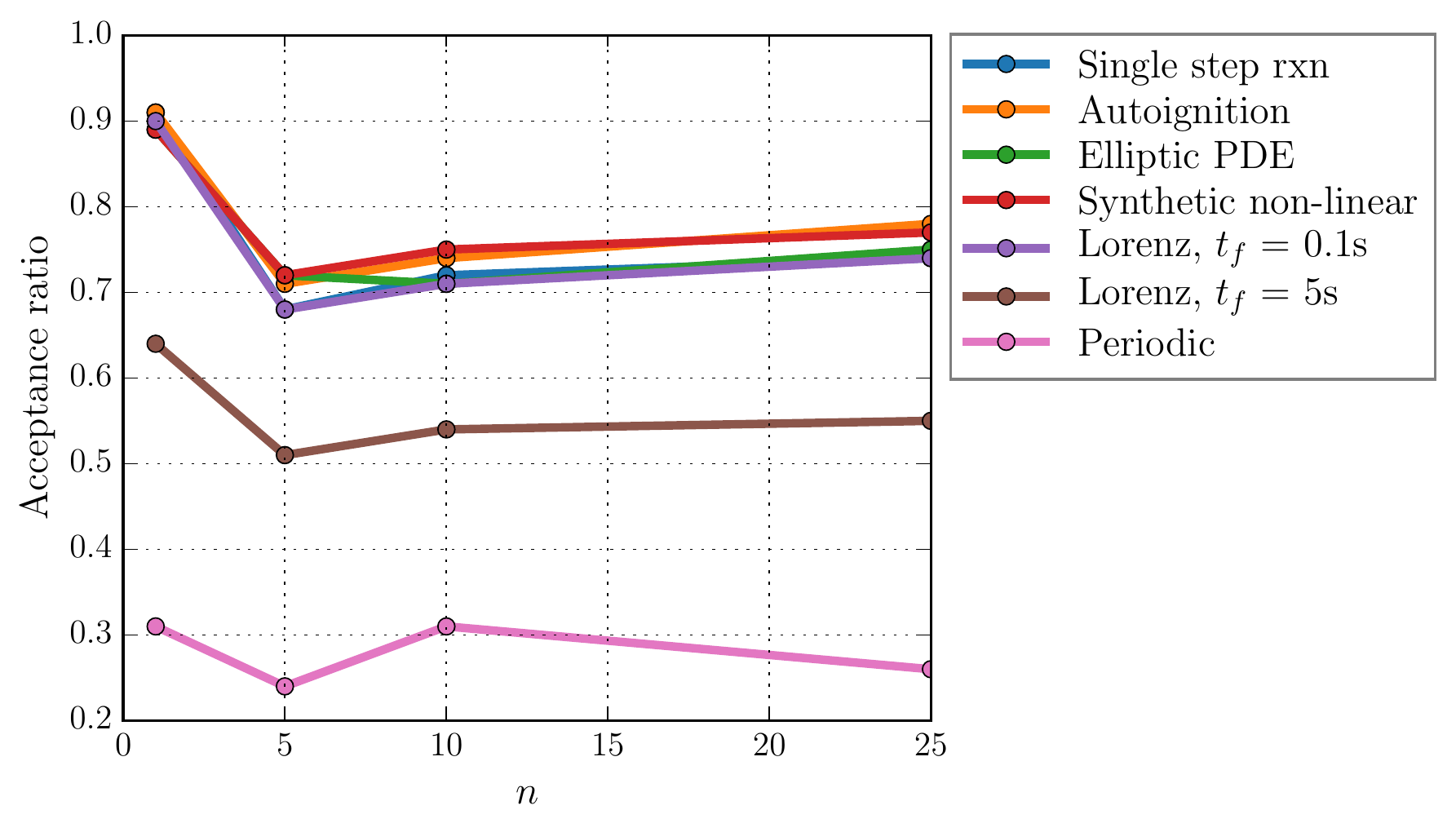} 
    \caption{Fraction of samples that evaluate inside $\DT$ for the different
    forward models at various values of $n$. In this experiment $N = 1000$ and
    $\mu$ spans two orders of magnitude, from $\mathcal{O}(10^{-2})$ to
    $\mathcal{O}(10^{-4})$. The BIMC methodology fails for the periodic and
    Lorenz, $t_f = 5$s cases, hence the lower acceptance ratio.}
\label{fig:frac}
\end{figure}

As a quantitative estimate of the quality of samples, we report the acceptance
ratio, defined as the fraction of samples that evaluate inside $\DT$. The
acceptance ratio resulting from BIMC is plotted in \Cref{fig:frac} (the
acceptance ratio from MC on the other hand is $\hat{\mu}$ by definition). 
We observe that $n = 1$ consistently leads 
to an acceptance ratio of around 90\% irrespective of $\mu$ (except in the
Periodic and Lorenz, $t_f = 5$s cases; these are failure cases and will be
discussed at the end of this section). The slight dip 
in the acceptance ratio when $n > 1$ can be attributed to the effect of always 
having $y_{\min}$ and  $y_{\max}$ as data points. Because these points 
lie at the edge of the interval $\DT$, they lead to an increased 
number of samples that are close to these limit 
points, but don't actually evaluate inside $\DT$. As $n$ increases however, the 
number of samples drawn from mixture components corresponding to these two points 
decreases and the acceptance ratio shows an upward trend. 

\paragraph{Convergence with number of samples}
Next, we compare the relative RMS error, 
$e_{\mathrm{RMS}}$, from MC and BIMC in \Cref{fig:numSamplesConv}. BIMC offers 
the same accuracy using far fewer number of samples and results in an order of 
magnitude or more of speedup. The exact speedup achieved depends on the 
magnitude of the probability. In addition, there is little asymptotic effect 
of using $n > 1$. The corresponding probability estimates are presented in 
the supplement in 
\Cref{supplement:results}.

\begin{figure}[H]
    \centering
    \begin{subfigure}[b]{0.45\textwidth}
        \includegraphics[width=\textwidth]{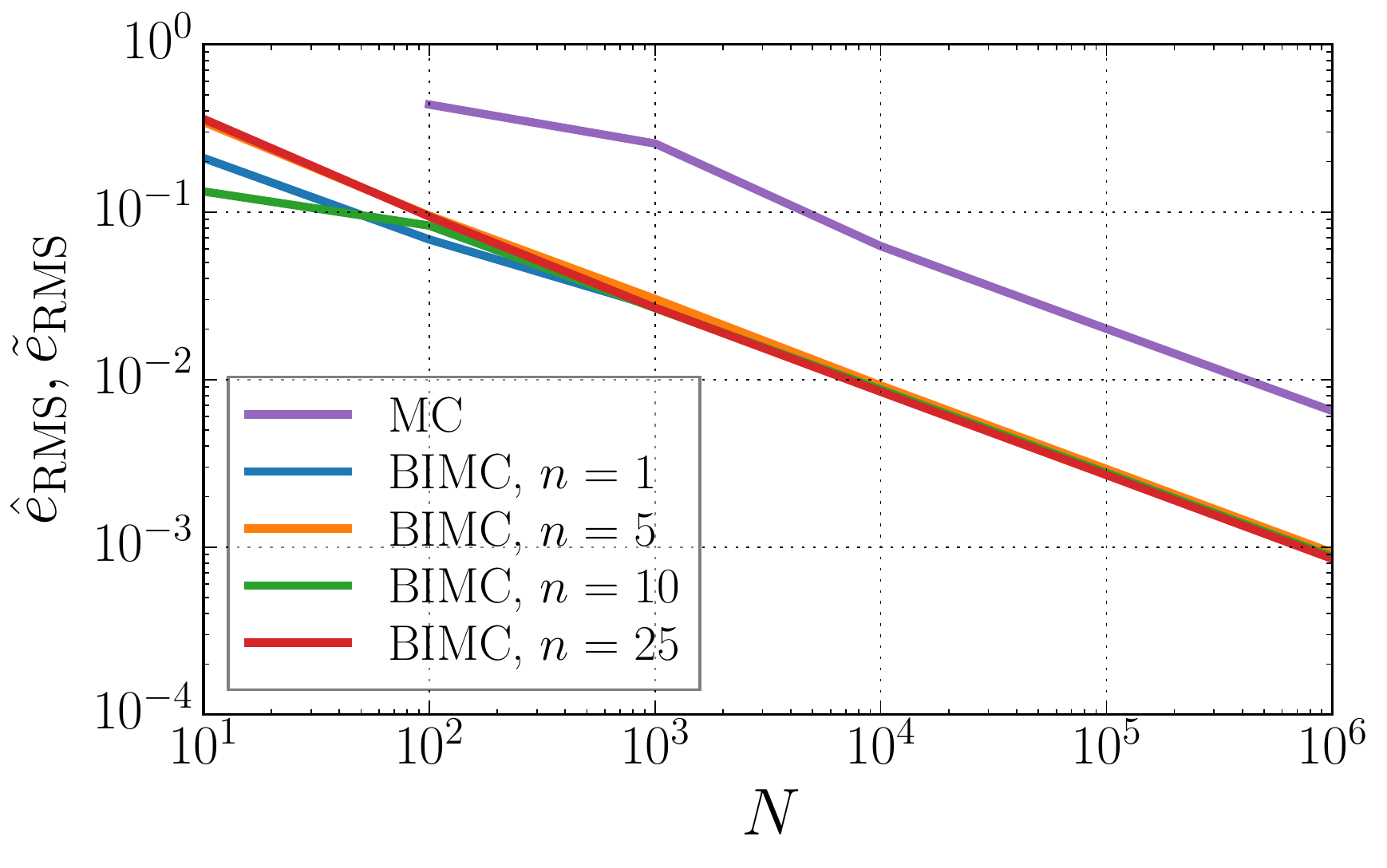} 
        \caption{Single step reaction, $\mu \approx 2.30\times10^{-2}$.}
   \label{fig:numSamplesConv1DComb}
    \end{subfigure}
    \begin{subfigure}[b]{0.45\textwidth}
        \includegraphics[width=\textwidth]{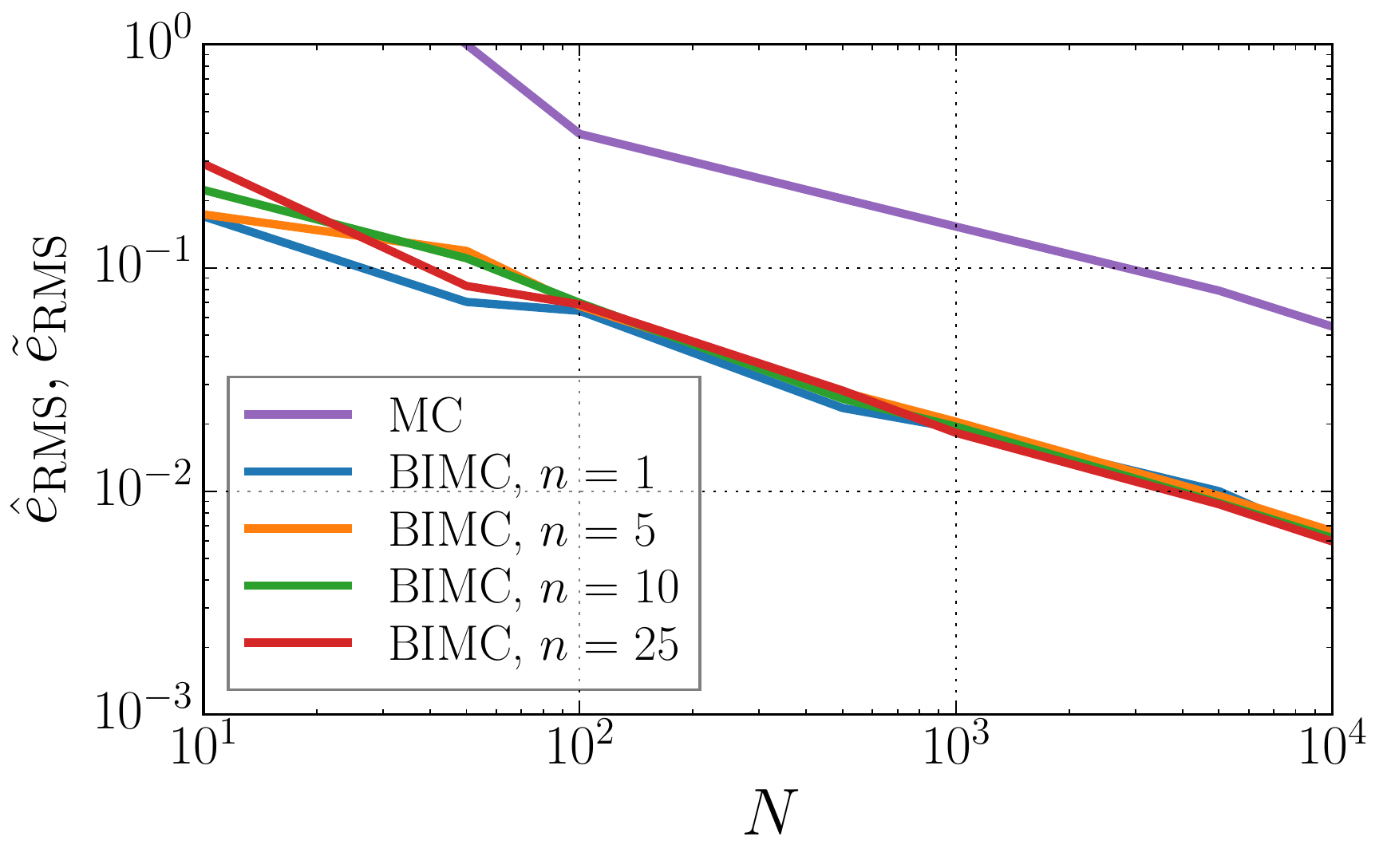} 
        \caption{Autoignition, $\mu\approx 3.24\times10^{-2}$.} 
   \label{fig:numSamplesConvAuto}
    \end{subfigure}
    \begin{subfigure}[b]{0.45\textwidth}
        \includegraphics[width=\textwidth]{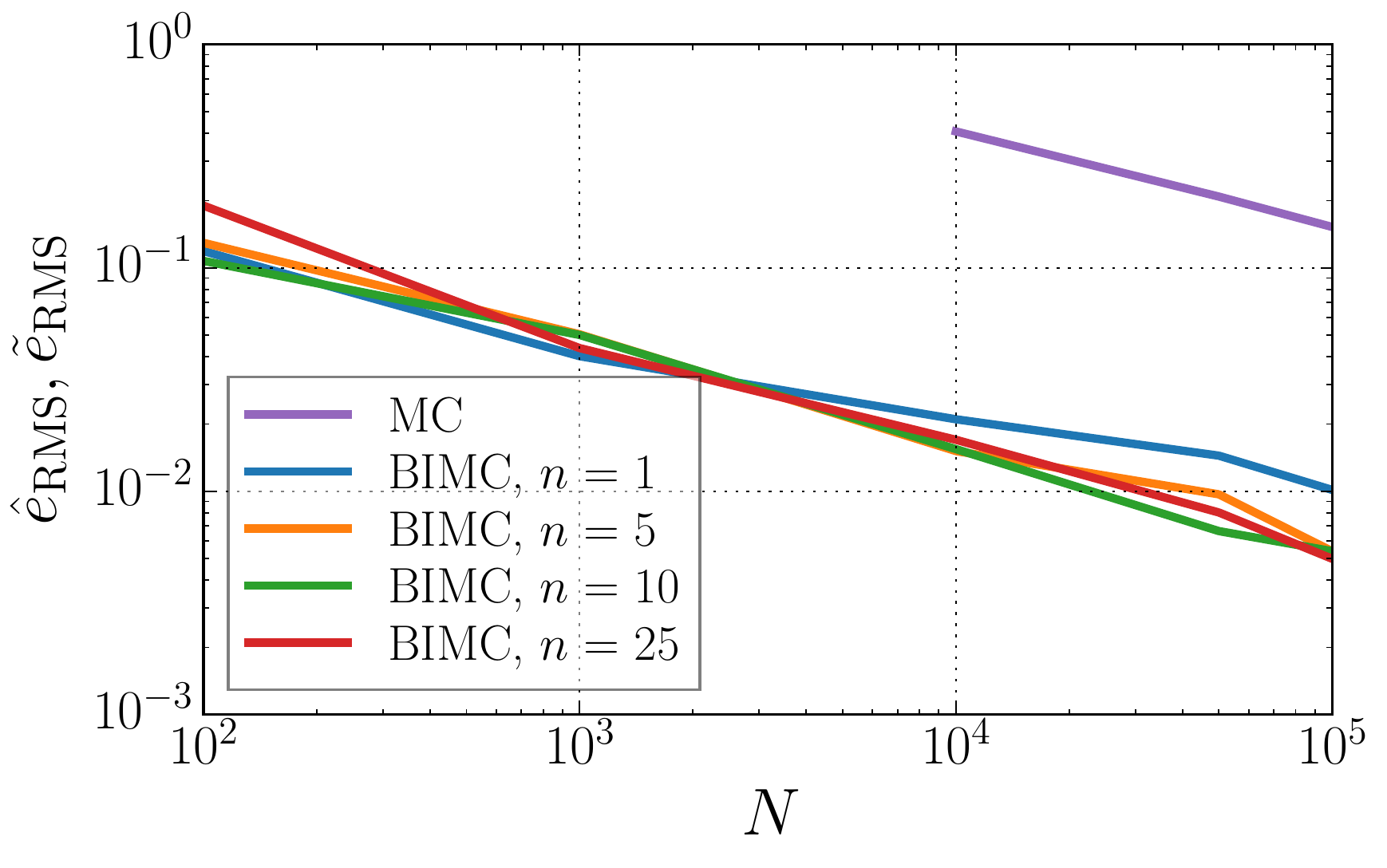} 
        \caption{Elliptic PDE, $\mu \approx 3.91\times10^{-4}$.} 
   \label{fig:numSamplesConvPoisson}
    \end{subfigure}
    \begin{subfigure}[b]{0.4\textwidth}
        \includegraphics[width=\textwidth]{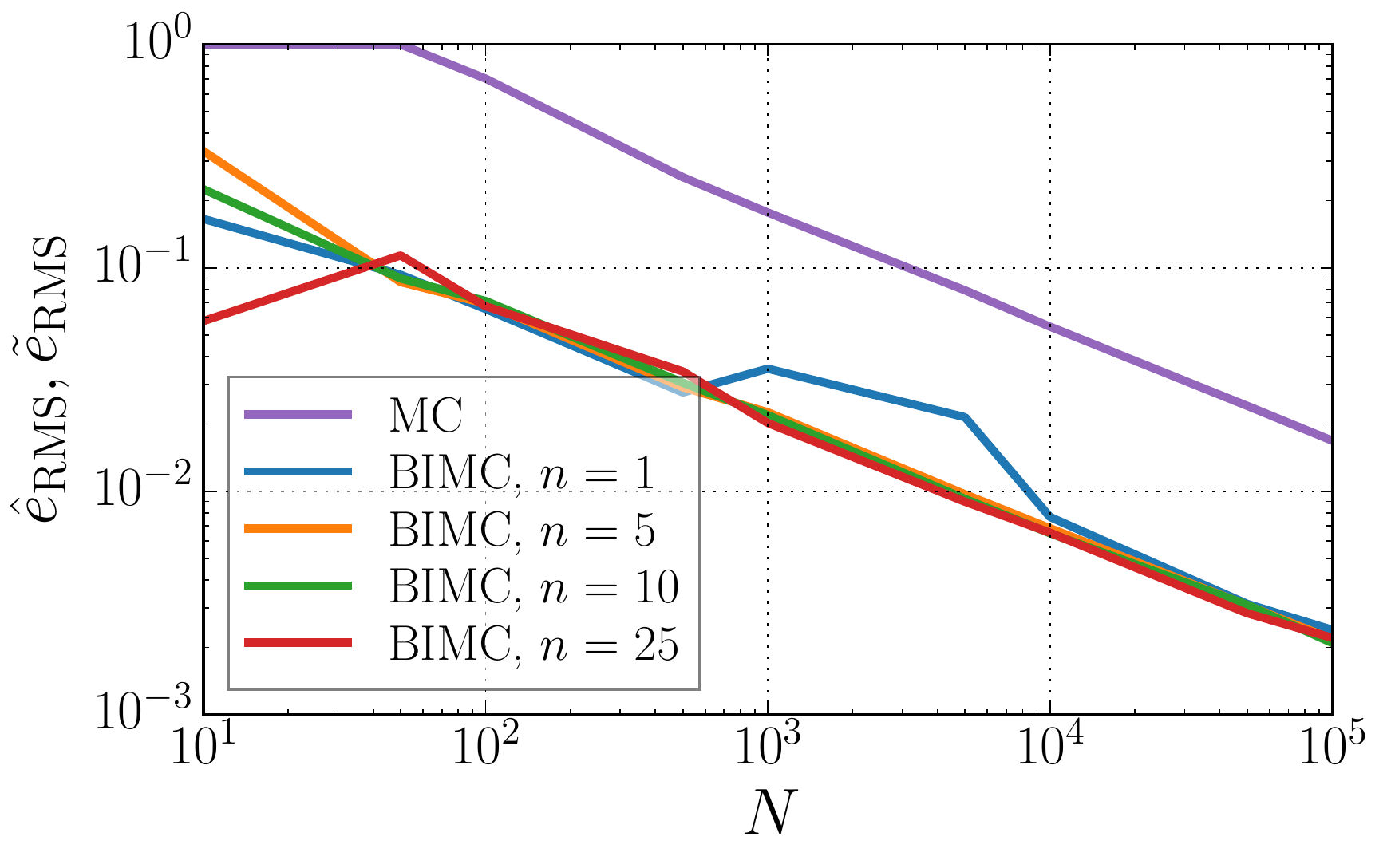} 
        \caption{Lorenz, $t_f$ = 0.1s., $\mu \approx 3.28\times10^{-2}$.} 
   \label{fig:numSamplesConvLorenzShort}
    \end{subfigure}    
    \caption{Comparison of performance of MC and BIMC. The variation of the 
        relative RMSE, $e_{\mathrm{RMS}}$, is plotted against the number of
        samples $N$. For reference, the most accurate probability estimate is
    also reported.} 
\label{fig:numSamplesConv}
\end{figure}

\paragraph{Effect of probability magnitude}
In \Cref{fig:probLevelConvRk1pert}, we study the effect of the probability
magnitude on the relative RMSE, $\tilde{e}_{\mathrm{RMS}}$. We notice that 
BIMC is only weakly dependent on the probability magnitude. This is because 
selecting parameters by minimizing $D_{\mathrm{KL}}$ leads to an IS density that 
is optimally adapted for sampling around $\DT$.

\begin{figure}[H]
    \centering
    \begin{subfigure}{0.4\textwidth}
        \includegraphics[width=\textwidth]{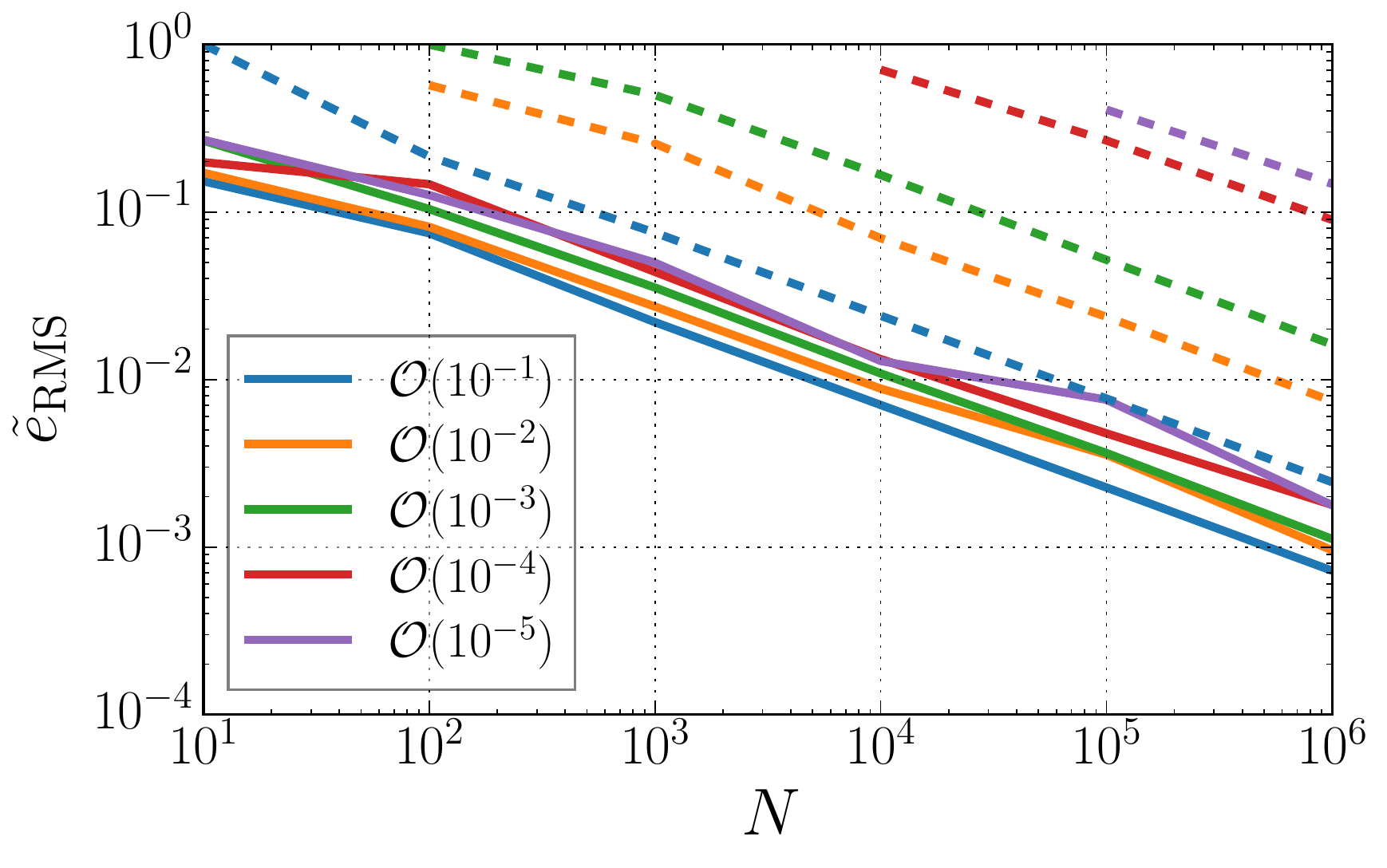} 
        \caption{Synthetic non-linear problem}
    \end{subfigure}
    \begin{subfigure}{0.4\textwidth}
        \includegraphics[width=\textwidth]{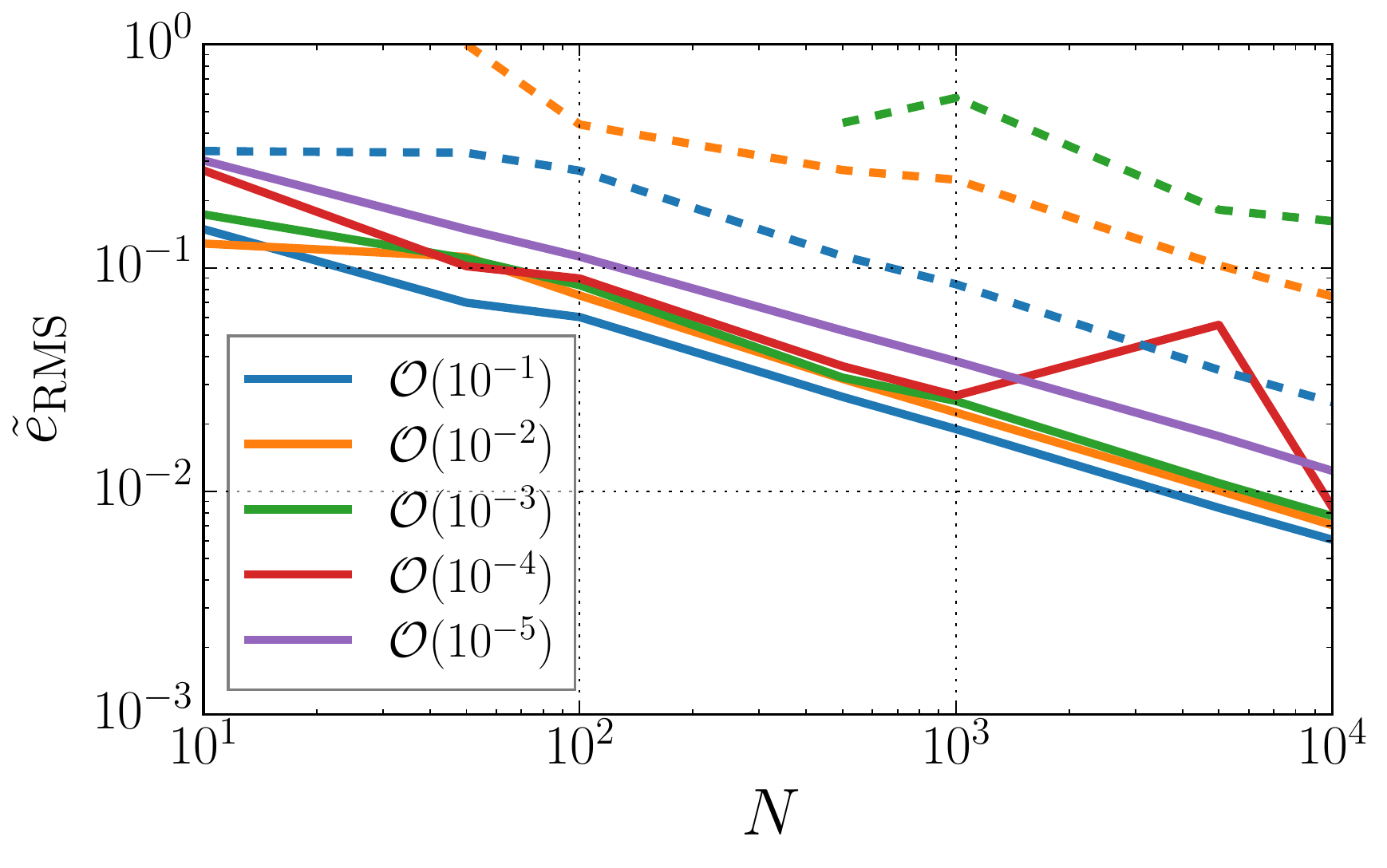} 
        \caption{Autoignition}
    \end{subfigure}
    \begin{subfigure}{0.4\textwidth}
        \includegraphics[width=\textwidth]{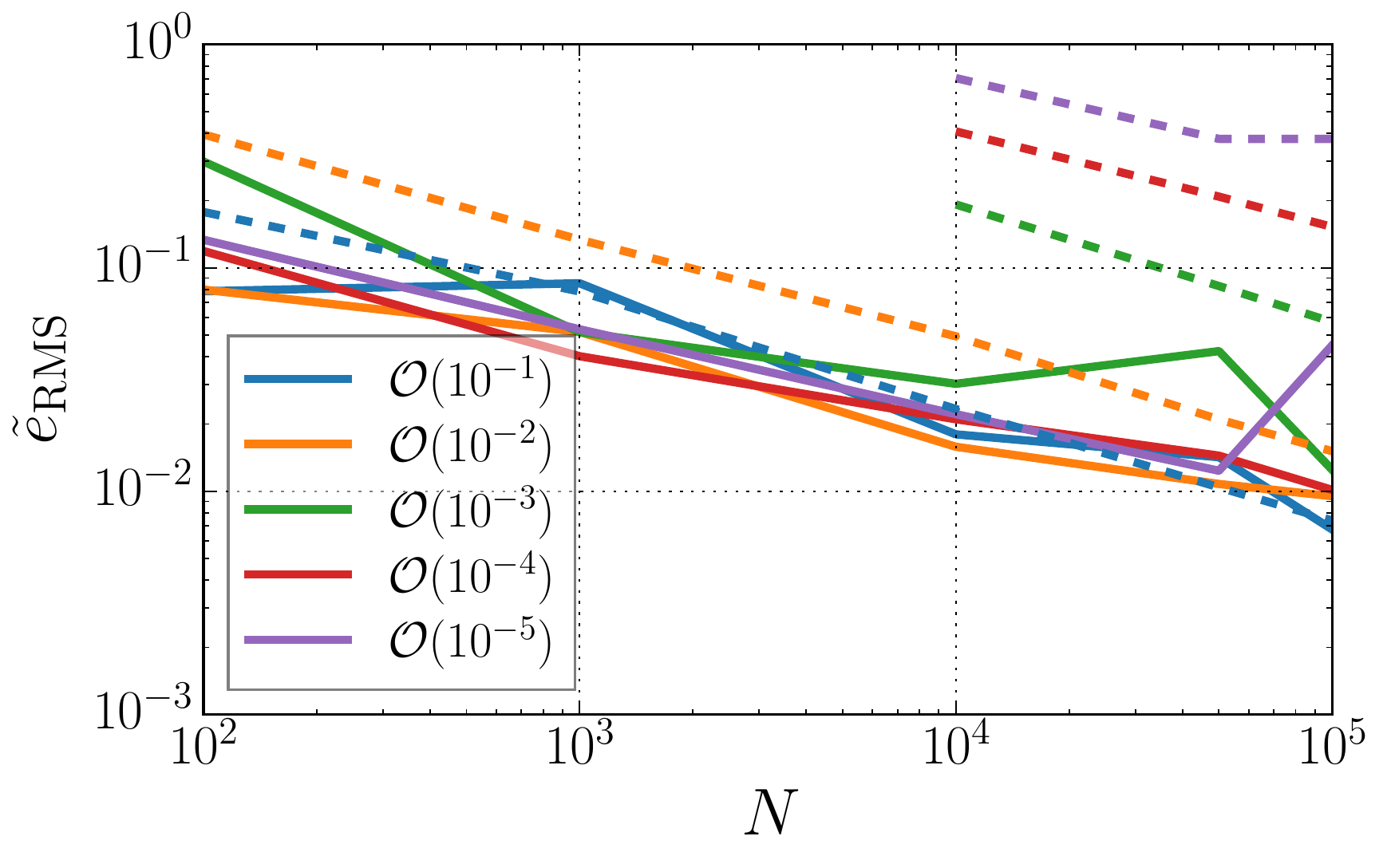} 
        \caption{Elliptic PDE}
    \end{subfigure}
\caption{Effect of varying probability levels. In this figure, we plot the
    variation of ensemble averaged relative $\tilde{e}_{\mathrm{RMS}}$ 
        with the number of samples $N$ 
        for various forward models with varying levels of probability.
        The varying probability levels are selected by moving $\DT$ to the tail
        regions of $p(\uq)$. The dashed lines indicate errors associated with MC
        while the solid lines indicate errors in BIMC. }
\label{fig:probLevelConvRk1pert}
\end{figure}

\paragraph{Extremely rare events}

In our final experiment, we push BIMC to compute probabilities of extremely rare
events. The rare events were
constructed by shifting $\DT$ further and further into the tail region of the
push forward of $p(\uq)$ under $f$. BIMC is able to compute
extremely small probabilities using a modest number of samples. This experiment
also corroborates our claim that the accuracy of our method is only weakly
dependent on the probability magnitude $\mu$. We also report the probability
estimate resulting from linearizing $f(\uq)$ around $\uq^{\mathrm{mid}}$ and
conclude that the linearized probability estimate is a good indicator of the
order of magnitude of the true probability. 

\begin{table}[H]
    \caption{Extremely rare events, $N$ = 1000.}
    \label{table:rareEvents}
    \begin{subtable}{\textwidth}
    \small
    \centering
    \begin{tabular}{c c c c c}
        \toprule
        \multicolumn{2}{c}{ BIMC, $n = 1$ } 
        & \multicolumn{2}{c}{ BIMC, $n = 5$ } 
        & Linearized\\
        \cmidrule(l{0.5em}r){1-2}
        \cmidrule(l{0.5em}r){3-4}
        \cmidrule(l{0.5em}r){5-5}
        $\tilde{\mu}^N$ & $\tilde{e}_{\mathrm{RMS}}^N$ & 
        $\tilde{\mu}^N$ & $\tilde{e}_{\mathrm{RMS}}^N$ & 
        $\mu_{\mathrm{lin}}$\\
        \cmidrule(l{0.5em}r){1-1}
        \cmidrule(l{0.5em}r){2-2}
        \cmidrule(l{0.5em}r){3-3}
        \cmidrule(l{0.5em}r){4-4}
        \cmidrule(l{0.5em}r){5-5}
        
        $3.6214 \times 10^{-3}$ & $3.24 \times 10^{-2}$
        & $3.7892 \times 10^{-3}$ & $3.87 \times 10^{-2}$
        & $5.9770 \times 10^{-3}$\\

        $2.3938 \times 10^{-6}$ & $6.00 \times 10^{-2}$
        & $2.1421 \times 10^{-6}$ & $6.10 \times 10^{-2}$
        & $1.8252 \times 10^{-6}$  \\
        
          $5.4224 \times 10^{-8}$ & $6.64 \times 10^{-2}$
        & $5.4310 \times 10^{-8}$ & $7.08 \times 10^{-2}$ 
        & $4.2072 \times 10^{-8}$\\

          $5.7578 \times 10^{-10}$ & $1.04 \times 10^{-1}$
        & $5.6271 \times 10^{-10}$ & $6.79 \times 10^{-2}$ 
        & $3.8026 \times 10^{-10}$\\
        
        \bottomrule
    \end{tabular}
    \caption{Synthetic non-linear problem}
    \label{table:rareRk1pert}
\end{subtable}
\begin{subtable}{\textwidth}
    \small
    \centering
    \begin{tabular}{c c c c c}
        \toprule
        \multicolumn{2}{c}{ $n = 1$ } 
        & \multicolumn{2}{c}{ $n = 5$ } 
        & Linearized\\
        \cmidrule(l{0.5em}r){1-2}
        \cmidrule(l{0.5em}r){3-4}
        \cmidrule(l{0.5em}r){5-5}

        $\tilde{\mu}^N$ & $\tilde{e}_{\mathrm{RMS}}^N$ & 
        $\tilde{\mu}^N$ & $\tilde{e}_{\mathrm{RMS}}^N$ & 
        $\mu_{\mathrm{lin}}$\\
        \cmidrule(l{0.5em}r){1-1}
        \cmidrule(l{0.5em}r){2-2}
        \cmidrule(l{0.5em}r){3-3}
        \cmidrule(l{0.5em}r){4-4}
        \cmidrule(l{0.5em}r){5-5}

        $4.3626 \times 10^{-3}$ & $2.47 \times 10^{-2}$
        & $4.3688 \times 10^{-3}$ & $3.52 \times 10^{-2}$
        & $4.5667 \times 10^{-3}$  \\

          $1.1158 \times 10^{-5}$ & $3.91 \times 10^{-2}$
        & $1.1278 \times 10^{-5}$ & $4.52 \times 10^{-2}$
        & $8.4646 \times 10^{-5}$  \\
        
          $7.6348 \times 10^{-7}$ & $5.86 \times 10^{-2}$
        & $8.0428 \times 10^{-7}$ & $7.76 \times 10^{-2}$ 
        & $3.8022 \times 10^{-6}$\\
        
          $3.5977 \times 10^{-10}$ & $9.69 \times 10^{-2}$
        & $3.8106 \times 10^{-10}$ & $1.54 \times 10^{-1}$ 
        & $2.6634 \times 10^{-10}$\\
        
        \bottomrule
    \end{tabular}
    \caption{Autoignition}
    \label{table:rareAuto}
\end{subtable}
\begin{subtable}{\textwidth}
    \small
    \centering
    \begin{tabular}{c c c c c}
        \toprule
        \multicolumn{2}{c}{ $n = 1$ } 
        & \multicolumn{2}{c}{ $n = 5$ } 
        & Linearized\\
        \cmidrule(l{0.5em}r){1-2}
        \cmidrule(l{0.5em}r){3-4}
        \cmidrule(l{0.5em}r){5-5}

        $\tilde{\mu}^N$ & $\tilde{e}_{\mathrm{RMS}}^N$ & 
        $\tilde{\mu}^N$ & $\tilde{e}_{\mathrm{RMS}}^N$ & 
        $\mu_{\mathrm{lin}}$\\
        \cmidrule(l{0.5em}r){1-1}
        \cmidrule(l{0.5em}r){2-2}
        \cmidrule(l{0.5em}r){3-3}
        \cmidrule(l{0.5em}r){4-4}
        \cmidrule(l{0.5em}r){5-5}
        
        $2.6422 \times 10^{-3}$ & $5.05 \times 10^{-2}$
        & $2.7045 \times 10^{-3}$ & $3.90 \times 10^{-2}$
        & $2.2526 \times 10^{-3}$\\

        $5.6726 \times 10^{-6}$ & $9.44 \times 10^{-2}$ 
        & $5.1764 \times 10^{-6}$ & $4.76 \times 10^{-2}$ 
        & $4.1409 \times 10^{-6}$ \\
        
          $8.4630 \times 10^{-9}$ & $4.94 \times 10^{-2}$ 
        & $8.6889 \times 10^{-9}$ & $5.58 \times 10^{-2}$ 
        & $8.7048 \times 10^{-9}$\\

          $8.2730 \times 10^{-10}$ & $4.99 \times 10^{-2}$ 
        & $8.0669 \times 10^{-10}$ & $7.36 \times 10^{-2}$
        & $9.1534 \times 10^{-10}$\\
        \bottomrule
    \end{tabular}
    \caption{Elliptic PDE}
    \label{table:rarePoisson}
\end{subtable}
\end{table}

\paragraph{Failure cases}

Here, we report cases which caused BIMC to fail.
\Cref{fig:failTaylorGreen} shows MC and BIMC samples for the periodic forward 
problem. Because $f(\uq)$ has circular contours, the ideal IS density $q^*$ has
support over a circular region in $\mathbb{R}^2$. This is also evident from how
the samples from $q^*$ are spread. Using a single Gaussian distribution to 
approximate this complicated density results in a poor fit, and hence, failure of the
BIMC method.  The nature of the poor fit is noteworthy. The IS density
approximates $q^*(\uq)$ well in the direction that is informed by the data. In
the directions orthogonal to this data-informed direction, it inherits the
covariance of $p(\uq)$, and as such, cannot approximate $q^*$ as it curves
around. 

Also, notice that the pre-image $f^{-1}(\DT)$ is the union of two disconnected
regions in parameter space. As a result, the ideal IS density, $q^*$, has two
modes, one near $[1, 1]^T$, and a weaker one near $[-1, 2.5]^T$. Which mode is
discovered depends on the initial guess provided to the numerical optimization
routine. Currently, there exists no robust mechanism in BIMC to discover all the
modes of $q^*$. This is also the cause of failure when the Lorenz
system is inverted over $t_f$ = 5s.

Another route to failure occurs if the optimal parameters based on an analysis
of the linearized inverse problem aren't appropriate for the full non-linear
problem. While we don't expect the two to be exactly equal, we implicitly assume
that they will be close enough, and serious problems may occur if they're not.
For instance, if the pseudo-likelihood variance from the linearized analysis is
much smaller than the (unknown) optimal pseudo-likelihood variance for the full
non-linear problem, then large IS weights may be observed, leading to biased
estimates of the failure probability. 

Finally, BIMC can also fail when the solution of the inverse problem cannot
be computed. This happens when the Lorenz problem is simulated over a much longer
time horizon, $t_f = 15$s. In this case, the optimizer failed to identify a
descent direction and converge to a
minimum. Physically, this happens because of the chaotic nature of the problem.
Since all trajectories of the Lorenz system eventually settle on the attractor,
going from a point on the attractor back in time is a highly ill-conditioned
problem.

\begin{figure}[H]
    \centering
    \includegraphics[width=0.8\textwidth]{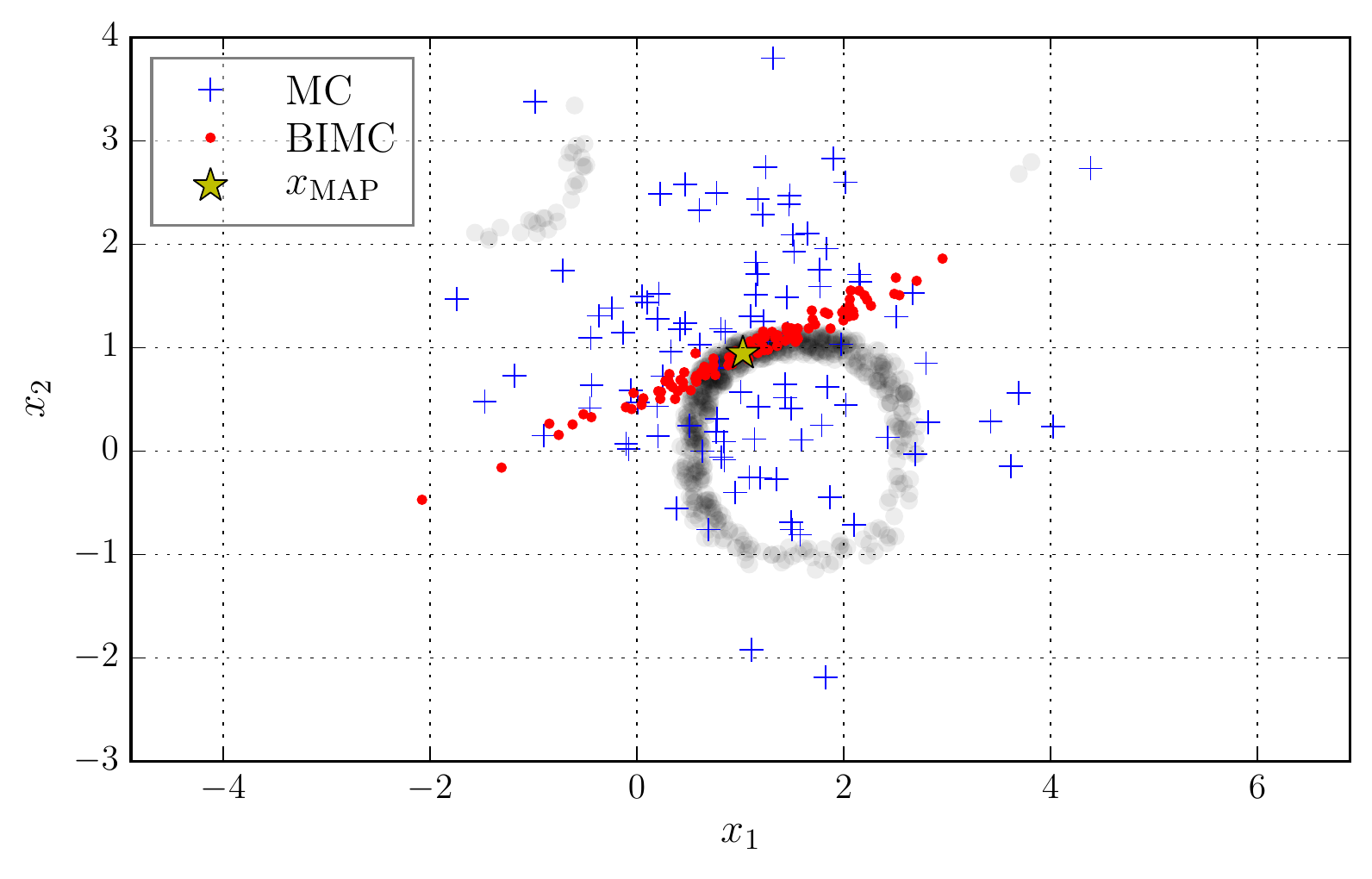} 
    \caption{A failure case. Here, $f(\uq) = \sin(x_1)\cos(x_2)$ is a periodic 
        function in $\mathbb{R}^2$. Gray markers depict samples from the ideal IS 
        density $q^*$ in this case.}
    \label{fig:failTaylorGreen}
\end{figure}

%

\paragraph{Summary}

In summary, the effectiveness of BIMC depends on the interplay between the
directions \emph{not} informed by the pseudo-data point, and the variation of 
the forward map in these directions. If, at the scale of the covariance of the
nominal density $p$,  $f(\uq)$ varies too quickly in these 
directions (like the Periodic example), the PDF constructed in BIMC 
will make for a poor IS density. On the other hand, if $f(\uq)$ varies slowly
enough (as in the synthetic non-linear, and autoignition examples) or not at all
(the affine case), then BIMC is effective. Thus, we conclude that BIMC is best
suited to forward maps that are weakly non-linear at the scale of the covariance
of the nominal density $p$. Physically, this means that the uncertainties in the
input parameters must small enough that $f$ appears almost linear. 
Note that $f$ can still be highly non-linear at
larger scales.

Apart from the forward map being only weakly non-linear, there are two 
additional requirements. The regions in parameter space that
evaluate inside $\DT$ should not be disjoint. The final and perhaps the most
important requirement is that the solution of the inverse problem must be
computable.

\section{Conclusion and future work}
\label{section:conclusion}


In this article, we addressed the problem of efficiently computing rare-event
probabilities in systems with uncertain input parameters. Our approach, called
BIMC, employs importance sampling in order to achieve efficiency. Noting the
structural similarity between the (theoretical) ideal importance sampling
density and the posterior distribution of a fictitious inference problem, our
importance sampling distribution is constructed by approximating such a
fictitious posterior via a Gaussian distribution. The approximation process
allows the incorporation of the derivatives of the input-output map into the
importance sampling distribution, which is how our scheme achieves parsimonious
sampling. Our theoretical analysis establishes that this procedure is optimal
in the setting where the input-output map is affine and the nominal density is
Gaussian. Hence, BIMC is best applied to maps that appear nearly affine at the
scale of the covariance of the nominal distribution. Our numerical experiments
support this conclusion and demonstrate that when this is the case, BIMC can
lead to speedups of several orders-of-magnitude. Experiments also reveal
several drawbacks in BIMC. We will concern ourselves with fixing these
drawbacks in part II of this paper.

\section*{Acknowledgments}
We would like to acknowledge Umberto Villa's assistance in setting up the
Elliptic PDE example. A conversation with Dr. Youssef Marzouk sparked the search
for an optimality result for BIMC.
\bibliographystyle{siamplain}
\bibliography{refs}
\end{document}


\maketitle

\section{Analysis of the affine-Gaussian case}
\label{section:supplementAffineGauss}

In this section, we offer detailed derivations of claims made in the main text. 
In what follows we assume:

\begin{itemize}
    \item $f(\uq)$ is affine and defined as $f(\uq) =
        \vect{v}^T\uq + \beta$ for some $\vect{v} \in \mathbb{R}^m$, $\beta
        \in \mathbb{R}$,
    \item the nominal density $p(\uq)$ is Gaussian with mean $\uq_0$ and
        covariance $\mathbf{\Sigma}_0$, $p(\uq) = \mathcal{N}\left(\uq_0,
        \mathbf{\Sigma}_0\right)$,
    \item the pseudo-data is some $y \in \DT$ and pseudo-likelihood density $p(y
        | \uq) = \mathcal{N}(f(\uq), \sigma^2)$.
\end{itemize}

\subsection{Kullback Leibler divergence}
\label{subsection:supplementKLDiv}

Here, we derive an analytical expression for the Kullback Leibler
divergence between the ideal IS distribution, $q^*(\uq)$, for this problem, and
the IS density, $q(\uq) = \mathcal{N}(\uqmap, \mathbf{H}_{\mathrm{GN}}^{-1})$.

By definition of the KL divergence, we have,
\begin{align}
    \begin{split}
        D_{\mathrm{KL}}\left(q^* ||  q\right) &= \int_{\mathbb{R}^{m}}
    \frac{\ind_{\DT}(f(\uq))p(\uq)}{\mu}\log\frac{\ind_{\DT}\left(f(\uq)\right) p(\uq)}{\mu q(\uq)} \mathrm{d} \uq\\
    &= \frac{1}{\mu}\int_{f^{-1}(\DT)} p(\uq) \Big( \log \frac{p(\uq)}{q(\uq)} - \log\mu\Big)\mathrm{d} \uq\\
    &= \frac{1}{\mu} \int_{f^{-1}(\DT)} p(\uq)
    \log\frac{p(\uq)}{q(\uq)}\mathrm{d} \uq - \log\mu .
    \end{split}
\end{align}
where $f^{-1}(\DT) 
= \{\uq \in \mathbb{R}^m: \vect{v}^T \uq + \beta \in \DT\}$.

Note that the pushforward density of $p(\uq)$ under $f(\uq)$ is also a Gaussian
with mean $\nu = \vect{v}^T\uq_0 + \beta$, and variance $\gamma^2 =
\vect{v}^T\mathbf{\Sigma}_0\vect{v}$. Let $\rho^2 = \sigma^2 / (\sigma^2 +
\gamma^2)$. It can be shown that 
\begin{align}
    \begin{split}
    \log\frac{p(\uq)}{q(\uq)} &= \frac{{(y - 
        f(\uq))}^2}{2\sigma^2} + \log \rho - \frac{{(y -
    \nu)}^2}{2(\sigma^2 + \gamma^2)} .
\end{split}
\end{align}

Therefore,
\begin{align}
    \begin{split}
        \int_{f^{-1}(\DT)} p(\uq) \log \frac{p(\uq)}{q(\uq)}\mathrm{d}\uq 
    &= \int_{f^{-1}(\DT)} \left(\frac{{\left(y - f(\uq)\right)}^2}{2\sigma^2} + \log\rho - \frac{{(y - \nu)}^2}{2(\sigma^2 + {\gamma}^2)}\right) p(\uq) \mathrm{d}\uq \\
    &= \int_{f^{-1}(\DT)} \frac{{\left(y -
    f(\uq)\right)}^2}{2\sigma^2}p(\uq)\mathrm{d}\uq + \left(\log\rho - \frac{{(y
- \nu)}^2}{2(\sigma^2 + {\gamma}^2)}\right)\mu .
 \end{split}
\end{align}

To continue further, we change variables to $z = f(\uq)$ and take advantage of
our knowledge of the probability density for $z$, $p_Z$. This probability
density is nothing but the push forward of $p(\uq)$ under $f$. Thus, $p_Z =
\mathcal{N}\left(\nu, \gamma^2\right)$. 

\begin{align}
    \begin{split}
        \int_{f^{-1}(\DT)} p(\uq) \log \frac{p(\uq)}{q(\uq)}\mathrm{d} \uq &=
    \int_{\DT} \frac{{\left(y - z\right)}^2}{2\sigma^2}p_Z(z) \mathrm{d} z +
    \left(\log\rho - \frac{{(y - \nu)}^2}{2(\sigma^2 + {\gamma}^2)}\right)\mu\\
                                                                           &=
    \frac{y^2}{2\sigma^2} \mu + \frac{1}{2\sigma^2} \int_{\DT} z^2 p_Z(z)
    \mathrm{d} z - \frac{y }{\sigma^2}\int_{\DT} z p_Z(z)\mathrm{d} z \\
    &+ \left(\log\rho - \frac{{(y - \nu)}^2}{2(\sigma^2 + {\gamma}^2)}\right)
    \mu .
\end{split}
\end{align}

Note that the two integrals in the final equation can be related to the mean 
and variance of a truncated normal distribution. Let $p_T(z)$ 
be the truncated distribution when $p_Z$ only has support on the interval $\DT$.
This truncated distribution is in fact the push forward under $f(\uq)$ of
$q^*(\uq)$. Thus, $p_T = {\ind_{\DT}(z) p_Z}/{\mu}$. We denote the mean and
variance of $p_T$ by $\nu_T$ and ${\gamma_T}^2$ respectively. 

Finally, 

\begin{align}
    \begin{split}
        \int_{f^{-1}({\DT})} p(\uq)\log\frac{p(\uq)}{q(\uq)}\mathrm{d} \uq &=
    \left(\frac{{(y - \nu_T)}^2 + {\gamma_T}^2}{2\sigma^2} + \log\rho -
        \frac{{(y - \nu)}^2}{2(\sigma^2 + {\gamma}^2)}\right)\mu. \\
    \end{split}
\end{align}

Computing $\nu_T$ and $\gamma_T$ is prone to catastrophic cancellation. 
We recommend using the scaled complementary error function, \texttt{erfcx}, for
this purpose, following~\cite{diaz2018moments}. The final expression for the KL divergence is, 

\begin{align}
    D_{\mathrm{KL}}(q^* || q) =  \frac{{\left(y - \nu_T\right)}^2 +
    {\gamma_T}^2}{2\sigma^2} - \frac{{(y - \nu)}^2}{2(\sigma^2 + {\gamma}^2)} +
    \log\rho - \log\mu.
    \label{dkl}
\end{align}

This expression is minimized when 

\begin{align}
    \begin{split}
        y^* &= \frac{\nu_T \gamma^2 - \nu \gamma_T^2}{\gamma^2 - \gamma_T^2}\\
        {\sigma^*}^2 &= \frac{\gamma_T^2\gamma^2}{\gamma^2 - \gamma_T^2}
    \end{split}
\end{align}

\subsection{Establishing BIMC optimality}
\label{supplement:bimc_optimality}

In this sub-section, we provide proof of \Cref{claim:bimc_optimality}.
We assume without loss of generality that $\uq_0 = \boldsymbol{0}$ and 
$\boldsymbol{\Sigma}_0 = \mathbf{I}$. Further, denote by $\hat{\vect{v}}
= \vect{v} / \|\vect{v}\|$. Now, select $\hat{\mathbf{V}}$ so that it's
columns form an orthonormal basis for 
$\mathbb{R}^m \textbackslash \,\text{span}(\hat{\vect{v}})$. Then, we have, for
any $\uq$:

\begin{align}
    \begin{split}
        \uq &= \hat{\vect{v}}\hat{\vect{v}}^T \uq + 
            \hat{\mathbf{V}}\hat{\mathbf{V}}^T\uq\\
            &= \hat{\boldsymbol{v}} x_1 + \hat{\mathbf{V}}\uq_{\perp}.\\
    \end{split}
\end{align}

where $x_1 := \hat{\boldsymbol{v}}^T\uq \in \mathbb{R}$ and 
$\uq_{\perp} := \hat{\mathbf{V}}^T\uq \in \mathbb{R}^{m - 1}$.

Using this decomposition of the parameter space, notice that 

\begin{align*}
    f(\uq) &= \vect{v}^T \uq + \beta\\
           &= \|\vect{v}\| x_1 + \beta\\
           &= \hat{f}(x_1),
\end{align*}

and, 

\begin{align}
    \label{eq:prior_decomp}
    p(\uq) &= p_1(x_1) p_{\perp}(\uq_{\perp}), 
\end{align}

where, $p_1(x_1) = \mathcal{N}\left(0, 1\right)$ and
$p_{\perp} = \mathcal{N}\left(\boldsymbol{0}_{\perp},
\mathbf{I}_{\perp}\right)$. By $\boldsymbol{0}_{\perp}$ and $\mathbf{I}_{\perp}$, 
we refer to the $m - 1$ dimensional zero vector and identity matrix respectively. 
Hence, such a decomposition of the parameter space
enables one to express ideal IS density, $q^*$, as 

\begin{align}
    \label{eq:ideal_is_decomp}
    \begin{split}
        q^*(\uq) &= \frac{\ind_{\DT}(f(\uq))p(\uq)}{\mu}\\
                 &= \frac{\ind_{\DT}(f(x_1))p_1(x_1)}{\mu}p_{\perp}(\uq_{\perp})\\
             &= q^*_1(x_1)p_{\perp}(\uq_{\perp}).
    \end{split}
\end{align}

Before we establish the optimality of the BIMC algorithm, we will require the
following result on the structure of the optimal Gaussian approximation of
$q^*(\uq_1, \uq_{\perp})$.

\begin{proposition}\label{prop:gaussian_structure}
    The Gaussian distribution closest in KL divergence to $q^*(x_1, \uq_{\perp}) =
    \ind_{\DT}(\hat{f}(x_1))p_1(x_1)p_{\perp}(\uq_{\perp})/\mu$ must be of the
    form $q_{1,\mathrm{opt}}(x_1)p_{\perp}(\uq_{\perp})$, 
    where $q_{1,\mathrm{opt}}(x_1)$ is the closest Gaussian
    approximation of $\ind_{\DT}(\hat{f}(x_1))p_1(x_1)  / \mu$.
\end{proposition}
\begin{proof}
    Let $\mathcal{G}^m$ denote the set of all Gaussian distributions over
    $\mathbb{R}^m$. Any Gaussian distribution in $\mathbb{R}^m$ can be expressed
    as a joint density over the variables $x_1, \uq_{\perp}$, $q(x_1,
    \uq_{\perp})$. Then, the closest Gaussian approximation of $q^*(x_1,
    \uq_{\perp})$ is given by:

    \begin{align*}
        q_{\mathrm{opt}}(x_1, \uq_{\perp}) 
        = \argmin_{q \in \mathcal{G}^m} 
            D_{\mathrm{KL}}\left(q^*(x_1, \uq_{\perp}) ||
                                 q(x_1, \uq_{\perp})\right).
    \end{align*}

    Using the definition of KL divergence, we have,

    \begin{align}
        D_{\mathrm{KL}}\left(q^*(x_1, \uq_{\perp})||q(x_1, \uq_{\perp})\right) 
        = \int q^*(x_1, \uq_{\perp}) \log \frac{q^*(x_1, \uq_{\perp})}{q(x_1, \uq_{\perp})}
        \mathrm{d}x_1\mathrm{d}\uq_{\perp}.
    \end{align}

    To simplify notation, denote $\ind_{\DT}(\hat{f}(x_1))p_1(x_1)/\mu$ with 
    $q^*_1(\uq_1)$. Now, using the fact that $q(x_1, \uq_{\perp}) =
    q_1(x_1)q_{\perp}(\uq_{\perp} | x_1)$, the chain rule of KL divergence
    gives:

    \begin{align*}
        D_{\mathrm{KL}}(q^*(x_1, \uq_{\perp}) || &q(x_1, \uq_{\perp})) \\
        &= D_{\mathrm{KL}}(q^*_1(x_1) || q_1(x_1)) +
        \mathbb{E}_{q^*_1}\left[D_{\mathrm{KL}}(p_{\perp}(\uq_{\perp}) ||
        q_{\perp}(\uq_{\perp} | x_1))\right].
    \end{align*}

    Then, because both terms are positive,

    \begin{align*}
        q_{\mathrm{opt}} &= \argmin_{q \in \mathcal{G}^m} D_{\mathrm{KL}}
                                (q^*(x_1, \uq_{\perp}) || q(x_1, \uq_{\perp}))\\
            &= \argmin_{q_1(x_1)q_{\perp}(\uq_{\perp}|x_1) \in \mathcal{G}^m}
                    D_{\mathrm{KL}}(q^*_1(x_1)||q_1(x_1)) \\
            & \quad + \argmin_{q_1(x_1)q_{\perp}(\uq_{\perp}|x_1) \in \mathcal{G}^m} 
            \mathbb{E}_{q^*_1}\left[D_{\mathrm{KL}}
                (p_{\perp}(\uq_{\perp})||q_{\perp}(\uq_{\perp}|x_1))\right].
    \end{align*}

    Since the KL divergence is a positive quantity,
    $\mathbb{E}_{q^*_1}\left[D_{\mathrm{KL}}
    (p_{\perp}(\uq_{\perp}) || q_{\perp}(\uq_{\perp} | x_1))\right] \ge 0$. 
    The equality is achieved iff $q_{\perp}(\uq_{\perp} | x_1) =
    p_{\perp}(\uq_{\perp})$. Hence,
    $D_{\mathrm{KL}}(p_{\perp}(\uq_{\perp})||q_{\perp}(\uq_{\perp}|x_1))$ is
    minimized when $q_{\perp}(\uq_{\perp}|x_1) = p_{\perp}(\uq_{\perp})$. By
    inspection, the first term is minimized by $q_{1,\mathrm{opt}}(x_1)$, the
    closest Gaussian approximation of $q^*_1(x_1)$. Thus, $q_{\mathrm{opt}} =
    q_{1,\mathrm{opt}}(x_1) p_{\perp}(\uq_{\perp})$.
\end{proof}

\Cref{prop:gaussian_structure} states that for the affine-Gaussian case, 
$q_{\mathrm{opt}}$ must marginalize to
the corresponding marginal of $p(\uq)$ 
in all directions in parameter space to which
$f(\uq)$ is insensitive. Note that in the proof, the affine property of $f(\uq)$ was
never explicitly used. Hence, this result can be easily extended to arbitrary
non-linear $f(\uq)$ in the following sense. If there are \emph{global}
directions in parameter space to which $f(\uq)$ is insensitive, an optimal
Gaussian or Gaussian mixture approximation of $q^*$ must marginalize to the
corresponding marginal of $p(\uq)$ in those directions.

Next, we prove \Cref{claim:bimc_optimality}.

\subsubsection{Proof of \Cref{claim:bimc_optimality}}
\begin{proof}
    We begin by showing that BIMC produces a Gaussian distribution that indeed
    satisfies the form prescribed by \cref{prop:gaussian_structure}. 
    Let the IS distribution produced by BIMC be denoted $q_{\mathrm{BIMC}}(x)$.
    Because $q_{\mathrm{BIMC}}$ is a Bayesian posterior distribution, it has the
    following form:
    
    \begin{align}
        q_{\mathrm{BIMC}}(\uq) = \frac{p(y | \uq) p(\uq)}{p(y)},
    \end{align}

    where,

    \begin{align*}
        p(y | \uq) &\propto \exp\left(-\frac{1}{2\sigma^2}(y -
        f(\uq))^2\right)\\
        &\propto \exp \left(-\frac{1}{2\sigma^2}(y - \hat{f}(x_1))^2\right)\\
        &\propto p(y | x_1).
    \end{align*}

    Also, since $p(\uq)$ is a standard Gaussian, it can be expressed as $p(\uq)
    = p_1(\uq_1)p_{\perp}(\uq_{\perp})$.

    Therefore,
    \begin{align}
        q_{\mathrm{BIMC}}(\uq) \propto p(y | x_1)p_1(x_1)p_{\perp}(\uq_{\perp}).
    \end{align}

    Hence, the marginal distribution of $\uq_{\perp}$ when $x \sim
    q_{\mathrm{BIMC}}(\uq)$  is $p_{\perp}(\uq_{\perp})$, as required by
    \cref{prop:gaussian_structure}. Additionally, since $\hat{f}(x_1)$ is affine,
    $p(y | x_1)p_1(x_1)$ is Gaussian. All that remains to be shown is that
    \cref{param_selection} results in parameters such that $p(y | x_1)p_1(x_1)$
    is the optimal Gaussian approximation of
    $q^*_1(x_1) = \ind_{\DT}(\hat{f})(\uq_1)p_1(x_1)/\mu$. 

    Now, notice that $q^*_1(x_1) = \ind_{\DT}(\hat{f}(x_1))p_1(x_1)/\mu$ is a truncated normal
    distribution. That is because

    \begin{align*}
        q^*_1(x_1) &= \ind_{\DT}(\hat{f}(x_1))p_1(x_1)/\mu\\
                   &= \ind_{[y_{\min}, y_{\max}]}(\|\vect{v}\|x_1 + \beta) p_1(x_1)/\mu\\
                   &= \ind_{[\frac{y_{\min} - \beta}{\|\vect{v}\|}, \frac{y_{\max} -
        \beta}{\|\vect{v}\|}]}(x_1)p_1(x_1)/\mu.
    \end{align*}

    Thus, $q_1^*$ is a standard normal distribution truncated over
    $\left[\frac{y_{\min} - \beta}{\|\vect{v}\|}, \frac{y_{\max} - \beta}{\|\vect{v}\|}\right]$. 
    Denote the mean and variance of $q^*_1$ by
    $\nu_T$ and $\gamma_T^2$. Then the optimal Gaussian approximation of $q^*_1$
    is $\mathcal{N}(\phi_T, \omega_T^2)$. Hence, we have to demonstrate that the
    minimizer of \cref{param_selection} are such that $p(y | x_1) p_1(x_1)$
    is the same as $\mathcal{N}(\phi_T, \omega_T^2)$. 

    The expression for $D_{\mathrm{KL}}(q^*(x_1, \uq_{\perp}) || q(x_1,
    \uq_{\perp})$ in this case is:

    \begin{align*}
        \begin{split}
        D_{\mathrm{KL}}(q^*(x_1, \uq_{\perp}) || q(x_1, \uq_{\perp})) 
        = \text{const. } &+ \frac{\|\vect{v}\|^2}{2\sigma^2}\left(\omega_T^2 + \left(\frac{y -
                \beta}{\|\vect{v}\|} - \phi_T\right)^2\right) \\
                &- \frac{1}{2}\frac{(y - \beta)^2}{\sigma^2 + \|\vect{v}\|^2} +
                \frac{1}{2}\log\frac{\sigma^2}{\sigma^2 + \|\vect{v}\|^2}.
        \end{split}
    \end{align*}

    This expression is minimized in the BIMC algorithm. The minimum occurs at

    \begin{align}
        \label{optimal_param_affine_gauss}
        \begin{split}
            y^* &= \frac{\|\vect{v}\|}{1 - \omega_T^2}\phi_T
        + \beta\\
        {\sigma^*}^2 &= \frac{\|\vect{v}\|^2\omega_T^2}{1
        - \omega_T^2}.
    \end{split}
    \end{align}

    Since $\omega_T^2 < 1$ (Remark 2.1 in~\cite{zidek2003uncertainty}), ${\sigma^*}^2$ is a valid variance for the likelihood
    distribution.

    Now, it can be shown that,

    \begin{align*}
        \frac{p(y^* | x_1) p_1(x_1)}{p(y^*)}
                              &=
        \mathcal{N}\left(\frac{\|\vect{v}\|}{\|\vect{v}\|^2 +
            {\sigma^*}^2}(y^* - \beta), \frac{{\sigma^*}^2}{{\sigma^*}^2 +
    \|\vect{v}\|^2}\right).
    \end{align*}
    Plugging in expressions for $y^*$ and ${\sigma^*}^2$ from
    \cref{optimal_param_affine_gauss}, we obtain,

    \begin{align*}
        \frac{p(y^* | x_1) p_1(x_1)}{p(y^*)}
                              &= \mathcal{N}(\phi_T, \omega_T^2).
    \end{align*}

    Therefore, $p(y^*| x_1)p_1(x_1)/p(y^*)$ is the Gaussian closest in KL
    divergence to $\ind_{\DT}(\hat{f}(\uq_1)p_1(\uq_1) / \mu$. Hence, by
    \cref{prop:gaussian_structure},
    $q_{\mathrm{BIMC}}(\uq) = p(y^* | x_1)p_1(x_1)p_{\perp}(\uq_{\perp})/p(y^*)$
    must be the closest Gaussian distribution to $q^*(\uq_r, \uq_{\perp})$
\end{proof}

\section{Implementation details}
\label{supplement:fwdModels}
\subsection{The affine case}

We construct a affine map from $\mathbb{R}^m$ to $\mathbb{R}$. The map is
defined as 

\begin{align}
    f(\uq) = \vect{o}^T\mathbf{A}\uq,
\end{align}

where, $\vect{o} = \left(\frac{1}{m}, \frac{1}{m}, \ldots,
\frac{1}{m}\right)^T \in \mathbb{R}^m$ is an observation operator and 

\begin{align}
    \mathbf{A} = 
    \begin{pmatrix}
        1      & 0           & 0           & \cdots & 0           & \\
        0      & \frac{1}{2} & 0           & \cdots & 0           & \\
        0      & 0           & \frac{1}{3} &        & \vdots      & \\
        \vdots & \vdots      &             & \ddots & 0           & \\
        0      & 0           & \cdots      & 0      & \frac{1}{m} & \\
    \end{pmatrix} .
\end{align}

\subsubsection{Implementation}

We assume $p(\uq)$ is a Gaussian distribution with mean $\uq_0$ and covariance
$\mathbf{\Sigma}_0$, where $\uq_0 = \left(1, 1, \ldots, 1\right)^T \in
\mathbb{R}^m$ and $\mathbf{\Sigma_0} = 0.1 \mathbf{I} \in \mathbb{R}^{m\times m}$. Here,
$\mathbf{I}$ is the $m$ dimensional identity matrix. We chose $\DT = 
\left[1.2803, 1.4571\right]$ when $m = 2$ and $\DT = \left[6.2 \times 10^{-2},
6.3\times10^{-3}\right]$ when $m = 100$. The MAP point was computed using
MATLAB's \texttt{fminunc} routine.

\subsection{Synthetic non-linear case}

For quick testing, we construct the following non-linear problem from 
$\mathbb{R}^m$ to $\mathbb{R}$:
\begin{align}
    f(\uq) &= \vect{o}^T\vect{u},\,\text{where,}\\
    \left(\mathbf{S} + \varepsilon \uq\uq^T\right)\vect{u} &= \vect{b} .
\end{align}

\subsubsection{Implementation}

Again, $\vect{o} \in \mathbb{R}^m$ is an observation operator. For this
problem, we chose $\vect{o} = \left[1, 0, \ldots, 0\right]^T$. $\mathbf{S}$ is
a randomly chosen symmetric positive definite matrix, while $\vect{b}$ is a 
randomly chosen vector whose entries are distributed according to the standard 
normal distribution. We set $\varepsilon = 0.01\|\mathbf{S}\|_2$. The nominal
probability density in this case is $p(\uq)$ is a Gaussian with mean $\uq_0 =
\left(1, 1, \ldots, 1\right)^T$ and covariance $\mathbf{\Sigma_0} =
0.01\mathbf{I}$. The target interval $\DT$ is chosen to be 
$\left[1.24, 1.25\right]$ when $m = 10$ and $\left[0.919, 0.923\right]$ 
when $m = 2$. MATLAB's \texttt{fminunc} routine was
used for optimization again.

\subsection{Single step reaction}

The single step reaction is described by the following ODE:

\begin{align}
  \begin{split}
      &\frac{\mathrm{d}u}{\mathrm{d}t} = \frac{S^{*}(u)}{\tau_\mathrm{R}},\quad
      0 < t < t_f,\\
      &S^{*}(u) =  B u(1 - u)\exp\left(\frac{-T_{\mathrm{Act}}}{T_\mathrm{u} + (T_\mathrm{b} -
  T_\mathrm{u}) u}\right),
  \end{split}
  \label{1dcombRxn}
\end{align}

and we define $x = u(0)$ and $f(x) = u(t_f)$.

This equation uses the Arrhenius equation to describe the rate of a chemical 
reaction in terms of a progress variable $u$. The progress variable is 
routinely employed in the analyses of turbulent flames and is $0$ in regions 
of pure reactants and $1$ in pure products \cite{pope1985flamelet}.

Here, $\tau_{\mathrm{R}}$ is the time scale of the reaction and $S^*(u)$ is the normalized 
source term. The numerical constant in $S^*(u)$, $B$, ensures that it 
integrates to unity, $T_{\mathrm{Act}}$ is the activation temperature,
$T_\mathrm{u}$ is the temperature of the unburned reactants and $T_\mathrm{b}$ 
the temperature of the burnt products.

Since $u$ is always bounded between $0$ and $1$, the differential operator defined 
in \cref{1dcombRxn} is a map from $[0, 1]$ to $[0, 1]$. 

\subsubsection{Implementation}

For our numerical experiments, we set, $T_{\mathrm{u}}$ = 300 K,
$T_{\mathrm{b}}$ = 2100 K, $T_{\mathrm{Act}}$ = 30,000 K, $B = 6.11 \times
10^7$, $\tau_{\mathrm{R}} = 1$s. 
 
Further, we choose $\DT = \left[0.7, 0.8\right]$ and $p(x)$ = $\mathcal{N}(0.5,
0.01)$. We used MATLAB's \texttt{fmincon} routine to perform constrained
optimization in order to compute the MAP point.

\subsection{Hydrogen Autoignition}
\label{supplement:h2Autoignition}
We observe the heat released, $Q$, during autoignition of a hydrogen-air 
mixture in an adiabatic, constant pressure, fixed mass reactor. To describe 
the chemistry, we use a reduced mechanism that involves 5 elementary reactions 
among 8 chemical species - $\ce{H2}, \ce{O2}, \ce{N2}, \ce{HO2}, \ce{H}, 
\ce{O}, \ce{OH}, \ce{H2O}$ \cite{williams2008detailed}. We assume- 

\begin{itemize}
  \item reactants are ideal gases,
  \item there are no spatial gradients of temperature or species concentrations,
  \item the volume of the reactor can change to keep the pressure constant,
  \item only $\ce{H2}, \ce{O2}$, and \ce{N2} are present in the reactor initially.
\end{itemize}

Then specifying the pressure ($P$), temperature ($T$), and equivalence ratio 
($\phi$, defined as $\frac{[\ce{H2}]}{2[\ce{O2}]}$) is sufficient to completely 
define the initial state of the system. It is this triad that we define as our 
parameter vector, $\uq = \left(\phi, T, P\right)^T$. As stated 
earlier, the observable is the total heat released $Q$. 

\subsubsection{Implementation}
We assume $p(\uq)$ is a Gaussian with mean $\uq_0$ and covariance $\Sigma_0$,
where, 
\begin{align}
    \begin{split}
    \uq_0 &= \begin{bmatrix}
            1\\1500\\1.01325
          \end{bmatrix}, \\
    \Sigma &= \begin{bmatrix}
                     0.01 & 0     & 0\\
                     0    & 15.0  & 0\\
                     0    & 0     & 0.00101
              \end{bmatrix} .
    \end{split}
\end{align}
In addition, we select $\DT = [21000, 22000]$. The initial volume of the 
reactor is set to $V_0 = 1 \mathrm{m^3}$. We use MATLAB's inbuilt
\texttt{fminunc} algorithm to perform the non-linear optimization.

\subsubsection{Summary of equations}

\begin{table}[h]
\small
\caption{Reduced chemistry for hydrogen autoignition. Reaction rate constant 
$k = AT^b e^{-E/RT}$. Units are mol, cm, s, K, kJ. Chaperone efficiencies are 
$2.5$ for \ce{H2}, $16.0$ for \ce{H2O} and $1.0$ for all other species. 
Troe falloff with $F_{cent} = 0.5$ is assumed for reaction 5.}
\label{table:rxnMech}
\centering
\begin{tabular}[c]{l l l l l l}
\toprule
No. & Reaction &  & $A$ & $b$ & $E$\\
\midrule
1 & \ce{H2 + O2 -> H + HO2} &  & $2.69 \times 10^{12}$ & 0.36 & 231.86\\
2 & \ce{H + O2 -> OH + O}   &  & $3.52 \times 10^{16}$ & -0.7 & 71.4\\
3 & \ce{O + H2 -> H + OH}   &  & $5.06 \times 10^{4}$ & 2.7 & 26.3\\
4 & \ce{OH + H2 -> H + H2O} &  & $1.17 \times 10^{9}$ & 1.3 & 15.2\\
\multirow{2}{*}{5} & \multirow{2}{*}{\ce{H + O2 + M -> HO2 + M}} & $k_0$ & $5.75 \times 10^{19}$ & -1.4 & 0.0\\
& & $k_{\infty}$ & $4.65 \times 10^{12}$ & 0.4 & 0.0\\
\bottomrule
\end{tabular}
\end{table}

Williams (\cite{williams2008detailed}) identified a set of 5 elementary steps to study autoiginition of hydrogen-air mixtures. These elementary steps are given in \cref{table:rxnMech}. 

The net rates of production of each species, $\dot{S}_i$,  are given below:

\begin{align}
  \begin{split}
    &\dot{S}_{\ce{H2}} = -\dot{\omega}_1 - \dot{\omega}_3 - \dot{\omega}_4,\\
    &\dot{S}_{\ce{O2}} = -\dot{\omega}_1 - \dot{\omega}_2 - \dot{\omega}_5,\\
    &\dot{S}_{\ce{O}} = \dot{\omega}_2 - \dot{\omega}_3,\\
    &\dot{S}_{\ce{H}} = \dot{\omega}_1 - \dot{\omega}_2 + \dot{\omega}_3 + \dot{\omega}_4 - \dot{\omega}_5,\\
    &\dot{S}_{\ce{OH}} = \dot{\omega}_2 + \dot{\omega}_3 - \dot{\omega}_4,\\
    &\dot{S}_{\ce{HO2}} = \dot{\omega}_1 + \dot{\omega}_5,\\
    &\dot{S}_{\ce{H2O}} = \dot{\omega}_4\\
  &\dot{S}_{\ce{N2}} = 0,
  \end{split}
\end{align}

where $\omega_i$ is the rate of the $i$th reaction in \cref{table:rxnMech}:

\begin{align}
  \begin{split}
    &\dot{\omega}_1 = k_{f,1} [H_2][O_2]\\
    &\dot{\omega}_2 = k_{f,2} [H][O_2]\\
    &\dot{\omega}_3 = k_{f,3} [O][H_2]\\
    &\dot{\omega}_4 = k_{f,4} [OH][H_2]\\
    &\dot{\omega}_5 = k_{f,5} [H][O_2].\\
  \end{split}
\end{align}

Here, 

\begin{align}
  k_{f,i} = A^{(i)} T^{b^{(i)}} e^{-E^{(i)}/RT},\, i = 1\dots4,
\end{align}

$R$, the universal gas constant has the value 8.314 J/K/mol. $A^{(i)}$, $b^{(i)}$ and $E^{(i)}$ are as specified in \cref{table:rxnMech}. $k_{f,5}$ is calculated using the Lindemann form as:

\begin{align}
    \begin{split}
  &P_r = \frac{{k_0}[M]}{k_{\infty}},\\
  &[M] = 2.5 [H_2] + [O_2] + [H] + [O] + [OH] + [HO_2] + 16 [H_2O] + [N_2],\\
  &k_{f,5} = k_{\infty}\Big(\frac{P_r}{1 + P_r}\Big),\\
  &k_{0} = A^{(0)} T^{b^{(0)}} e^{-E^{(0)}/RT},\\
  &k_{\infty} = A^{(\infty)} T^{b^{(\infty)}} e^{-E^{(\infty)}/RT},
    \end{split}
\end{align}

with the Arrhenius parameters as in \cref{table:rxnMech}.

The molar enthalpies ($\bar{h}$) and specific heats ($\bar{c}_p$) for each species are given by:

\begin{align}
    \begin{split}
&\frac{\bar{c}_{p}(T)}{R} = a_0 + a_1 T + a_2 T^2 + a_3 T^3 + a_4 T^4,\\
&\frac{\bar{h}(T)}{RT} = a_0 + \frac{a_1}{2} T + \frac{a_2}{3}T^2 + \frac{a_3}{4}T^3 + \frac{a_4}{5}T^4 + \frac{a_5}{T}.
    \end{split}
\end{align}

where the coefficients $a_0 \dots a_5$ are tabulated in ref.~\cite{gri_mech} for each species.

The rate of heat release per unit volume is:

\begin{align}
    \Delta \dot{q} = &- \bar{h}_{\ce{H2}} \dot{S}_{\ce{H_2}} - \bar{h}_{\ce{O2}} \dot{S}_{\ce{O2}}
                    - \bar{h}_{\ce{H}} \dot{S}_{\ce{H}}
                    - \bar{h}_{\ce{O}} \dot{S}_{\ce{O}} 
                    - \bar{h}_{\ce{OH}} \dot{S}_{\ce{OH}}\\
                    &- \bar{h}_{\ce{HO2}} \dot{S}_{\ce{HO2}} 
                    - \bar{h}_{\ce{H2O}} \dot{S}_{\ce{H2O}} 
                    - \bar{h}_{\ce{N2}} \dot{S}_{\ce{N2}} .
\end{align}

The following ODEs describe the chemistry inside a constant pressure adiabatic reactor:

\begin{align*}
    &\frac{\mathrm{d}{T}}{\mathrm{d}t} = \frac{\Delta\dot{q}}{\beta},\\
  &\frac{\mathrm{d}{[H_2]}}{\mathrm{d}t} = \dot{S}_{\ce{H2}} -\alpha [H_2], \\
  &\frac{\mathrm{d}{[O_2]}}{\mathrm{d}t} = \dot{S}_{\ce{O2}} -\alpha [O_2], \\
  &\frac{\mathrm{d}{[H]}}{\mathrm{d}t} = \dot{S}_{\ce{H}} -\alpha [H], \\
  &\frac{\mathrm{d}{[O]}}{\mathrm{d}t} = \dot{S}_{\ce{O}} -\alpha [O], \\
  &\frac{\mathrm{d}{[OH]}}{\mathrm{d}t} = \dot{S}_{\ce{OH}} -\alpha [OH], \\
  &\frac{\mathrm{d}{[HO_2]}}{\mathrm{d}t} = \dot{S}_{\ce{HO2}} -\alpha [HO_2], \\
  &\frac{\mathrm{d}{[H_2O]}}{\mathrm{d}t} = \dot{S}_{\ce{H2O}} -\alpha [H_2O], \\
  &\frac{\mathrm{d}{[N_2]}}{\mathrm{d}t} = \dot{S}_{\ce{N2}} -\alpha [N_2], \\
  &\frac{1}{V}\frac{\mathrm{d}{V}}{\mathrm{d}t} = \alpha,\\
  &\alpha = \Bigg(\frac{\dot{S}_{\ce{H2}} + 
                      \dot{S}_{\ce{O2}}
                     + \dot{S}_{\ce{H}}
                     + \dot{S}_{\ce{O}}
                     + \dot{S}_{\ce{OH}}
                     + \dot{S}_{\ce{HO2}}
                     + \dot{S}_{\ce{H2O}}
                     + \dot{S}_{\ce{N2}}}
                     {[H_2] + [O_2] + [H] + [O] + [OH] + [HO_2] + [H_2O] +
                 [N_2]} + \frac{1}{T} \frac{\mathrm{d}T}{\mathrm{d}t}\Bigg),\\
  &\beta = [H_2]\bar{c}_{p, \ce{H2}}
                     + [O_2]\bar{c}_{p, \ce{O2}}
                     + [H]\bar{c}_{p, \ce{H}}
                     + [O]\bar{c}_{p, \ce{O}}
                     + [OH]\bar{c}_{p, \ce{OH}}\\
                     &+ [HO_2]\bar{c}_{p, \ce{HO2}}
                     + [H_2O]\bar{c}_{p, \ce{H2O}}
                     + [N_2]\bar{c}_{p, \ce{N2}} ,
\end{align*}
The net heat released is 

\begin{align}
  Q = \int_{0}^{t_f} \Delta \dot{q} V(t) \mathrm{d}t.
\end{align}

\subsection{The Lorenz system}

The Lorenz system is defined by the following ordinary differential equations: 
\begin{align} 
    \label{lorenzODE} 
    \begin{split} 
        &\frac{\mathrm{d}{u_1(t)}}{\mathrm{d}t} = s(u_2 - u_1),\quad 0 < t < t_f,\\ 
        &\frac{\mathrm{d}{u_2(t)}}{\mathrm{d}t} = u_1 (r - u_3) - u_2,\\
        &\frac{\mathrm{d}{u_3(t)}}{\mathrm{d}t} = u_2 u_2 - b u_3,
  \end{split}
\end{align}

The parameter vector is the initial condition of the system, 
$\uq = \vect{u}(0)$ while the observable is $u_1(t_{f})$.

We set $s = 10$, $r = 28$, $b = 8/3$ and observe $u_1$ after two time 
horizons, $t_f = 0.1$s and $t_f = 5$s. The nominal density $p(\uq)$
is selected to be a Gaussian with mean $\uq_0$ and covariance 
$\mathbf{\Sigma}_0$ where, 

\begin{align}
    \begin{split}
\uq_0 &= \begin{bmatrix}1.508870\\-1.531271\\25.46091\end{bmatrix},\,\text{and,}\\
    \mathbf{\Sigma}_0 &= \begin{bmatrix}
                     0.01508870 & 0     & 0\\
                     0    & 0.01531271 & 0\\
                     0    & 0     & 0.02546091
                  \end{bmatrix}.
    \end{split}
\end{align}

The target intervals are chosen to be $\DT = [-5, -4]$ when $t_f = 0.1$s and 
$\DT = [-0.22, -0.21]$ when $t_f = 5$s.
 
 \subsubsection{Implementation}
This system is chaotic with maximal Lyapunov exponent $\lambda \approx 0.906$. 
We perform the optimization with MATLAB's inbuilt \texttt{fminunc} algorithm 
using analytically derived gradients.

  \subsection{Elliptic PDE}
In this experiment, we invert for the log permeability field, $g$ in the following elliptic PDE:

\begin{align}
    \begin{split}
- \nabla \cdot (e^{g} \nabla u) &= h\,\text{in}\,\Omega,\\
u &= u_D\,\text{on}\,\partial\Omega_{D},\\
e^{g}\nabla{u}\cdot\mb{n} &= u_N\,\text{on}\,\partial\Omega_N,
\end{split}
  \label{poissonEq}
\end{align}

where $\Omega \subset \mathbb{R}^{2}$ is an open domain with boundary 
$ \partial\Omega = \partial\Omega_D \cup \partial\Omega_N$, 
$\partial\Omega_D \cap \partial\Omega_N = \emptyset$. $\partial\Omega_D$ and 
$\partial \Omega_N$ denote Dirichlet and Neumann type boundaries with 
boundary values $u_D$ and $u_N$ respectively. $\mb{n}$ is a unit vector 
normal to $\partial \Omega$ in the outward direction and $h \in L^2(\Omega)$ 
is the source term.

We assume $\Omega$ is a unit square, there is no source term, the left and right walls are no-flux 
boundaries, and the top and bottom walls are Dirichlet boundaries. That is,

\begin{align}
    \begin{split}
  \Omega &= [0, 1] \times [0, 1],\\
  h &= 0,\\
  \partial\Omega_D &= (0, 1) \times \{1\} \cup (0, 1) \times \{0\}, \\
  \partial\Omega_N &= \{0, 1\} \times (0, 1), \\
  u_N &= 0,\\
  u_D &= 
    \begin{cases}
      1, &\quad \mb{x} \in (0, 1) \times \{1\}\\
      0, &\quad \mb{x} \in (0, 1) \times \{0\}
    \end{cases} .
\end{split}
\end{align}

This is an instance of Bayesian inference in infinite dimensions. The problem
can be reduced to finite dimensions by using, for instance, a finite element
discretization. The inference is then performed for the vector of coefficients
of the finite element basis functions chosen. Here, we use first order
Lagrange basis functions with 4225 degrees of freedom. Thus, the parameter
vector then is the vector of coefficients $\uq = (g_1, g_2, \ldots, g_m) \in
\mathbb{R}^m, m = 4225$. We define the parameter-to-observable map as the fluid 
velocity at a particular location in the domain $\Omega$, 
$f (\uq) = u(0.1, 0.5)$.

To solve the inverse problem, we use \texttt{hIPPYlib} \cite{villa2018hippylib,
villa2019hippylib}. 
\texttt{hIPPYlib} is a scalable software framework to  solve large scale PDE 
constrained inverse problems. It relies on \texttt{FEniCS} for the discretization and 
solution of the PDE and \texttt{PETSc} for efficient implementation of linear algebra 
routines. \texttt{hIPPYlib} provides state-of-the-art algorithms for PDE constrained 
optimization, including an implementation of the Inexact Newton-CG algorithm 
for computing the MAP point as well as randomized algorithms for constructing a 
low rank approximation of the Hessian at the MAP point. For full details, we 
refer the reader to \cite{villa2018hippylib, villa2019hippylib}. We would like to remark here that
this inverse problem appears as a model problem in \cite{villa2019hippylib} and
has been slightly modified for our experiments. For completeness, we reproduce 
relevant details from here.

While constructing $p(\uq)$, it was assumed that the true log-permeability at 5
locations in $\Omega = \left[0, 1\right] \times \left[0, 1\right],
\boldsymbol{\omega}_1 = (0.1, 0.1), \boldsymbol{\omega}_2 = (0.1, 0.9), 
\boldsymbol{\omega}_3 = (0.5, 0.5), \boldsymbol{\omega}_4 =
(0.9, 0.1), \boldsymbol{\omega}_5 = (0.9, 0.9)$, is known. Let the true log-permeability
at these points be $x_{\mathrm{true}}^1, x_{\mathrm{true}}^2, \ldots,
x_{\mathrm{true}}^5$. In addition, the following mollifier functions were 
defined

\begin{align}
    \begin{split}
    \delta_i(\boldsymbol{\omega}) =
    \exp\left(-\frac{\gamma^2}{\delta^2}\|\boldsymbol{\omega} -
    \boldsymbol{\omega}_i\|_{\mathbf{\Theta}^{-1}}\right) .
\end{split}
\end{align}

where $\mathbf{\Theta}$ is an anisotropic symmetric positive definite tensor of
the form 

\begin{align}
    \mathbf{\Theta} = 
    \begin{pmatrix}
        \theta_1 \sin^2\alpha & (\theta_1 - \theta_2) \sin\alpha \cos\alpha\\
        (\theta_1 - \theta_2) \sin\alpha\cos\alpha & \theta_2 \cos^2\alpha\\
    \end{pmatrix} .
\end{align}
The various parameters were set to $\gamma = 0.1, \delta = 0.5, \alpha = \pi/4, \theta_1 = 2, \theta_2 =
0.5$. The covariance of $p(\uq)$ was finally defined as $\mathbf{\Sigma}_0 =
\mathcal{A}^{-2}$, where, 

\begin{align}
    \mathcal{A} &= \tilde{\mathcal{A}} + p \sum_{i = 1}^{5}\delta_i\mathbf{I},\\
                &= \tilde{\mathcal{A}} + p \mathcal{M}.
\end{align}

where $\tilde{A}$ is a differential operator of the form
$\gamma\nabla\cdot\left(\theta\nabla\right) + \delta \mathbf{I}$ and $p$ is a
penalization parameter, which was set to $p = 10$.

The mean of the nominal PDF $p(\uq)$, $\uq_0$, was set to be the solution of the following  
regularized least squares problem 

\begin{align}
    \uq_0 = \argmin_{\uq} =
    \frac{1}{2}\langle\uq,\uq\rangle_{\tilde{\mathcal{A}}} +
    \frac{p}{2}\langle\uq_{\mathrm{true}} - \uq, \uq_{\mathrm{true}} -
    \uq\rangle_{\mathcal{M}} .
\end{align}

The nominal distribution $p(\uq)$ is then defined to be $\mathcal{N}\left(\uq_0,
\mathbf{\Sigma}_0\right)$ and $\DT = \left[0.6, 0.7\right]$.

\subsection{Periodic map}

In this case the input-output map is defined to be:

\begin{align}
    f(\uq) = \sin(x_1)\cos(x_2) .
\end{align}

\subsection{Implementation}

Again, $p(\uq)$ is assumed to be a Gaussian distribution 
with mean $\uq_0$ and covariance $\mathbf{\Sigma}_0$, 
where $\uq_0 = \left(1, 1\right)^T$ 
and $\mathbf{\Sigma_0} = \mathbf{I}$. We chose $\DT = 
\left[0.4, 0.6\right]$. The MAP point was computed using
MATLAB's \texttt{fminunc} routine.

\section{Results}
\label{supplement:results}   

Here we report the probability estimates corresponding to \Cref{fig:numSamplesConv}.
\begin{table}[H]
\small
    \centering
\caption{Single step reaction}
\begin{filecontents*}{./tables/numSamplesConv1DCombPGF.dat}
1.000000000000000000e+01 0.000000000000000000e+00 1.722726477053999894e-02 2.583712735792999954e-02 3.896463185777999938e-02 2.206416283569000159e-02
1.000000000000000000e+02 5.000000000000000278e-02 2.267253623262999868e-02 2.059819503050000153e-02 2.067945177034000159e-02 1.920816536195999943e-02
1.000000000000000000e+03 1.499999999999999944e-02 2.267497920428999947e-02 2.201587470922999890e-02 2.344830747742999988e-02 2.250492985892000100e-02
1.000000000000000000e+04 2.489999999999999852e-02 2.309118606080000011e-02 2.329417970448000030e-02 2.308915377107000066e-02 2.309347423599000110e-02
1.000000000000000000e+05 2.417000000000000051e-02 2.303197703458999962e-02 2.301099218058999876e-02 2.309021763378000039e-02 2.312670254037000120e-02
1.000000000000000000e+06 2.312099999999999919e-02 2.303355006169999913e-02 2.303296864024999974e-02 2.300984887732000078e-02 2.303980342798999867e-02
\end{filecontents*}
\pgfplotstabletypeset[
    every head row/.style={before row=\toprule,after row=\midrule},
    every last row/.style={after row=\bottomrule},
    columns/0/.style={column name=$N$},
    columns/1/.style={column name=MC, 
        int detect, sci, precision=4},
    columns/2/.style={column name=BIMC{,} $n {=} 1$,
        sci, precision=4, zerofill},
    columns/3/.style={column name=BIMC{,} $n {=} 5$, 
        sci, precision=4, zerofill},
    columns/4/.style={column name=BIMC{,} $n {=} 10$, 
        sci, precision=4, zerofill},
    columns/5/.style={column name=BIMC{,} $n {=} 25$, 
        sci, precision=4, zerofill},
    col sep=space
]{./tables/numSamplesConv1DCombPGF.dat}
\end{table}
\begin{table}[H]
\small
    \centering
\caption{Autoignition}
\begin{filecontents*}{./tables/numSamplesConvAutoPGF.dat}
1.000000000000000000e+01 0.000000000000000000e+00 3.763154177817999713e-02 3.816739556944000189e-02 3.046170970078000001e-02 2.755639536147000065e-02
5.000000000000000000e+01 2.000000000000000042e-02 3.022797749139000123e-02 2.726551503945999835e-02 2.829410768031000167e-02 3.454241719197000127e-02
1.000000000000000000e+02 5.999999999999999778e-02 3.413819039389000276e-02 3.187043349639000211e-02 3.065111808754000106e-02 3.050744698390000023e-02
5.000000000000000000e+02 4.599999999999999922e-02 3.126387028672000107e-02 3.313403732867999801e-02 3.395486315930999754e-02 3.183190759142000281e-02
1.000000000000000000e+03 4.100000000000000172e-02 3.245910189994000161e-02 3.264166110122999898e-02 3.234928638565000292e-02 3.338889922285000200e-02
5.000000000000000000e+03 3.099999999999999978e-02 3.230232449046999826e-02 3.216179447025999810e-02 3.221496464752000161e-02 3.184905121531000222e-02
1.000000000000000000e+04 3.239999999999999825e-02 3.238205141971999684e-02 3.243029503910999783e-02 3.246358676748999778e-02 3.258599300860000325e-02
\end{filecontents*}
\pgfplotstabletypeset[
    every head row/.style={before row=\toprule,after row=\midrule},
    every last row/.style={after row=\bottomrule},
    skip rows between index={1}{2},
    skip rows between index={3}{4},
    skip rows between index={5}{6},
    columns/0/.style={column name=$N$},
    columns/1/.style={column name=MC, 
        int detect, sci, precision=4},
    columns/2/.style={column name=BIMC{,} $n {=} 1$,
        sci, precision=4, zerofill},
    columns/3/.style={column name=BIMC{,} $n {=} 5$, 
        sci, precision=4, zerofill},
    columns/4/.style={column name=BIMC{,} $n {=} 10$, 
        sci, precision=4, zerofill},
    columns/5/.style={column name=BIMC{,} $n {=} 25$, 
        sci, precision=4, zerofill},
    col sep=space
]{./tables/numSamplesConvAutoPGF.dat}
\end{table}

\begin{table}[H]
\small
    \centering
\caption{Lorenz, $t_f$ = 0.1s}
\begin{filecontents*}{./tables/numSamplesConvLorenzShortPGF.dat}
1.000000000000000000e+01 1.000000000000000056e-01 3.346166449643000335e-02 2.145130411378000021e-02 3.206408478027999998e-02 4.337177332035999783e-02
5.000000000000000000e+01 2.000000000000000042e-02 3.383711183744000234e-02 3.428505990332000181e-02 3.388355796033999667e-02 2.996317092389000000e-02
1.000000000000000000e+02 2.000000000000000042e-02 3.275260647282000198e-02 3.267329409925000261e-02 3.003454648106999858e-02 3.161827819877999907e-02
5.000000000000000000e+02 2.999999999999999889e-02 3.398889906409999911e-02 3.347822101549000062e-02 3.221128991206000103e-02 3.240006749477000042e-02
1.000000000000000000e+03 3.099999999999999978e-02 3.360619375407999931e-02 3.229713080596999936e-02 3.249264943231999769e-02 3.300753631922000164e-02
5.000000000000000000e+03 3.080000000000000099e-02 3.335180168700999787e-02 3.306886541137000340e-02 3.266303876118999833e-02 3.306527346954000329e-02
1.000000000000000000e+04 3.250000000000000111e-02 3.290631071421999676e-02 3.325035088257000115e-02 3.302037365207999980e-02 3.245556852670000020e-02
5.000000000000000000e+04 3.323999999999999871e-02 3.287911342430999834e-02 3.291215769237999800e-02 3.307297431548000127e-02 3.301071491254999740e-02
1.000000000000000000e+05 3.402000000000000163e-02 3.285345850758000091e-02 3.300382701374000166e-02 3.287758719791000067e-02 3.296890759508999680e-02
\end{filecontents*}
\pgfplotstabletypeset[
    every head row/.style={before row=\toprule,after row=\midrule},
    every last row/.style={after row=\bottomrule},
    skip rows between index={1}{2},
    skip rows between index={3}{4},
    skip rows between index={5}{6},
    skip rows between index={7}{8},
    columns/0/.style={column name=$N$},
    columns/1/.style={column name=MC, 
        int detect, sci, precision=4},
    columns/2/.style={column name=BIMC{,} $n {=} 1$,
        sci, precision=4, zerofill},
    columns/3/.style={column name=BIMC{,} $n {=} 5$, 
        sci, precision=4, zerofill},
    columns/4/.style={column name=BIMC{,} $n {=} 10$, 
        sci, precision=4, zerofill},
    columns/5/.style={column name=BIMC{,} $n {=} 25$, 
        sci, precision=4, zerofill},
    col sep=space
]{./tables/numSamplesConvLorenzShortPGF.dat}
\end{table}

\begin{table}[H]
    \small
    \centering
\caption{Elliptic PDE}
\begin{filecontents*}{./tables/numSamplesConvPoissonPGF.dat}
1.000000000000000000e+02 0.000000000000000000e+00 3.834970622902297822e-04 3.550171136677842299e-04 4.605656129399613385e-04 4.399509022584452444e-04
1.000000000000000000e+03 0.000000000000000000e+00 3.621953185965516987e-04 4.159972929098075013e-04 3.829997408195328040e-04 3.860834515740661820e-04
1.000000000000000000e+04 5.999999999999999474e-04 3.851700637854484402e-04 3.850259158337940497e-04 3.866622915566879370e-04 3.821293349630274080e-04
5.000000000000000000e+04 4.600000000000000139e-04 3.933986447462493809e-04 3.980585874097639131e-04 3.911901891555152801e-04 3.933028421541054538e-04
1.000000000000000000e+05 4.299999999999999894e-04 3.914116502904114005e-04 3.909747688020006728e-04 3.902058006128265996e-04 3.893673147184734508e-04
\end{filecontents*}
\pgfplotstabletypeset[
    every head row/.style={before row=\toprule,after row=\midrule},
    every last row/.style={after row=\bottomrule},
    skip rows between index={3}{4},
    columns/0/.style={column name=$N$},
    columns/1/.style={column name=MC, 
        int detect, sci, precision=4},
    columns/2/.style={column name=BIMC{,} $n {=} 1$,
        sci, precision=4, zerofill},
    columns/3/.style={column name=BIMC{,} $n {=} 5$, 
        sci, precision=4, zerofill},
    columns/4/.style={column name=BIMC{,} $n {=} 10$, 
        sci, precision=4, zerofill},
    columns/5/.style={column name=BIMC{,} $n {=} 25$, 
        sci, precision=4, zerofill},
    col sep=space
]{./tables/numSamplesConvPoissonPGF.dat}
\end{table}
\bibliographystyle{siamplain}
\bibliography{refs}